\documentclass[draftclsnofoot,onecolumn,12pt]{IEEEtran}
\usepackage{breqn}
\usepackage{amsthm,amsmath,amssymb}
\usepackage{caption}
\usepackage{graphicx}
\usepackage{stfloats}
\usepackage[font=small]{caption}
\usepackage{cite}
\usepackage{booktabs}
\usepackage[labelformat=simple]{subcaption}
\usepackage{soul}
\usepackage{tabularx}
\usepackage{textcomp}
\usepackage{mathtools}
\usepackage[ruled,vlined,linesnumbered]{algorithm2e}
\usepackage{color}
\usepackage[dvipsnames]{xcolor}
\usepackage[normalem]{ulem}
\newcommand\soutpars[1]{\let\helpcmd\sout\parhelp#1\par\relax\relax}
\long\def\parhelp#1\par#2\relax{%
  \helpcmd{#1}\ifx\relax#2\else\par\parhelp#2\relax\fi%
}

\usepackage{array}
\newcommand{\PreserveBackslash}[1]{\let\temp=\\#1\let\\=\temp}
\newcolumntype{C}[1]{>{\PreserveBackslash\centering}m{#1}}
\newcolumntype{R}[1]{>{\PreserveBackslash\raggedleft}m{#1}}
\newcolumntype{L}[1]{>{\PreserveBackslash\raggedright}m{#1}}

\newcommand{\blue}[1]{\color{black}{#1}\color{black}}

\ifCLASSINFOpdf

\else

\fi
\newtheorem{defn}{Definition}
\newtheorem{assump}{Assumption}
\newtheorem{prop}{Proposition}
\newtheorem{coro}{Corollary}
\newtheorem{remark}{Remark}

\allowdisplaybreaks

\begin{document}
\bstctlcite{IEEEexample:BSTcontrol}
\title{Intelligent Reflecting Surface Aided MIMO with Cascaded LoS Links: Channel Modelling and Full Multiplexing Region}
%
%
%

\author{Mingchen~Zhang,~\IEEEmembership{Student~Member,~IEEE,}
        Xiaojun~Yuan,~\IEEEmembership{Senior~Member,~IEEE}
\thanks{M. Zhang and X. Yuan are with the National Key Laboratory of Science and Technology on Communications, University of Electronic Science and Technology of China, Chengdu 611731, China (e-mail: zhangmingchen@std.uestc.edu.cn; xjyuan@uestc.edu.cn). This paper was presented in part at 2022 IEEE International Conference on Communications, Seoul, South Korea, May 2022.}}


\maketitle

\begin{abstract}
    This work studies the modelling and the optimization of intelligent reflecting surface (IRS) assisted multiple-input multiple-output (MIMO) systems through cascaded line-of-sight (LoS) links.
    In Part I of this work, we build up a new IRS-aided MIMO channel model, named the cascaded LoS MIMO channel. The proposed channel model consists of a transmitter (Tx) and a receiver (Rx) both equipped with uniform linear arrays, and an IRS is used to enable communications between the transmitter and the receiver through the LoS links seen by the IRS. When modeling the reflection of electromagnetic waves at the IRS, we take into account the curvature of the wavefront on different reflecting elements. Based on the established model, we study the spatial multiplexing capability of the cascaded LoS MIMO system. We introduce the notion of full multiplexing region (FMR) for the cascaded LoS MIMO channel, where the FMR is the union of Tx-IRS and IRS-Rx distance pairs that enable full multiplexing communication. Under a special passive beamforming strategy named reflective focusing, we derive an inner bound of the FMR, and provide the corresponding orientation settings of the antenna arrays that enable full multiplexing. Based on the proposed channel model and reflective focusing, the mutual information maximization problem is discussed in Part II.
\end{abstract}


\IEEEpeerreviewmaketitle

\section{Introduction}

Recently, intelligent reflecting surface (IRS), also known as (a.k.a.) reconfigurable intelligent surface (RIS), has been extensively studied in a variety of new technical challenges and new IRS-aided communication scenarios \cite{he_cascaded_2020,liu2020matrix,wuIntelligentReflectingSurface2019a,diHybridBeamformingReconfigurable2020a,yan_passive_2020,pan_uav-assisted_2021,wan_terahertz_2021}.
An IRS is made of a large number of low-cost reconfigurable elements, a.k.a. meta atoms or unit cells, that are able to control how incident electromagnetic (EM) waves are reflected. 
Unlike relays that are equipped with radio frequency (RF) chains and active amplifiers, IRS works passively after being fabricated. Due to the appealing properties of low cost and high energy efficiency, 
IRS is widely considered as a promising next-generation technology to revolutionize the classical paradigm of wireless communications.

The deployment of IRS fundamentally changes the electromagnetic propagation environment of wireless communications, and accurate channel models are of crucial importance to fully exploit the potential of an IRS-assisted wireless networks. Much research effort has been made towards the understanding of the cascaded channel created by IRS reflection.
Based on the scalar diffraction theory and the Huygens-Fresnel principle, the authors of \cite{di_renzo_analytical_2020} proposed a practical path loss model in both the \textit{near-field} and far-field of a special one-dimensional IRS.
In \cite{garcia_reconfigurable_2020}, the authors decomposed the \textit{Fresnel} zone of the IRS, and discussed the power scaling of the IRS reflected wave with the distance to the IRS numerically.
In \cite{tang2020wireless}, the authors developed a path loss model for IRS based on an empirical response function of unit cells.
The accuracy of the proposed model were practically validated by an IRS fabricated in an anechoic chamber. 
Ref. \cite{ellingson2019path} verified the findings in \cite{tang2020wireless} by unifying the path loss models of discrete unit-cell array and continuous scattering plate. 
In \cite{ozdogan_intelligent_2020}, the authors computed the scattered EM field of a planar IRS for an incident wave with specific polarization. 
Ref. \cite{bjornson_power_2020} studied the power scaling law of an IRS with a large number of unit cells.
In \cite{najafi_physics-based_2020}, the authors developed a physics-based path loss model for multi-tile IRSs, and accordingly built up an IRS-aided MIMO model and an IRS optimization framework, where general polarizations and directions of the incident/reflected waves are considered under the far-field assumption. 
In \cite{danufanePathLossReconfigurableIntelligent2021a}, the authors generalized the analyses in \cite{di_renzo_analytical_2020}, and developed a path-loss model for two-dimensional homogenizable IRS based on the vector generalization of the Green's theorem. 
\blue{In \cite{basarIndoorOutdoorPhysical2021}, the authors considered a 5G mmWave communication scenario with a random number of clusters/scatters, and provided a narrowband channel model for IRS-assisted systems for indoor and outdoor environments.}


Most existing works adopted the \textit{far-field} assumption on modelling the IRS-aided wireless channels, i.e., the size of an IRS is sufficiently small as compared to the distances from the IRS to the transmit/receive antenna arrays. In other words, the Tx and the Rx operate in the \textit{Fraunhofer} region of the IRS, where the wavefront originating from a Tx/Rx antenna can be approximated as a plane at the IRS. Such assumption may result the channel models oversimplified in many practical IRS-aided communication scenarios. 
Particularly, when the size of an IRS becomes comparable to the link distances, the Tx and the Rx operate in the Fresnel region of the IRS, where the curvature of the wavefront impinging upon the IRS cannot be ignored. For example, it was shown in \cite{dardari_holographic_2021} that when an IRS with the size between 10 cm and 1 m operates in the millimeter-wave band, the link distances between 1 m and 100 m are included almost entirely in the Fresnel region, where the plane wave approximation of the wavefront does not hold anymore. \blue{In \cite{dovelosIntelligentReflectingSurfaces2021a}, the authors proposed a near-field channel model for IRS-aided Terahertz (THz) single-input-single-output (SISO) systems. 
It shows that due to the large number of reflecting elements with respect to the wavelength, the Fresnel region of a THz IRS can be as short as several meters. } As such, it is of pressing importance to understand and characterize spherical wavefront propagation in an IRS-aided wireless communication environment.


To tackle this issue, in this paper, we build up a novel IRS-assisted channel model, named \textit{cascaded LoS MIMO channel}. In this model, both the Tx and the Rx are equipped with uniform linear arrays (ULAs), and the direct link between the Tx and the Rx is blocked by obstacles. 
The communication is enabled by an IRS that connects the LoS links between the IRS and the Tx/Rx. When modeling the reflection of electromagnetic waves at the IRS, the curvature of the wavefront on different REs are taken into account, which is distinct from many existing works that take the plane-wave assumption. Based on the proposed channel model, we investigate the spatial multiplexing capability of the cascaded LoS MIMO channel, where spatial multiplexing is a critical performance metric to evaluate the quality of a MIMO channel. We show that for a cascaded LoS MIMO channel, the maximum spatial multiplexing capability can only cover a certain area (from the IRS), named the full multiplexing region (FMR). The main contribution of this paper is summarized as follows:
\begin{itemize}
    \item We establish the cascaded LoS MIMO channel model. In this model, when modelling the reflection of EM waves at the IRS, we assume that the size of each RE is small enough relative to the link distances so that the wavefront impinging upon a single RE can still be viewed as a plane, but for the IRS as a whole, the curvature of the wavefront are taken into account. The path loss and the phase shift of each Tx-IRS-Rx link are derived based on physics and geometry. Based on that, we represent the cascaded LoS MIMO channel by an explicit expression.
    \item We define the notion of FMR for characterizing the spatial multiplexing capability of a cascaded LoS MIMO channel. A MIMO channel is referred to as able to support full multiplexing communication between the Tx and the Rx when all the eigenmodes of the channel share an equal channel gain. The FMR of a cascaded LoS MIMO channel is defined by the union of all Tx-IRS and IRS-Rx distance pairs that enable full multiplexing, with adjustable Tx/Rx orientations and the PB of the IRS. We show that even if both the Tx-IRS channel and the IRS-Rx channel can support full multiplexing, the full multiplexing capability of the overall cascaded LoS MIMO channel cannot be necessarily guaranteed.
    \item We introduce a special PB strategy called \textit{reflective focusing (RF)} which, as a concept borrowed from optics, aims to coherently superimpose the IRS-reflected waves (originating from a single Tx antenna) at a single Rx antenna.
    With RF, we derive a closed-form inner bound of the FMR, and provide the Tx/Rx orientations to achieve full multiplexing.
\end{itemize}

The remainder of this paper is organized as follows. In Section \ref{SecModel}, we present the proposed cascaded LoS MIMO channel model and the reflective focusing PB strategy. In Section \ref{SecFMR}, we introduce the full multiplexing region of the cascaded LoS MIMO channel, and derive an inner bound of the FMR based on reflective focusing. The conclusions are drawn in Section \ref{SecConclusion}. 

\textit{Notation:} We use a bold symbol lowercase letter and bold symbol capital letter to denote a vector and a matrix, respectively. The trace, conjugate, transpose, conjugate transpose, and inverse of a matrix are denoted by $\text{Tr}[\cdot]$, $(\cdot)^*$, $(\cdot)^{\textrm{T}}$, $(\cdot)^{\textrm{H}}$, and $(\cdot)^{-1}$, respectively; $|\cdot|$ denotes the modulus of a complex number or the determinant of a square matrix; $||\cdot||$ denotes the $\ell^2$ norm; ${\textrm{diag}}\left\{\boldsymbol{a}\right\}$ represents the diagonal matrices with the diagonal specified by $\boldsymbol{a}$; $\mathcal{I}_N$ denotes the index set $\{-\frac{N-1}{2},\ldots,\frac{N-1}{2}\}$, where $N$ is a positive odd number.

\section{Cascaded LoS Channel Modelling}\label{SecModel}

In this section, we first introduce the geometric model of the considered cascaded LoS MIMO system, and the response function of each reflecting element of the IRS. Based on that, we derive the attenuation and phase shift of the EM wave experienced in each Tx-IRS-Rx link, and establish an explicit channel representation of the considered system.

\subsection{IRS-Aided MIMO System}\label{SecSystem}


The considered cascaded LoS MIMO system is illustrated in Fig. \ref{SysFiga}. The direct path between the Tx and the Rx is blocked by obstacles, and an IRS is deployed to assist the communication between the Tx and the Rx. The propagation paths from the Tx to the IRS and from the IRS to the Rx are in LoS, and they are only affected by free-space attenuation. The Tx and the Rx are equipped with ULAs with $N_{\mathsf{t}}$ and $N_{\mathsf{r}}$ antennas, respectively. The inter-antenna spacings of the Tx and the Rx are $d_{\mathsf{t}}$ and $d_{\mathsf{r}}$, respectively. The origin $O$ is taken to be at the center of the IRS. The $x$- and the $y$- axes are taken to be parallel to the sides of the IRS, and the $z$-axis is taken to be perpendicular to the IRS. The lengths of Tx and Rx are given by $L_{\mathsf{t}}$ and $L_{\mathsf{r}}$, respectively.
For notational brevity, we assume that $N_{\mathsf{t}}$ and $N_{\mathsf{r}}$ are both odd numbers, and denote by $\mathcal{I}_N$ the index set $\{-\frac{N-1}{2},\ldots,\frac{N-1}{2}\}$.
Denote by $t_p$ the $p$-th antenna of the Tx, and by $r_q$ the $q$-th antenna of the Rx, where $p \in \mathcal{I}_{N_{\mathsf{t}}}$ and $q\in \mathcal{I}_{N_{\mathsf{r}}}$. $t_0$ and $r_0$ are at the centers of the Tx and the Rx, respectively.
To describe the locations of the Tx and the Rx, denote by $D_{\mathsf{t}}$ and $D_{\mathsf{r}}$ the distances from the origin to $t_0$ and $r_0$, respectively, and denote by $\omega_{\mathsf{t}}$/$\varphi_{\mathsf{t}}$ and $\omega_{\mathsf{r}}$/$\varphi_{\mathsf{r}}$ the azimuth/elevation angles of $t_0$ and $r_0$, respectively, with $\varphi_{\mathsf{t}},\varphi_{\mathsf{r}}\in [0,\frac{\pi}{2}]$, $\omega_{\mathsf{t}},\omega_{\mathsf{r}} \in [0,2\pi)$. $\omega_{\mathsf{r}}$ and $\varphi_{\mathsf{r}}$ are illustrated in Fig. \ref{SysFigb} for an example.
To describe the orientations of the Tx and the Rx, we introduce auxiliary coordinate axes and angles 
as shown in Fig. \ref{SysFigc}. Specifically, take the $z_{\mathsf{r}}$-axis in the direction from the origin to $r_0$, and the $y_{\mathsf{r}}$-axis on the $z$-$z_{\mathsf{r}}$ plane and perpendicular to the $z_{\mathsf{r}}$-axis, with the angle between the positive directions of the $z$- and $y_{\mathsf{r}}$- axes equal to $\varphi_{\mathsf{r}}+\frac{\pi}{2}$. The $x_{\mathsf{r}}$-axis is defined by its unit vector $\mathbf{n}_{x_{\mathsf{r}}} = \mathbf{n}_{y_{\mathsf{r}}}\times \mathbf{n}_{z_{\mathsf{r}}}$, where $\mathbf{n}_{y_{\mathsf{r}}}$ and $\mathbf{n}_{z_{\mathsf{r}}}$ are the unit vectors of the $y_{\mathsf{r}}$- and $z_{\mathsf{r}}$- axes, respectively, and $\times$ is the cross product operator. Define the direction from $r_{-\frac{N_{\mathsf{r}}-1}{2}}$ to $r_{\frac{N_{\mathsf{r}}-1}{2}}$ as the principal direction of the Rx. $\psi_{\mathsf{r}}\in [0,\pi]$ and $\gamma_{\mathsf{r}} \in [0,2\pi)$ are defined as the elevation and the azimuth angles of the principal direction of the Rx in the $x_{\mathsf{r}}$-$y_{\mathsf{r}}$-$z_{\mathsf{r}}$ coordinate system, respectively. Such a geometrical model with parameters $\varphi_{\mathsf{r}}$, $\omega_{\mathsf{r}}$, $\psi_{\mathsf{r}}$, $\gamma_{\mathsf{r}}$, and $D_{\mathsf{r}}$ are competent to describe any location and orientation of the Rx relative to the IRS. For the Tx-IRS side, axes $x_{\mathsf{t}}$-, $y_{\mathsf{t}}$-, $z_{\mathsf{t}}$-, and angles $\psi_{\mathsf{t}}$ and $\gamma_{\mathsf{t}}$ are defined in a similar manner as the counterparts in the IRS-Rx side, with $\psi_{\mathsf{t}}\in [0,\pi]$ and $\gamma_{\mathsf{t}} \in [0,2\pi)$.

\begin{figure}[t]
    \centering
    \begin{subfigure}[b]{0.35\textwidth}
        \centering
        \includegraphics[width=\textwidth]{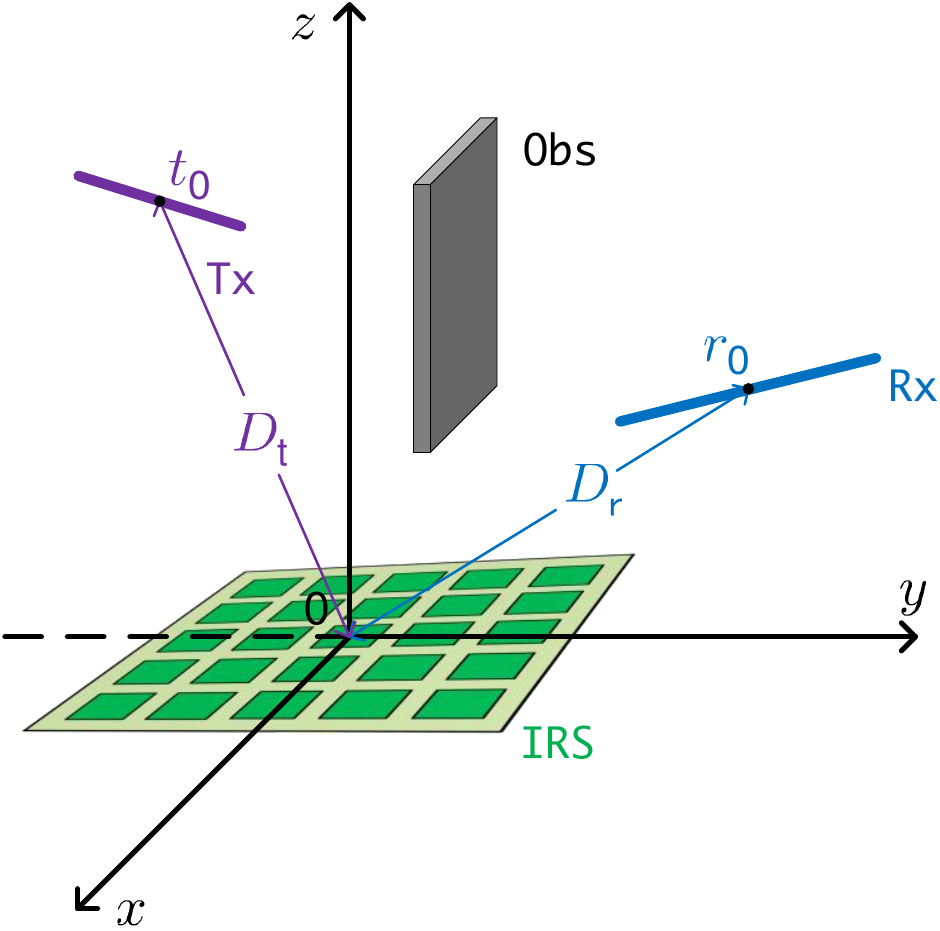}
        \caption{}
        \label{SysFiga}
    \end{subfigure}
    \hfill
    \begin{subfigure}[b]{0.35\textwidth}
        \centering
        \includegraphics[width=\textwidth]{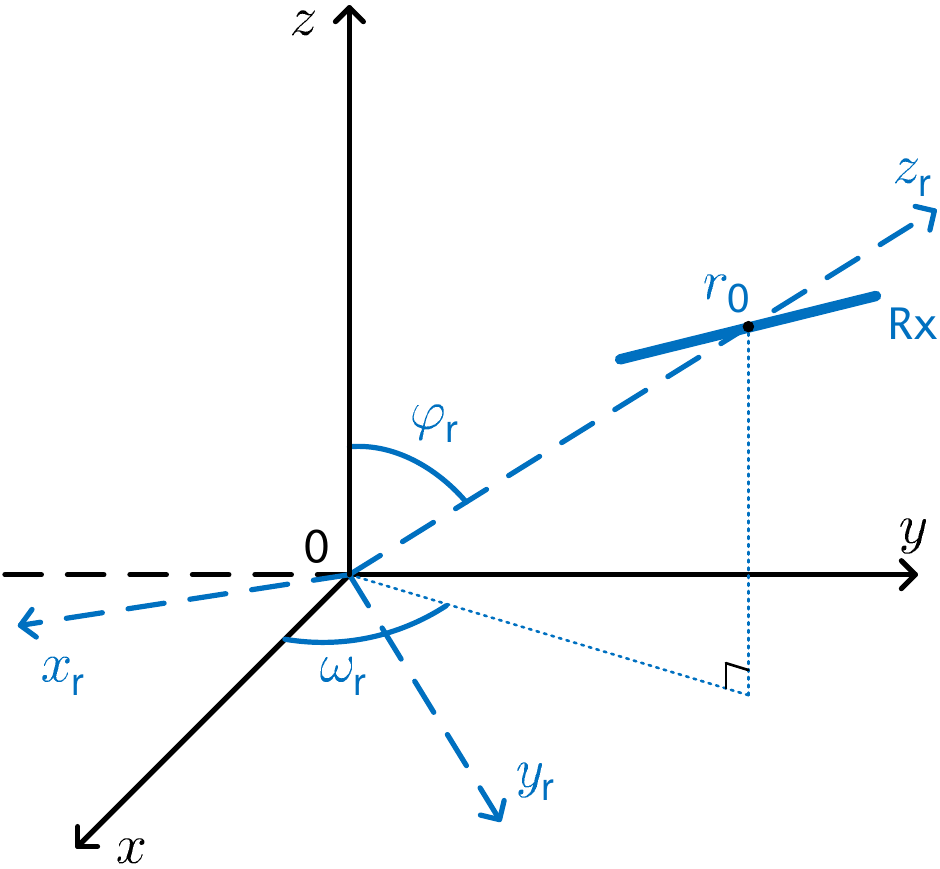}
        \caption{}
        \label{SysFigb}
    \end{subfigure}
    \hfill
    \begin{subfigure}[b]{0.27\textwidth}
        \centering
        \includegraphics[width=\textwidth]{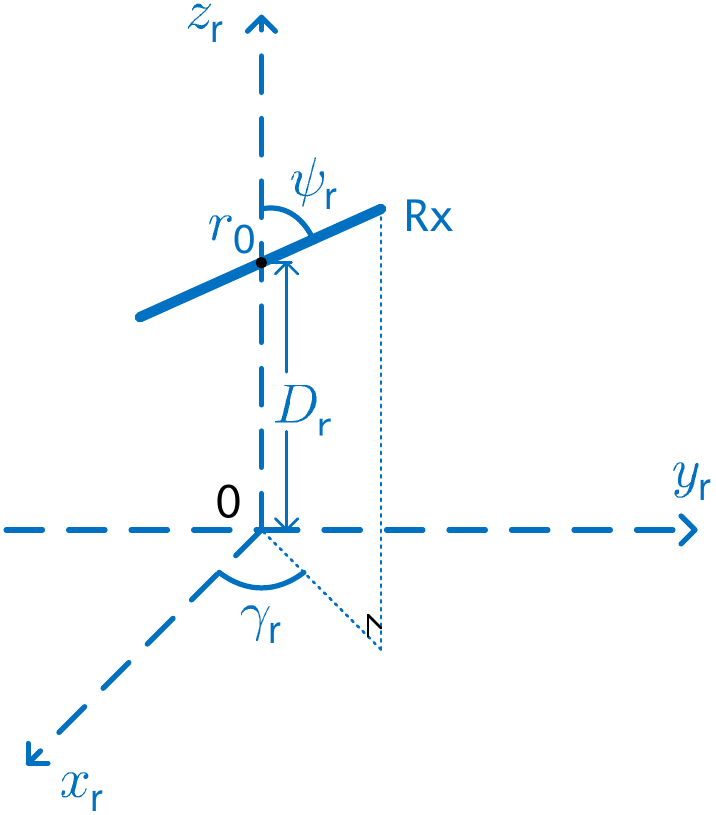}
        \caption{}
        \label{SysFigc}
    \end{subfigure}
    \caption{Geometrical model for a cascaded LoS MIMO system. (a) The LoS channel between the Tx and the Rx is blocked by obstacles. A cascaded LoS link aided by an IRS is established between the two end. (b) Parameters to describe the centers of the Tx and the Rx. (c) Parameters to describe the orientations of the Tx and the Rx.}
    \label{SysFig}
\end{figure}

The IRS consists of multiple reflecting elements (REs) evenly spaced on a grid, as illustrated in Fig. \ref{PlaneelementFig}. Each RE is composed of numerous sub-wavelength unit cells (also known as meta atoms). Each unit cell contains programmable components (such as tunable varactor diodes or switchable positive-intrinsic-negative diodes) and is able to change the properties of the reflected EM waves. In this paper, 
we consider that the unit cells on one RE are arranged densely enough so that the collection of them acts as a continuous programmable surface. The surface impedance is suitably designed to realize reflection coefficient $\Gamma=\tau e^{j\beta(x,y)}$, where ${\beta(x,y)}$ is the phase shift applied at point $(x,y)$ on the RE, and $\tau$ is the amplitude of the reflection coefficient which is assumed to be constant across the RE. This definition of RE is similar to that of \textit{continuous tile} in \cite{najafi_physics-based_2020}. The REs are spaced by $S_{\mathsf{x}}$ and $S_{\mathsf{y}}$ along the $x$- and $y$- axes, respectively, with each RE of size $L_{\mathsf{x}} \times L_{\mathsf{y}}$. Each column ($x$ direction) and each row ($y$ direction) of the IRS contain $Q_{\mathsf{x}}$ and $Q_{\mathsf{y}}$ REs, respectively.
The RE in the $k$-th row and the $l$-th column of the IRS is denoted by $m_{k,l},~ k \in \mathcal{I}_{Q_{\mathsf{x}}},~ l \in \mathcal{I}_{Q_{\mathsf{y}}}$, and $m_{0,0}$ is centered at the origin $O$\footnotemark.\footnotetext{In this paper, $Q_{\mathsf{x}}$, $Q_{\mathsf{y}}$, $N_{\mathsf{t}}$, and $N_{\mathsf{r}}$ are assumed to be odd numbers for simplifying the coordinate expressions of antennas and REs. The results of this paper can be easily extended to the cases where $Q_{\mathsf{x}}$, $Q_{\mathsf{y}}$, $N_{\mathsf{t}}$, and/or $N_{\mathsf{r}}$ are even numbers.}
The size of the entire IRS is given by $L_{\mathsf{x}}^{\mathsf{tot}} \times L_{\mathsf{x}}^{\mathsf{tot}}$.
We next characterize the cascaded LoS MIMO channel based on the geometric system model described above.



\begin{figure}[t]
    \centering
    \begin{subfigure}[b]{0.35\linewidth}
        \centering
        \includegraphics[width=\linewidth]{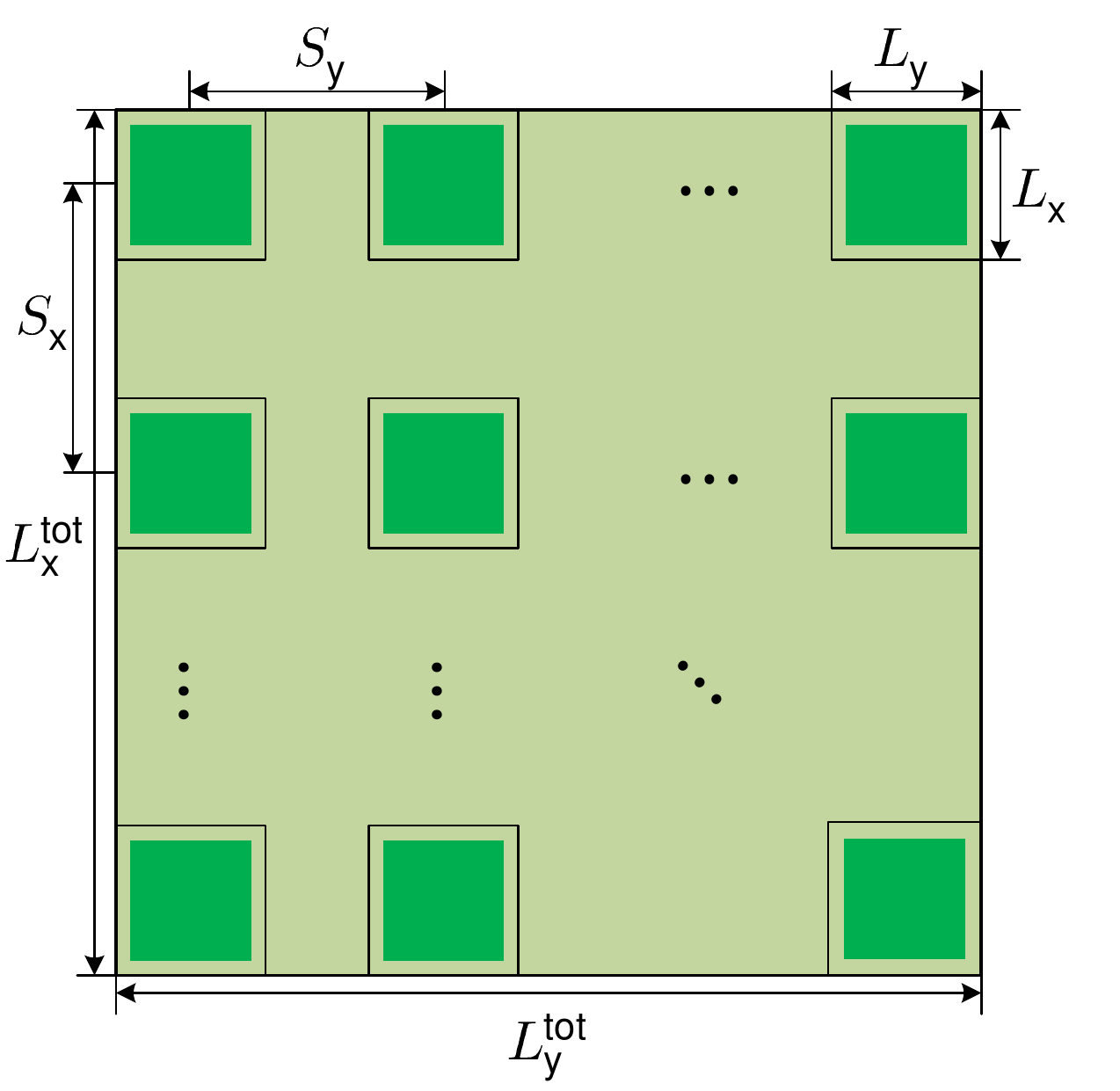}
        \caption{}
        \label{PlaneelementFig}
    \end{subfigure}
    \begin{subfigure}[b]{.35\linewidth}
        \centering
        \includegraphics[width=\linewidth]{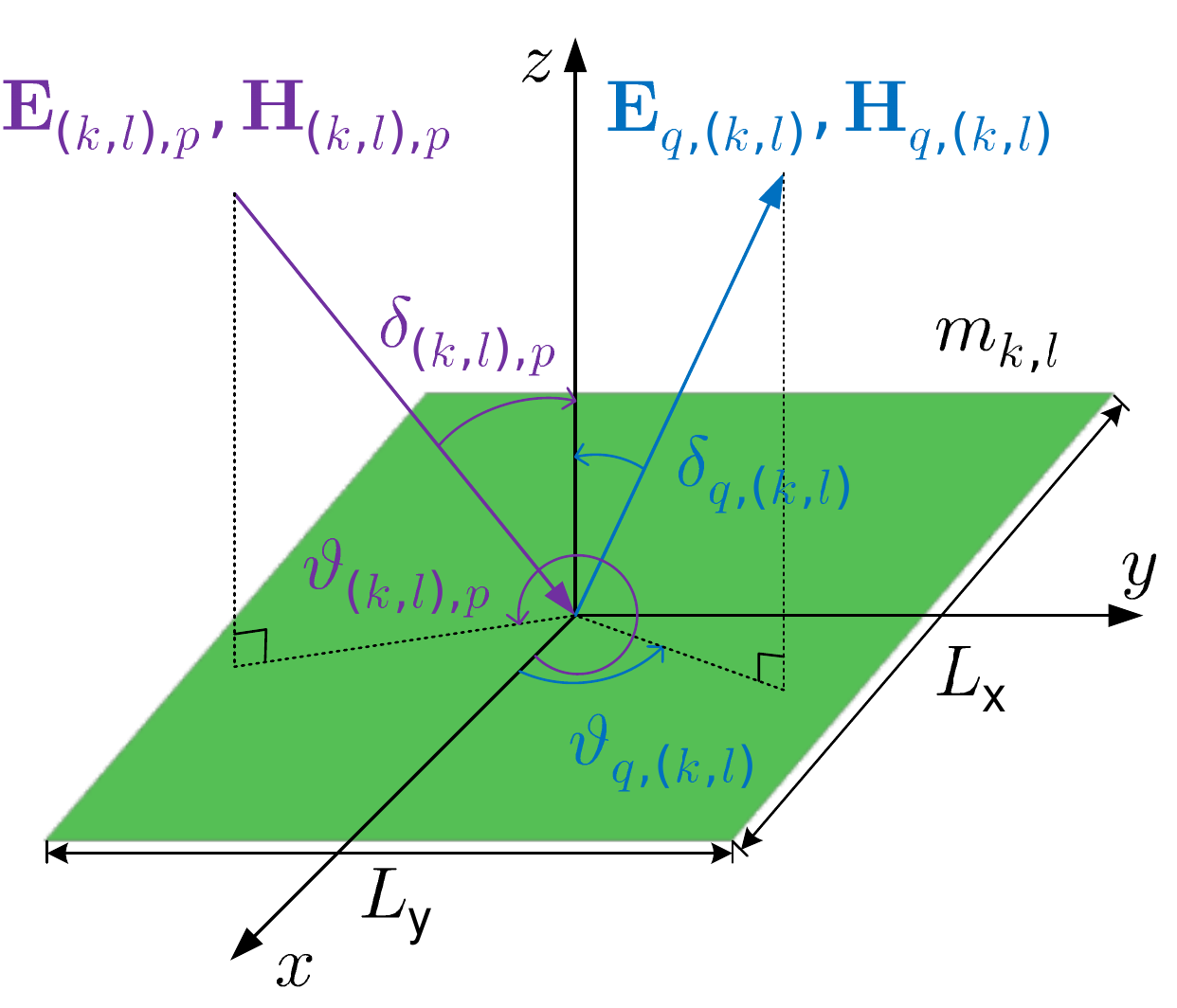}
        \caption{}
        \label{PhysicsFig}
    \end{subfigure}
    \caption{(a) Illustration of a planar IRS of size $L_{\mathsf{x}}^{\mathsf{tot}} \times L_{\mathsf{y}}^{\mathsf{tot}}$ comprised of REs of size $L_{\mathsf{x}} \times L_{\mathsf{y}}$ with spacing $S_{\mathsf{x}}$ and $S_{\mathsf{y}}$. Each RE works as a continuous surface. (b) Plane wave impinges on RE $m_{k,l}$ with incident angle $\boldsymbol{\Psi}_{(k,l),p} = (\delta_{(k,l),p},\vartheta_{(k,l),p},\varsigma_{(k,l),p})$ and is reflected with reflection angle $\boldsymbol{\Psi}_{q,(k,l)} = (\delta_{q,(k,l)},\vartheta_{q,(k,l)})$.}
\end{figure}


\subsection{RE Response Function}

In this subsection, we consider the impact of each RE imposed on its reflected impinging EM wave. To simplify the discussion, we make the following basic assumption:
\begin{assump}\label{REFarFieldAssump}
    Each RE of the IRS is in the far field of each single Tx and Rx antenna. 
    For an RE of size $L_{\mathsf{x}}\times L_{\mathsf{y}}$, the boundary of its far field and near field is defined as $B_{\mathsf{RE}}\triangleq \frac{2\left(L_{\mathsf{x}}^2+L_{\mathsf{y}}^2 \right)}{\lambda}$.\footnotemark Beyond this boundary, the maximum phase difference (MPD) of the received signal on an RE is less than $\frac{\pi}{8}$, and thus the curvature of the wavefront on the RE can be neglected. 
\end{assump}
\footnotetext{\blue{This far-field boundary is widely adopted in the literature; see, e.g., \cite[Page 111]{huang_antennas_nodate}, \cite[Page 34]{balanis_antenna_2016}, and \cite{wan_terahertz_2021}.}}
Assumption \ref{REFarFieldAssump} is made for a small RE size relative to the Tx-IRS and IRS-Rx distances. 
\blue{
As a justification of Assumption \ref{REFarFieldAssump}, in Fig. \ref{Range}, we show the MPDs of a single RE and of the entire IRS under different carrier frequencies. We consider each RE is of size $2\textrm{cm} \times 2\textrm{cm}$, and the IRS is of size $0.4\textrm{m} \times 0.4\textrm{m}$. The considered RE is at the center of the IRS, and the source equipped with a single antenna is in front of the IRS. The distance between the source and the center of the IRS is defined by $D$. The carrier frequency $f$ is chosen to $75$GHz, $140$GHz, and $338$GHz, which are typical values for E-band, D-band, and THz-band communications, respectively. We observe that the MPDs decrease with the distance $D$, and that the higher the carrier frequency, the larger the phase differences over a single RE and over the entire IRS. Denote by $B_{\mathsf{RE}}$ and $B_{\mathsf{IRS}}$ the distances that $\textrm{MPD}_{\mathsf{RE}}$ and $\textrm{MPD}_{\mathsf{IRS}}$ reduce to $\frac{\pi}{8}$, respectively.
From Fig. \ref{Range}, $(B_{\mathsf{RE}},B_{\mathsf{IRS}})$ for $f=75$GHz, 140GHz, and 338GHz, are respectively $(0.4, 160)$ (m), $(0.75,298.7)$ (m), and $(1.8,721)$ (m). 
The results show that Assumption \ref{REFarFieldAssump} is valid for the RE several meters away from the source, while for the entire IRS, it may still in the near field of the source even at hundreds of meters away. It is worth noting that although our proposed channel model is developed based on Assumption \ref{REFarFieldAssump}, the proposed model still holds for $D>B_{\mathsf{IRS}}$ and reduces to the far-field model that has been extensively studied. 
}



\begin{figure}[t]
    \centering
    \begin{subfigure}[b]{0.43\linewidth}
        \centering
        \includegraphics[width=\linewidth]{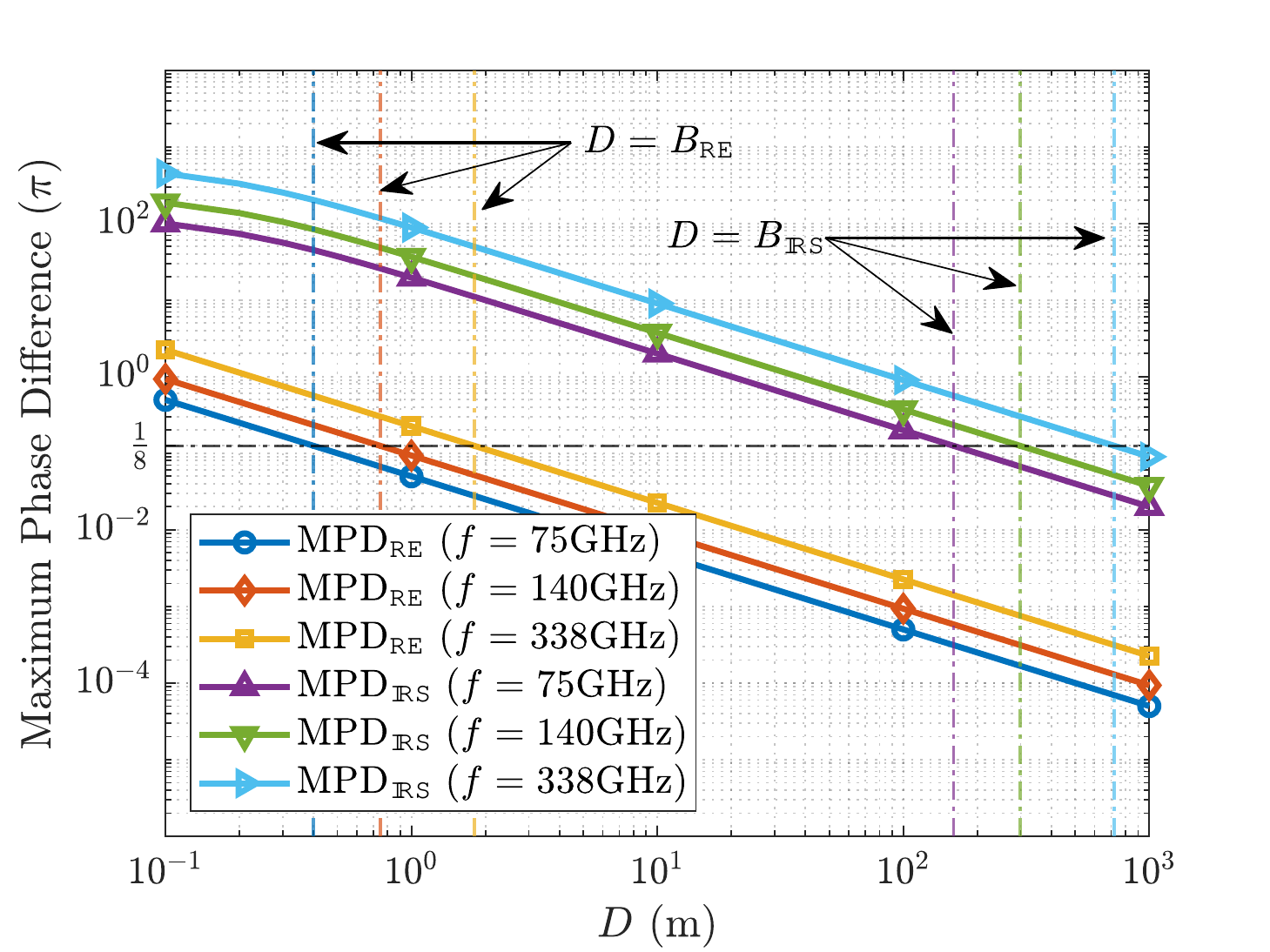}
        \caption{}
        \label{Range}
    \end{subfigure}
    \begin{subfigure}[b]{.43\linewidth}
        \centering
        \includegraphics[width=\linewidth]{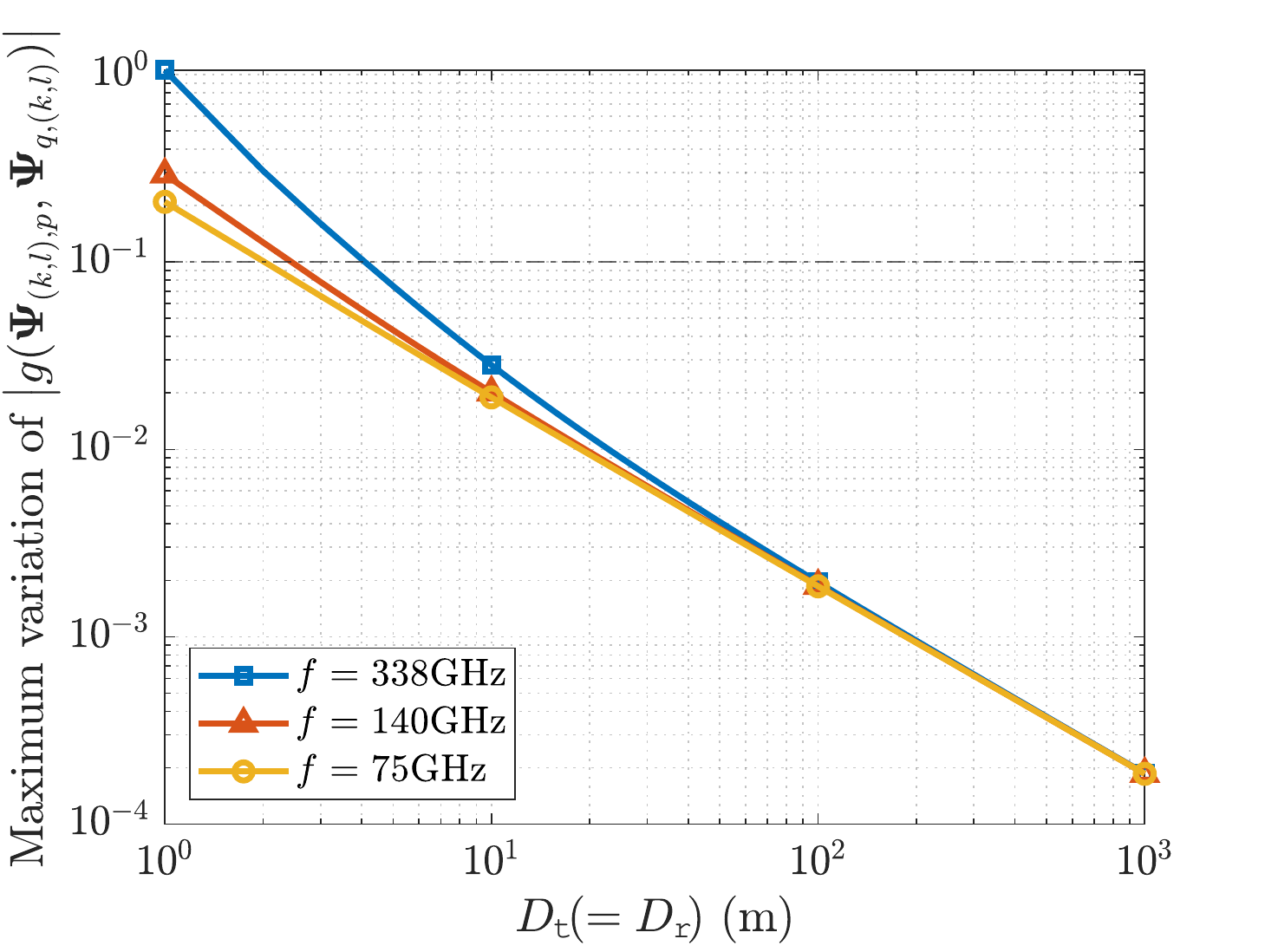}
        \caption{}
        \label{VARG}
    \end{subfigure}
    \caption{(a) The maximum phase differences (MPDs) along a single RE and the entire IRS against the source-IRS distance. (b) Maximum variation of $\left|g(\mathbf{\Psi}_{(k,l),p},\mathbf{\Psi}_{q,(k,l)})\right|$ over $|g_0|$ against the Tx-IRS (IRS-Rx) distance.}
\end{figure}


Under Assumption \ref{REFarFieldAssump}, the wave radiated by each Tx/Rx antenna can be approximately regarded as a plane wave at each single RE, and the REs work in a similar manner as continuous tiles in \cite{najafi_physics-based_2020}. Thus the continuous tile channel model developed in \cite{najafi_physics-based_2020} can be borrowed for our purpose, as detailed below. For a specific RE $m_{k,l}$, $k \in \mathcal{I}_{Q_{\mathsf{x}}},~ l \in \mathcal{I}_{Q_{\mathsf{y}}}$, when reflecting the impinging wave radiated from the Tx antenna $t_p$, to the Rx antenna $r_q$, the impact imposed on the reflected wave can be characterized by a complex response function $g(\boldsymbol{\Psi}_{(k,l),p},\boldsymbol{\Psi}_{q,(k,l)})$, $p \in \mathcal{I}_{N_{\mathsf{t}}},~ q \in \mathcal{I}_{N_{\mathsf{r}}}$, which is referred to the \textit{RE response function} of $m_{k,l}$. Specifically, $\boldsymbol{\Psi}_{(k,l),p} = (\delta_{(k,l),p},\vartheta_{(k,l),p},\varsigma_{(k,l),p})$ models the impinging wave, where $\delta_{(k,l),p}$ and $\vartheta_{(k,l),p}$ denote the elevation and azimuth angles of the direction from $t_p$ to $m_{k,l}$ in the $x$-$y$-$z$ coordinate system, respectively, and $\varsigma_{(k,l),p}$ represents the polarization of the impinging wave; $\boldsymbol{\Psi}_{q,(k,l)} = (\delta_{q,(k,l)},\vartheta_{q,(k,l)})$ models the reflection direction from $m_{k,l}$ to $r_q$, where $\delta_{q,(k,l)}$ and $\vartheta_{q,(k,l)}$ denote the elevation and azimuth angles of the direction from $m_{k,l}$ to $r_q$, in the $x$-$y$-$z$ coordinate system, respectively. The directional angles described above are illustrated in Fig. \ref{PhysicsFig}. Under Assumption \ref{REFarFieldAssump}, the RE response function of $m_{k,l}$ for $t_p$ and $r_q$ in the far field (of $m_{k,l}$) is given by \cite[Eq. 1]{najafi_physics-based_2020}
\begin{equation}\label{gLimit}
    g(\boldsymbol{\Psi}_{(k,l),p},\boldsymbol{\Psi}_{q,(k,l)}) = \lim_{d_{q,(k,l)} \to \infty} \sqrt{4\pi d_{q,(k,l)}^2} e^{\frac{j2\pi d_{q,(k,l)}}{\lambda}} \frac{E_{\mathsf{r}}(\boldsymbol{\Psi}_{q,(k,l)})}{E_{\mathsf{t}}(\boldsymbol{\Psi}_{(k,l),p})},
\end{equation}
where $E_{\mathsf{t}}(\boldsymbol{\Psi}_{(k,l),p})$ is a phasor denoting the complex amplitude of the incident electric field from the direction $\boldsymbol{\Psi}_{(k,l),p}$ on the center of $m_{k,l}$, and $E_{\mathsf{r}}(\boldsymbol{\Psi}_{q,(k,l)})$ is a phasor denoting the complex amplitude of the reflected electric field at $r_q$, and $d_{q,(k,l)}$ is the distance from $r_q$ to the center of $m_{k,l}$. In \eqref{gLimit}, besides depending on the impinging wave $\boldsymbol{\Psi}_{(k,l),p}$ and the reflection direction $\boldsymbol{\Psi}_{q,(k,l)}$, $E_{\mathsf{r}}(\boldsymbol{\Psi}_{q,(k,l)})$ is also affected by the reflection coefficients configured on $m_{k,l}$. Denote by $\Gamma_{k,l}(x,y)=\tau_{k,l} e^{j \beta_{k,l}(x,y)}$ the reflection coefficient realized at point $(x,y)$ (with the center of $m_{k,l}$ as the origin, $x \in (-L_{\mathsf{x}}/2, L_{\mathsf{x}}/2)$, $y \in (-L_{\mathsf{y}}/2, L_{\mathsf{y}}/2)$) of $m_{k,l}$, where $\beta_{k,l}(x,y)$ is referred to as the phase shift function of $m_{k,l}$, and $\tau_{k,l}$ is the amplitude which is assumed to be constant across $m_{k,l}$, $k \in \mathcal{I}_{Q_{\mathsf{x}}},~ l \in \mathcal{I}_{Q_{\mathsf{y}}}$. For simplicity, we further assume that all REs are properly designed to achieve a common reflection coefficient amplitude $\tau$, i.e., $\tau_{k,l}=\tau$, $\forall k \in \mathcal{I}_{Q_{{\mathsf{x}}}}$, $\forall l \in \mathcal{I}_{Q_{{\mathsf{y}}}}$.

\blue{
From \cite[Proposition 1]{najafi_physics-based_2020}, 
assume that the RE $m_{k,l}$ is configured to reflect an EM wave impinging from a certain direction $\tilde{\boldsymbol{\Psi}}_{\mathsf{t};k,l}=(\tilde{\delta}_{\mathsf{t};k,l},\tilde{\vartheta}_{\mathsf{t};k,l},\tilde{\varsigma}_{\mathsf{t};k,l})$ towards another direction $\tilde{\boldsymbol{\Psi}}_{\mathsf{r};k,l}=(\tilde{\delta}_{\mathsf{r};k,l},\tilde{\vartheta}_{\mathsf{r};k,l})$, by imposing the linear phase-shift function
\begin{equation}\label{BetaklSum}
    \beta_{k,l}(x,y) = \beta_{k,l}(x;\tilde{\boldsymbol{\Psi}}_{\mathsf{t};k,l}, \tilde{\boldsymbol{\Psi}}_{\mathsf{r};k,l}) + \beta_{k,l}(y;\tilde{\boldsymbol{\Psi}}_{\mathsf{t};k,l}, \tilde{\boldsymbol{\Psi}}_{\mathsf{r};k,l}) + \beta_{k,l} - \frac{\pi}{2}
\end{equation}
with 
\begin{align}\label{PhaseShift}
    \beta_{k,l}(x;\tilde{\boldsymbol{\Psi}}_{\mathsf{t};k,l}, \tilde{\boldsymbol{\Psi}}_{\mathsf{r};k,l}) = -\frac{2\pi A_{\mathsf{x}} (\tilde{\boldsymbol{\Psi}}_{\mathsf{t};k,l}, \tilde{\boldsymbol{\Psi}}_{\mathsf{r};k,l})x}{\lambda}  ~ \textrm{and} ~ \beta_{k,l}(y;\tilde{\boldsymbol{\Psi}}_{\mathsf{t};k,l}, \tilde{\boldsymbol{\Psi}}_{\mathsf{r};k,l}) = -\frac{2\pi A_{\mathsf{y}} (\tilde{\boldsymbol{\Psi}}_{\mathsf{t};k,l}, \tilde{\boldsymbol{\Psi}}_{\mathsf{r};k,l})y}{\lambda} ,
\end{align}
where
\begin{align}
    A_{\mathsf{x}}(\tilde{\boldsymbol{\Psi}}_{\mathsf{t};k,l}, \tilde{\boldsymbol{\Psi}}_{\mathsf{r};k,l}) &\triangleq \sin(\tilde{\delta}_{\mathsf{t};k,l}) \cos(\tilde{\vartheta}_{\mathsf{t};k,l}) + \sin(\tilde{\delta}_{\mathsf{r};k,l}) \cos(\tilde{\vartheta}_{\mathsf{r};k,l}) ~ \textrm{and} \\
    A_{\mathsf{y}}(\tilde{\boldsymbol{\Psi}}_{\mathsf{t};k,l}, \tilde{\boldsymbol{\Psi}}_{\mathsf{r};k,l}) &\triangleq \sin(\tilde{\delta}_{\mathsf{t};k,l}) \sin(\tilde{\vartheta}_{\mathsf{t};k,l}) + \sin(\tilde{\delta}_{\mathsf{r};k,l}) \sin(\tilde{\vartheta}_{\mathsf{r};k,l}).
\end{align}
Then, for an incident wave radiated from the transmit antenna $t_p$ and observed at the receive antenna $r_q$, the amplitude of the RE response function is given by 
\begin{align}
    \left|g(\boldsymbol{\Psi}_{(k,l),p},\boldsymbol{\Psi}_{q,(k,l)})\right| &= \frac{\sqrt{4\pi}\tau  L_{\mathsf{x}} L_{\mathsf{y}}}{\lambda} \tilde{g}(\boldsymbol{\Psi}_{(k,l),p},\boldsymbol{\Psi}_{q,(k,l)}) \notag  \\
    &\times \textrm{sinc} \left( \frac{\pi L_{\mathsf{x}}[A_{\mathsf{x}}(\boldsymbol{\Psi}_{(k,l),p},\boldsymbol{\Psi}_{q,(k,l)})-A_{\mathsf{x}}(\tilde{\boldsymbol{\Psi}}_{\mathsf{t};k,l}, \tilde{\boldsymbol{\Psi}}_{\mathsf{r};k,l})] }{\lambda} \right) \label{gAmpMax}  \\
    &\times \textrm{sinc} \left( \frac{\pi L_{\mathsf{y}}[A_{\mathsf{y}}(\boldsymbol{\Psi}_{(k,l),p},\boldsymbol{\Psi}_{q,(k,l)})-A_{\mathsf{y}}(\tilde{\boldsymbol{\Psi}}_{\mathsf{t};k,l}, \tilde{\boldsymbol{\Psi}}_{\mathsf{r};k,l})] }{\lambda} \right), \notag
\end{align}
where $\textrm{sinc}(x) = \sin(x)/x$ and
\begin{multline}
    \tilde{g}(\boldsymbol{\Psi}_{(k,l),p},\boldsymbol{\Psi}_{q,(k,l)}) = c(\boldsymbol{\Psi}_{(k,l),p}) \times \\
    \left\|\left[\begin{array}{c}
    \cos \left(\varsigma_{(k,l),p}\right) \cos \left(\delta_{q,(k,l)}\right) \sin \left(\vartheta_{q,(k,l)}\right)-\sin \left(\varsigma_{(k,l),p}\right) \cos \left(\delta_{q,(k,l)}\right) \cos \left(\vartheta_{q,(k,l)}\right) \\
    \sin \left(\varsigma_{(k,l),p}\right) \sin \left(\vartheta_{q,(k,l)}\right)+\cos \left(\varsigma_{(k,l),p}\right) \cos \left(\vartheta_{q,(k,l)}\right)
    \end{array}\right]\right\|_{2},
\end{multline}
\begin{equation}
    c(\boldsymbol{\Psi}_{(k,l),p}) = \frac{A_{\mathsf{z}}\left(\boldsymbol{\Psi}_{\mathsf{t}}\right)}{\sqrt{A_{\mathsf{x,y}}^2 \left(\boldsymbol{\Psi}_{\mathsf{t}}\right) + A_{\mathsf{z}}^2 \left(\boldsymbol{\Psi}_{\mathsf{t}}\right)}}
\end{equation}
with $A_{\mathsf{z}}\left(\boldsymbol{\Psi}_{\mathsf{t}}\right) \triangleq \cos(\delta_{(k,l),p})$ and $A_{\mathsf{x,y}} \left(\boldsymbol{\Psi}_{\mathsf{t}}\right) \triangleq  \cos(\varsigma_{(k,l),p})\sin(\delta_{(k,l),p})\cos(\vartheta_{(k,l),p}) + \sin(\varsigma_{(k,l),p}) \allowbreak \sin(\delta_{(k,l),p})  \allowbreak \sin(\vartheta_{(k,l),p})$, $k \in \mathcal{I}_{Q_{\mathsf{x}}},~ l \in \mathcal{I}_{Q_{\mathsf{y}}}$, $p \in \mathcal{I}_{N_{\mathsf{t}}},~ q \in \mathcal{I}_{N_{\mathsf{r}}}$. The phase of the RE response function is given by 
\begin{equation}\label{REPHASE}
    \angle g(\boldsymbol{\Psi}_{(k,l),p},\boldsymbol{\Psi}_{q,(k,l)}) = \beta_{k,l},~\forall p \in \mathcal{I}_{N_{\mathsf{t}}},~ \forall q \in \mathcal{I}_{N_{\mathsf{r}}}.
\end{equation}

In the remainder of this paper, we always assume that all the REs apply the phase shift functions of the form \eqref{BetaklSum} with $\tilde{\boldsymbol{\Psi}}_{\mathsf{t};k,l} = \boldsymbol{\Psi}_{(k,l),0}$ and $\tilde{\boldsymbol{\Psi}}_{\mathsf{r};k,l} = \boldsymbol{\Psi}_{0,(k,l)}$, $\forall k \in \mathcal{I}_{Q_{{\mathsf{x}}}}$, $\forall l \in \mathcal{I}_{Q_{{\mathsf{y}}}}$.
This enables all the REs to achieve the maximum response amplitude for the wave from the transmit antenna $t_0$ and reflected to the receive antenna $r_0$, with $\left|g(\boldsymbol{\Psi}_{(k,l),0},\boldsymbol{\Psi}_{0,(k,l)})\right| = \frac{\sqrt{4\pi}\tau  L_{\mathsf{x}} L_{\mathsf{y}}}{\lambda} \tilde{g}(\boldsymbol{\Psi}_{(k,l),0},\boldsymbol{\Psi}_{0,(k,l)})$, $k \in \mathcal{I}_{Q_{\mathsf{x}}},~ l \in \mathcal{I}_{Q_{\mathsf{y}}}$. 
Denote by $g_0$ the response of $m_{0,0}$ to the impinging wave from $t_0$ to $r_0$, i.e., $|g_0| = \frac{\sqrt{4\pi}\tau L_{\mathsf{x}} L_{\mathsf{y}}}{\lambda}\tilde{g}(\boldsymbol{\Psi}_{(0,0),0},\boldsymbol{\Psi}_{0,(0,0)})= \frac{\sqrt{4\pi}\tau L_{\mathsf{x}} L_{\mathsf{y}}}{\lambda}\tilde{g}((\varphi_{\mathsf{t}},\omega_{\mathsf{t}},\varsigma_{(0,0),0}),(\varphi_{\mathsf{r}},\psi_{\mathsf{r}}))$, $\angle g_0 = \beta_{0,0}$.
}

To further simplify the channel model, we make the following assumption:
\begin{assump}\label{SizeAssump}
    The distance between the centers of the Tx and the IRS is much greater than the sizes of the Tx and the IRS, and the distance between the centers of the Rx and the IRS is much greater than the sizes of the Rx and the IRS, i.e., $D_{\mathsf{t}} \gg L_{\mathsf{t}}, L_{\mathsf{x}}^{\mathsf{tot}}, L_{\mathsf{y}}^{\mathsf{tot}} $, and $D_{\mathsf{r}} \gg L_{\mathsf{r}}, L_{\mathsf{x}}^{\mathsf{tot}}, L_{\mathsf{y}}^{\mathsf{tot}}$.
\end{assump}
Under Assumption \ref{SizeAssump}, in all the $t_p$-$m_{k,l}$-$r_q$ links, the amplitudes of the RE response functions are approximately the same as $|g_0|$, i.e.,
\begin{equation}\label{CommonREAmp}
    |g(\boldsymbol{\Psi}_{(k,l),p},\boldsymbol{\Psi}_{q,(k,l)})|=|g_0|,~ \forall k \in \mathcal{I}_{Q_{\mathsf{x}}},~ \forall l \in \mathcal{I}_{Q_{\mathsf{y}}},~ \forall p \in \mathcal{I}_{N_{\mathsf{t}}},~ \forall q \in \mathcal{I}_{N_{\mathsf{r}}}.
\end{equation}
This approximation is based on the fact the amplitude of $g(\boldsymbol{\Psi}_{(k,l),p},\boldsymbol{\Psi}_{q,(k,l)})$ is insensitive to the variations in $\boldsymbol{\Psi}_{(k,l),p}$ and $\boldsymbol{\Psi}_{q,(k,l)}$, i.e., under Assumption \ref{SizeAssump}, the differences of $\boldsymbol{\Psi}_{(k,l),p}$ are small, and so are the differences of $\boldsymbol{\Psi}_{q,(k,l)}$, $k \in \mathcal{I}_{Q_{\mathsf{x}}},~ l \in \mathcal{I}_{Q_{\mathsf{y}}}$, $p \in \mathcal{I}_{N_{\mathsf{t}}},~ q \in \mathcal{I}_{N_{\mathsf{r}}}$.

\blue{
As a justification, denote by $\xi_{g} \triangleq \frac{\left\|\left|g(\mathbf{\Psi}_{(k,l),p},\mathbf{\Psi}_{q,(k,l)})\right|-|g_0| \right\|_{\max}}{|g_0|}$ as the maximum variation of $\left|g(\mathbf{\Psi}_{(k,l),p},\mathbf{\Psi}_{q,(k,l)})\right|$ (over $|g_0|$), where the maximization is taken over all the transmit/receive antenna and RE indices, $k \in \mathcal{I}_{Q_{\mathsf{x}}},~ l \in \mathcal{I}_{Q_{\mathsf{y}}}$, $p \in \mathcal{I}_{N_{\mathsf{t}}},~ q \in \mathcal{I}_{N_{\mathsf{r}}}$.
Fig. \ref{VARG} illustrates the change of $\xi_g$ over the link distances. The Tx/Rx directions and orientations are fixed to $\omega_{\mathsf{t}}=\frac{3\pi}{2}$, $\varphi_{\mathsf{t}}=\frac{\pi}{4}$, $\gamma_{\mathsf{t}}=0$, $\psi_{\mathsf{t}}=\frac{\pi}{2}$, $\omega_{\mathsf{r}}=\frac{\pi}{2}$, $\varphi_{\mathsf{r}}=\frac{\pi}{6}$, $\gamma_{\mathsf{r}}=0$, $\psi_{\mathsf{r}}=\frac{\pi}{2}$. $N_{\mathsf{t}}=N_{\mathsf{r}}=5$. Each RE is of size $2\textrm{cm} \times 2\textrm{cm}$. The IRS contains $15\times15=225$ REs and is of size $0.3\textrm{m} \times 0.3\textrm{m}$. For simplicity, we consider the case that the Tx-IRS and IRS-Rx distances are equal, i.e., $D_{\mathsf{t}}=D_{\mathsf{r}}$. For different $t_p$-$m_{k,l}$-$r_q$ links, the differences of incident and reflection angles at different REs are taken into account. From Fig. \ref{VARG}, we observe that for the carrier frequencies $f=75$GHz, $f=140$GHz, and $f=338$GHz, $\xi_g$ drops below $10^{-1}$ at $D_{\mathsf{t}}(=D_{\mathsf{r}})=2.1$m, $D_{\mathsf{t}}(=D_{\mathsf{r}})=2.5$m, and $D_{\mathsf{t}}(=D_{\mathsf{r}})=4.2$m, respectively. These results show that the approximation for $\left|g(\mathbf{\Psi}_{(k,l),p},\mathbf{\Psi}_{q,(k,l)})\right|$ to $|g_0|$ is valid for Tx and Rx several meters away from the IRS, $\forall k \in \mathcal{I}_{Q_{\mathsf{x}}},~ \forall l \in \mathcal{I}_{Q_{\mathsf{y}}},~ \forall p \in \mathcal{I}_{N_{\mathsf{t}}},~ \forall q \in \mathcal{I}_{N_{\mathsf{r}}}$.
}


\begin{remark}
    Assumption \ref{SizeAssump} does not necessarily mean that the IRS is in the far field of the Tx and the Rx. For an IRS of the size $L_{\mathsf{x}}^{\mathsf{tot}} \times L_{\mathsf{y}}^{\mathsf{tot}}$, the boundary of its far field and near field is given by $B_{\mathsf{IRS}} = 2 \frac{\left(L_{\mathsf{x}}^{\mathsf{tot}}\right)^2+\left(L_{\mathsf{y}}^{\mathsf{tot}}\right)^2}{\lambda}$ \cite[Eq. 4.9]{huang_antennas_nodate}. For $L_{\mathsf{x}}^{\mathsf{tot}}=L_{\mathsf{y}}^{\mathsf{tot}}=0.3\textup{m}$ and $f=140\textup{GHz}$ (with $\lambda=0.0021\textup{m}$), $B_{\rm{IRS}}=168.1\textup{m}$. However, from Fig. \refeq{VARG}, we see that $D_{\mathsf{t}}=D_{\mathsf{r}}=2.5\textup{m}$ is sufficiently large to ensure the approximation in \eqref{CommonREAmp}.
\end{remark}


\subsection{Link Distance}
The path loss and phase shift of the EM wave experienced in each cascaded LoS link are related to the link distance. Thus we need to describe the link distance first. Denote by $d_{(k,l),p}$ the distance from $t_p$ to $m_{k,l}$, and by $d_{q,(k,l)}$ the distance from $m_{k,l}$ to $r_q$, $k \in \mathcal{I}_{Q_{\mathsf{x}}},~ l \in \mathcal{I}_{Q_{\mathsf{y}}}$, $p \in \mathcal{I}_{N_{\mathsf{t}}},~ q \in \mathcal{I}_{N_{\mathsf{r}}}$.
Based on the axes and the parameters defined in Section \ref{SecSystem}, the coordinates of the Tx antenna $t_p$, the Rx antenna $r_q$, and the center of the RE $m_{k,l}$ can be written into explicit expressions. $d_{(k,l),p}$ and $d_{q,(k,l)}$ can be analytically derived by calculating the Euclidean distances between $m_{k,l}$ and $t_p$ and between $m_{k,l}$ and $r_q$, respectively. Define $\mathbf{a}_{{\mathsf{t}},p}$ as the vector from the origin to $t_p$, $\mathbf{a}_{{\mathsf{r}},q}$ as the vector from the origin to $r_q$, and $\mathbf{a}_{k,l}$ as the vector from the origin to the center of $m_{k,l}$. From Fig. \ref{SysFig}, these vectors are given by
\begin{subequations}
    \begin{align}
        \mathbf{a}_{{\mathsf{t}},p} &= (pd_{\mathsf{t}}\sin \psi_{\mathsf{t}} \cos\gamma_{\mathsf{t}})\mathbf{n}_{x_{\mathsf{t}}} + (pd_{\mathsf{t}}\sin \psi_{\mathsf{t}} \sin\gamma_{\mathsf{t}})\mathbf{n}_{y_{\mathsf{t}}} + (D_{\mathsf{t}}+pd_{\mathsf{t}}\cos \psi_{\mathsf{t}})\mathbf{n}_{z_{\mathsf{t}}},~ p \in \mathcal{I}_{p} \label{ant} \\
        \mathbf{a}_{{\mathsf{r}},q} &= (qd_{\mathsf{r}}\sin \psi_{\mathsf{r}} \cos\gamma_{\mathsf{r}})\mathbf{n}_{x_{\mathsf{r}}} + (qd_{\mathsf{r}}\sin \psi_{\mathsf{r}} \sin\gamma_{\mathsf{r}})\mathbf{n}_{y_{\mathsf{r}}} + (D_{\mathsf{r}}+qd_{\mathsf{r}}\cos \psi_{\mathsf{r}})
        \mathbf{n}_{z_{\mathsf{r}}},~ q \in \mathcal{I}_{q} \label{anr} \\
        \mathbf{a}_{k,l} & = k S_{\mathsf{x}} \mathbf{n}_x + l S_{\mathsf{y}} \mathbf{n}_y \\
        & = (k S_{\mathsf{x}} \sin \omega_{\mathsf{t}}-lS_{\mathsf{y}}\cos \omega_{\mathsf{t}})\mathbf{n}_{x_{\mathsf{t}}} + (kS_{\mathsf{x}}\cos \varphi_{\mathsf{t}}\cos \omega_{\mathsf{t}} + lS_{\mathsf{y}}\cos \varphi_{\mathsf{t}}\sin \omega_{\mathsf{t}})\mathbf{n}_{y_{\mathsf{t}}} \notag \\
        & \quad +  (kS_{\mathsf{x}}\sin \varphi_{\mathsf{t}} \cos \omega_{\mathsf{t}} + lS_{\mathsf{y}}\sin \varphi_{\mathsf{t}} \sin \omega_{\mathsf{t}})\mathbf{n}_{z_{\mathsf{t}}},~ k\in \mathcal{I}_{Q_{\mathsf{x}}},~ l \in \mathcal{I}_{Q_{\mathsf{y}}} \label{aijt} \\
        & = (k S_{\mathsf{x}} \sin \omega_{\mathsf{r}}-lS_{\mathsf{y}}\cos \omega_{\mathsf{r}})\mathbf{n}_{x_{\mathsf{r}}} + (kS_{\mathsf{x}}\cos \varphi_{\mathsf{r}}\cos \omega_{\mathsf{r}} + lS_{\mathsf{y}}\cos \varphi_{\mathsf{r}}\sin \omega_{\mathsf{r}})\mathbf{n}_{y_{\mathsf{r}}} \notag \\
        & \quad +  (kS_{\mathsf{x}}\sin \varphi_{\mathsf{r}} \cos \omega_{\mathsf{r}} + lS_{\mathsf{y}}\sin \varphi_{\mathsf{r}} \sin \omega_{\mathsf{r}})\mathbf{n}_{z_{\mathsf{r}}},~ k\in \mathcal{I}_{Q_{\mathsf{x}}},~ l \in \mathcal{I}_{Q_{\mathsf{y}}}, \label{aijr}
    \end{align}
\end{subequations}
where $\mathbf{n}_x$, $\mathbf{n}_y$, $\mathbf{n}_z$, $\mathbf{n}_{x_{\mathsf{t}}}$, $\mathbf{n}_{y_{\mathsf{t}}}$, $\mathbf{n}_{z_{\mathsf{t}}}$, $\mathbf{n}_{x_{\mathsf{r}}}$, $\mathbf{n}_{y_{\mathsf{r}}}$, and $\mathbf{n}_{z_{\mathsf{r}}}$ are the unit vectors of the $x$, $y$, $z$, $x_{\mathsf{t}}$, $y_{\mathsf{t}}$, $z_{\mathsf{t}}$, $x_{\mathsf{r}}$, $y_{\mathsf{r}}$, and $z_{\mathsf{r}}$ axes, respectively. Based on \eqref{ant} and \eqref{aijt}, $d_{(k,l),p}$ is given by
\begin{subequations}\label{dklpORI}
    \begin{align}
        d_{(k,l),p} & = ||\mathbf{a}_{{\mathsf{t}},p} - \mathbf{a}_{k,l}|| \\
        &= \left[(pd_{\mathsf{t}}\sin \psi_{\mathsf{t}} \cos\gamma_{\mathsf{t}}-k S_{\mathsf{x}} \sin \omega_{\mathsf{t}}+lS_{\mathsf{y}}\cos \omega_{\mathsf{t}})^2 \right. \notag \\
        & \quad \left. + (pd_{\mathsf{t}}\sin \psi_{\mathsf{t}} \sin\gamma_{\mathsf{t}}-kS_{\mathsf{x}}\cos \varphi_{\mathsf{t}}\cos \omega_{\mathsf{t}} - lS_{\mathsf{y}}\cos \varphi_{\mathsf{t}}\sin \omega_{\mathsf{t}})^2 \right. \notag \\
        & \quad \left. + (D_{\mathsf{t}} +pd_{\mathsf{t}}\cos\psi_{\mathsf{t}}-kS_{\mathsf{x}}\sin \varphi_{\mathsf{t}} \cos \omega_{\mathsf{t}} - lS_{\mathsf{y}}\sin \varphi_{\mathsf{t}} \sin \omega_{\mathsf{t}})^2\right]^{\frac{1}{2}} \label{Exactklp} \\
        &\approx \frac{(pd_{\mathsf{t}}\sin \psi_{\mathsf{t}} \cos\gamma_{\mathsf{t}}-k S_{\mathsf{x}} \sin \omega_{\mathsf{t}}+lS_{\mathsf{y}}\cos \omega_{\mathsf{t}})^2}{2(D_{\mathsf{t}} +pd_{\mathsf{t}}\cos\psi_{\mathsf{t}}-kS_{\mathsf{x}}\sin \varphi_{\mathsf{t}} \cos \omega_{\mathsf{t}} - lS_{\mathsf{y}}\sin \varphi_{\mathsf{t}} \sin \omega_{\mathsf{t}})} \notag \\
        & \quad + \frac{(pd_{\mathsf{t}}\sin \psi_{\mathsf{t}} \sin\gamma_{\mathsf{t}}-kS_{\mathsf{x}}\cos \varphi_{\mathsf{t}}\cos \omega_{\mathsf{t}} - lS_{\mathsf{y}}\cos \varphi_{\mathsf{t}}\sin \omega_{\mathsf{t}})^2}{2(D_{\mathsf{t}} +pd_{\mathsf{t}}\cos\psi_{\mathsf{t}}-kS_{\mathsf{x}}\sin \varphi_{\mathsf{t}} \cos \omega_{\mathsf{t}} - lS_{\mathsf{y}}\sin \varphi_{\mathsf{t}} \sin \omega_{\mathsf{t}})} \notag \\
        & \quad + D_{\mathsf{t}} +pd_{\mathsf{t}}\cos\psi_{\mathsf{t}}-kS_{\mathsf{x}}\sin \varphi_{\mathsf{t}} \cos \omega_{\mathsf{t}} - lS_{\mathsf{y}}\sin \varphi_{\mathsf{t}} \sin \omega_{\mathsf{t}}, \label{Approxklpmid} \\
        &\approx \frac{(pd_{\mathsf{t}}\sin \psi_{\mathsf{t}} \cos\gamma_{\mathsf{t}}-k S_{\mathsf{x}} \sin \omega_{\mathsf{t}}+lS_{\mathsf{y}}\cos \omega_{\mathsf{t}})^2}{2D_{\mathsf{t}}} \notag \\
        & \quad + \frac{(pd_{\mathsf{t}}\sin \psi_{\mathsf{t}} \sin\gamma_{\mathsf{t}}-kS_{\mathsf{x}}\cos \varphi_{\mathsf{t}}\cos \omega_{\mathsf{t}} - lS_{\mathsf{y}}\cos \varphi_{\mathsf{t}}\sin \omega_{\mathsf{t}})^2}{2D_{\mathsf{t}}} \notag \\
        & \quad + D_{\mathsf{t}} +pd_{\mathsf{t}}\cos\psi_{\mathsf{t}}-kS_{\mathsf{x}}\sin \varphi_{\mathsf{t}} \cos \omega_{\mathsf{t}} - lS_{\mathsf{y}}\sin \varphi_{\mathsf{t}} \sin \omega_{\mathsf{t}}, \label{Approxklp}
    \end{align}
\end{subequations}
where \eqref{Exactklp} gives the exact path length, and \eqref{Approxklpmid} is obtained by a Taylor series expansion. The two approximations in \eqref{Approxklpmid} and \eqref{Approxklp} are valid when $D_{\mathsf{t}} \gg L_{\mathsf{t}}, L_{\mathsf{x}}^{\mathsf{tot}}, L_{\mathsf{y}}^{\mathsf{tot}} $, which is guaranteed by Assumption \ref{SizeAssump}.\footnotemark
\footnotetext{Similar approximations have been adopted in modelling conventional MIMO channels; see, e.g., \cite{bohagen_construction_2005,bohagen_design_2007,bohagen_optimal_2007,bohagen_spherical_2009,wang_tens_2014}.} 
Similarly, based on \eqref{anr} and \eqref{aijr}, we obtain
\begin{subequations}\label{dqklORI}
    \begin{align}
        d_{q,(k,l)} & = ||\mathbf{a}_{{\mathsf{r}},q} - \mathbf{a}_{k,l}|| \\
        &\approx \frac{(qd_{\mathsf{r}}\sin \psi_{\mathsf{r}} \cos\gamma_{\mathsf{r}}-k S_{\mathsf{x}} \sin \omega_{\mathsf{r}}+lS_{\mathsf{y}}\cos \omega_{\mathsf{r}})^2}{2D_{\mathsf{r}}} \notag \\
        & \quad + \frac{(qd_{\mathsf{r}}\sin \psi_{\mathsf{r}} \sin\gamma_{\mathsf{r}}-kS_{\mathsf{x}}\cos \varphi_{\mathsf{r}}\cos \omega_{\mathsf{r}} - lS_{\mathsf{y}}\cos \varphi_{\mathsf{r}}\sin \omega_{\mathsf{r}})^2}{2D_{\mathsf{r}}} \notag \\
        & \quad + D_{\mathsf{r}} +qd_{\mathsf{r}}\cos\psi_{\mathsf{r}}-kS_{\mathsf{x}}\sin \varphi_{\mathsf{r}} \cos \omega_{\mathsf{r}} - lS_{\mathsf{y}}\sin \varphi_{\mathsf{r}} \sin \omega_{\mathsf{r}}, \label{Approxqkl}
    \end{align}
\end{subequations}
$k \in \mathcal{I}_{Q_{\mathsf{x}}},~ l \in \mathcal{I}_{Q_{\mathsf{y}}},~ q \in \mathcal{I}_{N_{\mathsf{r}}}$.

\subsection{Path Loss and Phase Shift Model}
With the link distances given in \eqref{dklpORI} and \eqref{dqklORI}, we are now ready to describe the path loss and path phase shift model.
Combining the discussions in the preceding subsections and based on \cite[Lemma 1]{najafi_physics-based_2020}, the total path loss of the $t_p$-$m_{k,l}$-$r_q$ link is given by
\begin{equation}\label{PL}
    \textrm{PL}_{q,(k,l),p} = \frac{4\pi \left|g(\boldsymbol{\Psi}_{(k,l),p},\boldsymbol{\Psi}_{q,(k,l)})\right|^2}{\lambda^2} \textrm{PL}_{(k,l),p}\textrm{PL}_{q,(k,l)},
\end{equation}
$k \in \mathcal{I}_{Q_{\mathsf{x}}},~ l \in \mathcal{I}_{Q_{\mathsf{y}}},~ p \in \mathcal{I}_{N_{\mathsf{t}}},~ q \in \mathcal{I}_{N_{\mathsf{r}}}$.
\blue{In \eqref{PL}, $\textrm{PL}_{(k,l),p} \triangleq \left(\frac{\lambda}{4\pi d_{(k,l),p}}\right)^2 e^{-\kappa_{\textrm{abs}}\left( f \right) d_{(k,l),p} }$ and $\textrm{PL}_{q,(k,l)} \triangleq \left(\frac{\lambda}{4\pi d_{q,(k,l)}}\right)^2 e^{-\kappa_{\textrm{abs}}\left( f \right) d_{q,(k,l)} }$ stand for the path losses of the LoS $t_p$-$m_{k,l}$ and $m_{k,l}$-$r_q$ channels, respectively, where the fractional terms and the exponential terms represent the free space attenuations and the molecular absorption losses, respectively.
$d_{(k,l),p}$ and $d_{q,(k,l)}$ are the distances from $m_{k,l}$ to $t_p$ and $r_q$, respectively. $\kappa_{\textrm{abs}}(f)$ is the molecular absorption coefficient at the carrier frequency $f$ and can be calculated based on the high resolution transmission molecular absorption database (HITRAN) \cite{jornet_channel_2011,gordon_hitran2020_2022}. } Under Assumption \ref{SizeAssump}, all the LoS path loss $\textrm{PL}_{(k,l),p}$ are approximately the same as $\textrm{PL}_{(0,0),0}$, and all the LoS path loss $\textrm{PL}_{q,(k,l)}$ are approximately the same as $\textrm{PL}_{0,(0,0)}$, i.e., 
\begin{equation}\label{LoSPLApprox}
    \textrm{PL}_{(k,l),p} = \textrm{PL}_{(0,0),0},~ \textrm{PL}_{q,(k,l)} = \textrm{PL}_{0,(0,0)},
\end{equation}
$\forall k \in \mathcal{I}_{Q_{\mathsf{x}}},~ \forall l \in \mathcal{I}_{Q_{\mathsf{y}}},~ \forall p \in \mathcal{I}_{N_{\mathsf{t}}},~ \forall q \in \mathcal{I}_{N_{\mathsf{r}}}$. Similar approximations for the LoS path loss have been adopted when the size of the antenna array is much smaller relative to the link distance; see, e.g., \cite{bohagen_construction_2005,bohagen_design_2007,bohagen_optimal_2007,bohagen_spherical_2009,wang_tens_2014}. 

\blue{
From \eqref{CommonREAmp} and \eqref{LoSPLApprox}, the total path losses of all the Tx-IRS-Rx links reduce to a constant over the indices $k$, $l$, $p$, $q$, i.e.,
\begin{equation}\label{CommonPL}
    \textrm{PL}_{q,(k,l),p} = \eta_0^2,
\end{equation}
$\forall k \in \mathcal{I}_{Q_{\mathsf{x}}},~ \forall l \in \mathcal{I}_{Q_{\mathsf{y}}},~ \forall p \in \mathcal{I}_{N_{\mathsf{t}}},~ \forall q \in \mathcal{I}_{N_{\mathsf{r}}}$, where
\begin{subequations}
    \begin{align}
        \eta_0 &\triangleq \sqrt{\frac{4\pi \left|g_0 \right|^2}{\lambda^2}\textrm{PL}_{(0,0),0}\textrm{PL}_{0,(0,0)}} \\
        &=\frac{\tau L_{\mathsf{x}} L_{\mathsf{y}}}{4\pi D_{\mathsf{t}}D_{\mathsf{r}}} \tilde{g}((\varphi_{\mathsf{t}},\omega_{\mathsf{t}},\varsigma_{(0,0),0}),(\varphi_{\mathsf{r}},\psi_{\mathsf{r}})) e^{-\frac{1}{2} \kappa_{\textrm{abs}}(f)[D_{\mathsf{t}} + D_{\mathsf{r}}]}.
    \end{align}
\end{subequations}

}

The received signal via the $t_p$-$m_{k,l}$-$r_q$ link experiences a total path phase shift of
\begin{equation}\label{totalpps}
    \zeta_{q,(k,l),p} \triangleq \zeta_{(k,l),p} + \zeta_{q,(k,l)} + \angle g(\boldsymbol{\Psi}_{(k,l),p},\boldsymbol{\Psi}_{q,(k,l)}),
\end{equation}
where
\begin{align}\label{ZETA}
    \zeta_{(k,l),p} \triangleq \frac{2\pi d_{(k,l),p}}{\lambda} ~ \textrm{and} ~
    \zeta_{q,(k,l)} \triangleq \frac{2\pi d_{q,(k,l)}}{\lambda}
\end{align}
are the LoS path phase shifts of the links $t_p$-$m_{k,l}$ and $m_{k,l}$-$r_q$, respectively, $\angle g(\boldsymbol{\Psi}_{(k,l),p},\boldsymbol{\Psi}_{q,(k,l)}) = \beta_{k,l}$ is the phase of the RE response $g(\boldsymbol{\Psi}_{(k,l),p},\boldsymbol{\Psi}_{q,(k,l)})$, 
$k \in \mathcal{I}_{Q_{\mathsf{x}}},~ l \in \mathcal{I}_{Q_{\mathsf{y}}},~ p \in \mathcal{I}_{N_{\mathsf{t}}},~ q \in \mathcal{I}_{N_{\mathsf{r}}}$. 



\subsection{Channel Representation}
The considered IRS-aided MIMO system can be generally described by 
\begin{equation}\label{LinearModel}
    \boldsymbol{y} = \boldsymbol{H}\boldsymbol{x} + \boldsymbol{w},
\end{equation}
where $\boldsymbol{H}\in \mathbb{C}^{N_{\mathsf{r}} \times N_{\mathsf{t}}}$ denotes the overall channel coefficient matrix; $\boldsymbol{x}\in \mathbb{C}^{N_{\mathsf{t}}}$ denotes the transmit signal vector; $\boldsymbol{y} \in \mathbb{C}^{N_{\mathsf{r}}}$ denotes the receive signal vector; $\boldsymbol{w}\in \mathbb{C}^{N_{\mathsf{r}}}$ denotes the circularly symmetric complex Gaussian noise that follows $\mathcal{CN}(\boldsymbol{0},\sigma_w^2 \boldsymbol{I})$. The element in the $\left(q+\frac{N_{\mathsf{r}}+1}{2}\right)$-th row and the $\left(p+\frac{N_{\mathsf{t}}+1}{2}\right)$-th column of $\boldsymbol{H}$, denoted by $h_{q,p}$, represents the overall channel between the Tx antenna $t_p$ and the Rx antenna $r_q$, $p \in \mathcal{I}_{N_{\mathsf{t}}},~ q \in \mathcal{I}_{N_{\mathsf{r}}}$. 

We now represent the channel matrix $\boldsymbol{H}$ based on the total path loss and the phase shift model described in the preceding two subsections. Specifically, denote by
\begin{subequations}\label{hklp}
    \begin{align}
        h_{(k,l),p} &\triangleq e^{-j\zeta_{(k,l),p}} \\
        &=
        \begin{multlined}[t]\label{hklpb}
            \exp \left\{ -j \left( \frac{\pi(pd_{\mathsf{t}}\sin \psi_{\mathsf{t}} \cos\gamma_{\mathsf{t}}-k S_{\mathsf{x}} \sin \omega_{\mathsf{t}}+lS_{\mathsf{y}}\cos \omega_{\mathsf{t}})^2}{\lambda D_{\mathsf{t}}} \right. \right. \\
            \left. \left. \quad + \frac{\pi(pd_{\mathsf{t}}\sin \psi_{\mathsf{t}} \sin\gamma_{\mathsf{t}}-kS_{\mathsf{x}}\cos \varphi_{\mathsf{t}}\cos \omega_{\mathsf{t}} - lS_{\mathsf{y}}\cos \varphi_{\mathsf{t}}\sin \omega_{\mathsf{t}})^2}{\lambda D_{\mathsf{t}}} \right. \right. \\
            \left. \left. \quad + \frac{2\pi\left(D_{\mathsf{t}} +pd_{\mathsf{t}}\cos\psi_{\mathsf{t}}-kS_{\mathsf{x}}\sin \varphi_{\mathsf{t}} \cos \omega_{\mathsf{t}} - lS_{\mathsf{y}}\sin \varphi_{\mathsf{t}} \sin \omega_{\mathsf{t}}\right)}{\lambda} \right) \right\}
        \end{multlined}
    \end{align}
\end{subequations}
and 
\begin{subequations}\label{hqkl}
    \begin{align}
        h_{q,(k,l)} &\triangleq e^{-j\zeta_{q,(k,l)}} \\
        &=
        \begin{multlined}[t]\label{hqklb}
            \exp \left\{ -j \left( \frac{\pi(qd_{\mathsf{r}}\sin \psi_{\mathsf{r}} \cos\gamma_{\mathsf{r}}-k S_{\mathsf{x}} \sin \omega_{\mathsf{r}}+lS_{\mathsf{y}}\cos \omega_{\mathsf{r}})^2}{\lambda D_{\mathsf{r}}} \right. \right. \\
            \left. \left. \quad + \frac{\pi(qd_{\mathsf{r}}\sin \psi_{\mathsf{r}} \sin\gamma_{\mathsf{r}}-kS_{\mathsf{x}}\cos \varphi_{\mathsf{r}}\cos \omega_{\mathsf{r}} - lS_{\mathsf{y}}\cos \varphi_{\mathsf{r}}\sin \omega_{\mathsf{r}})^2}{\lambda D_{\mathsf{r}}} \right. \right. \\
            \left. \left. \quad + \frac{2\pi\left(D_{\mathsf{r}} +qd_{\mathsf{r}}\cos\psi_{\mathsf{r}}-kS_{\mathsf{x}}\sin \varphi_{\mathsf{r}} \cos \omega_{\mathsf{r}} - lS_{\mathsf{y}}\sin \varphi_{\mathsf{r}} \sin \omega_{\mathsf{r}}\right)}{\lambda} \right) \right\}
        \end{multlined}
    \end{align}
\end{subequations}
the normalized LoS channels between the Tx antenna $t_p$ and the RE $m_{k,l}$ and between $m_{k,l}$ and the Rx antenna $r_q$, respectively, $k \in \mathcal{I}_{Q_{\mathsf{x}}},~ l \in \mathcal{I}_{Q_{\mathsf{y}}},~ p \in \mathcal{I}_{N_{\mathsf{t}}},~ q \in \mathcal{I}_{N_{\mathsf{r}}}$. \eqref{hklpb} and \eqref{hqklb} are based on \eqref{ZETA} and the link distances in \eqref{Approxklp} and \eqref{Approxqkl}. Let
$\boldsymbol{H}_{\mathsf{t}} \in \mathbb{C}^{Q_{\mathsf{x}}Q_{\mathsf{y}}\times N_{\mathsf{t}}}$ and $\boldsymbol{H}_{\mathsf{r}} \in \mathbb{C}^{N_{\mathsf{r}}\times Q_{\mathsf{x}}Q_{\mathsf{y}}}$
be the normalized LoS Tx-IRS and IRS-Rx channels, respectively, where the element in the $\left(\left(k+\frac{Q_{\mathsf{x}}-1}{2}\right)Q_{\mathsf{y}}+\frac{Q_{\mathsf{y}}+1}{2}+l\right)$-th row and the $\left(p+\frac{N_{\mathsf{t}}+1}{2}\right)$-th column of $\boldsymbol{H}_{\mathsf{t}}$ is given by $h_{(k,l),p}$, and the element in the $\left(q+\frac{N_{\mathsf{r}}+1}{2}\right)$-th row and the $\left(\left(k+\frac{Q_{\mathsf{x}}-1}{2}\right)Q_{\mathsf{y}}+\frac{Q_{\mathsf{y}}+1}{2}+l\right)$-th column of $\boldsymbol{H}_{\mathsf{r}}$ is given by $h_{q,(k,l)}$, $k \in \mathcal{I}_{Q_{\mathsf{x}}},~ l \in \mathcal{I}_{Q_{\mathsf{y}}},~ p \in \mathcal{I}_{N_{\mathsf{t}}},~ q \in \mathcal{I}_{N_{\mathsf{r}}}$. 
Based on \eqref{REPHASE}, denote by
\begin{equation}\label{thetakl}
    \theta_{k,l}\triangleq e^{j\beta_{k,l}}
\end{equation}
the phase shift imposed by the RE $m_{k,l}$, and define a diagonal matrix
\begin{equation}\label{PSI}
    \boldsymbol{\Theta} \triangleq \textrm{diag}\left\{ \left[ \theta_{-\frac{Q_{\mathsf{x}}-1}{2},-\frac{Q_{\mathsf{y}}-1}{2}},\ldots, \theta_{\frac{Q_{\mathsf{x}}-1}{2},\frac{Q_{\mathsf{y}}-1}{2}} \right] \right\}
\end{equation}
 that accounts for the phase shifts imposed by all the REs of the IRS, with $\theta_{k,l}$ in the $\Big(\left(k+\frac{Q_{\mathsf{x}}-1}{2}\right)\allowbreak Q_{\mathsf{y}}+ \frac{Q_{\mathsf{y}}+1}{2} + l \Big)$-th diagonal position, $k \in \mathcal{I}_{Q_{\mathsf{x}}},~ l \in \mathcal{I}_{Q_{\mathsf{y}}}$.
Based on the above expressions, the IRS-aided MIMO channel matrix is given by the following proposition:
\begin{prop}\label{ChannelProp}
    Under Assumptions \ref{REFarFieldAssump} and \ref{SizeAssump}, the overall cascaded LoS MIMO channel $\boldsymbol{H}$ is 
    \begin{equation}\label{H}
        \boldsymbol{H} = \eta_0 \boldsymbol{H}_{\mathsf{r}} \boldsymbol{\Theta} \boldsymbol{H}_{\mathsf{t}}.
    \end{equation}
\end{prop}
\begin{proof}
    From \eqref{CommonPL}, all the Tx-IRS-Rx links share a common total path loss $\eta_0$. 
    From \eqref{totalpps}, \eqref{hklp}, \eqref{hqkl}, and \eqref{thetakl}, the total path phase shift of the received signal experienced in the $t_p$-$m_{k,l}$-$r_q$ link (in exponential) is $h_{(k,l),p} h_{q,(k,l)} \theta_{k,l}$, with $h_{(k,l),p}$ defined in \eqref{hklp}, $h_{q,(k,l)}$ defined in \eqref{hqkl}, and $\theta_{k,l}$ defined in \eqref{thetakl}, $k \in \mathcal{I}_{Q_{\mathsf{x}}},~ l \in \mathcal{I}_{Q_{\mathsf{y}}},~ p \in \mathcal{I}_{N_{\mathsf{t}}},~ q \in \mathcal{I}_{N_{\mathsf{r}}}$.
    Therefore, $\boldsymbol{H}$ is given by \eqref{H}.
\end{proof}

The cascaded LoS MIMO channel model given in \eqref{H} has a similar form to those proposed in the existing works, such as in \cite{najafi_physics-based_2020,tarable_meta-surface_2020,ozdogan_intelligent_2020,zhang_capacity_2020}. The difference is that, based on the link distances derived under the coordinate system, our model characterize the path loss and the path phase shift precisely. 
From \eqref{hklp} and \eqref{hqkl}, we see that the LoS path phase shifts $\zeta_{(k,l),p}$ and $\zeta_{q,(k,l)}$ contain quadratic terms of antenna indices $k$ and $l$. This means that a spherical wavefront on different REs are taken into account in our model. When $D_{\mathsf{t}}$ and $D_{\mathsf{r}}$ tend to infinity, the quadratic terms of $k$ and $l$ in $\zeta_{(k,l),p}$ and $\zeta_{q,(k,l)}$ diminish, and $\zeta_{(k,l),p}$ and $\zeta_{q,(k,l)}$ reduce to linear functions over $k$ and $l$. This means that the entire IRS see a plane wavefront, and our model in \eqref{H} reduces to the plane wave model adopted in \cite{najafi_physics-based_2020,tarable_meta-surface_2020,ozdogan_intelligent_2020,zhang_capacity_2020,li_joint_2019,cai_hierarchical_nodate}. 

In the literature, designing the phase shifts imposed by the REs in reflection is referred to as \textit{passive beamforming}, which largely affects the overall channel properties. In our model, passive beamforming (PB) corresponds to design the initial phase of each RE phase shift function, i.e., $\beta_{k,l}$ in \eqref{thetakl}, $k \in \mathcal{I}_{Q_{\mathsf{x}}},~ l \in \mathcal{I}_{Q_{\mathsf{y}}}$, or the equivalent $\boldsymbol{\Theta}$ in \eqref{PSI}. In \cite{li_joint_2019}, the authors show that when the Tx and the Rx are in the far-field of the IRS, the PB strategy that compensates the LoS path phase shift differences over Tx-IRS-Rx links, leads to the maximum receive signal-to-noise ratio (SNR).
Under such a PB design, the phase shift imposed by each RE is a linear function of the RE indices, and the IRS achieves anomalous reflection based on the generalized Snell's law \cite{yu_light_2011}. Similar PB strategies that employ linear phase shifts (over the RE indices) are also discussed in \cite{tarable_meta-surface_2020,najafi_physics-based_2020,cai_hierarchical_nodate}. 
But in the channel model considered here, since the curvature of the wavefront on the IRS is taken into account, PB with linear phase shifts is insufficient to compensate the path phase shifts differences over Tx-IRS-Rx links. This inspires us to come up with a new PB strategy, named \textit{reflective focusing}\footnotemark, as detailed in the following subsection.
\footnotetext{\blue{Note that reflective focusing is also referred to as \textit{beamfocusing} in a parallel work \cite{dovelosIntelligentReflectingSurfaces2021a}.}}

\subsection{Channel Representation with Reflective Focusing}
In this subsection, we describe a special channel representation with reflective focusing. 
Reflective focusing aims to make the EM wave radiated by a certain Tx antenna, after being reflected by different REs, coherently superimposed at a certain Rx antenna. This can be achieved by properly designing the phase shifts of RE phase shift functions. 
To be specific, the reflective focusing for the pair of Tx antenna $t_p$ and Rx antenna $r_q$ can be achieved by letting
\begin{equation}\label{REFOCUSING}
    \beta_{k,l} = \frac{d_{(k,l),p} + d_{q,(k,l)}}{\lambda},
\end{equation}
where $d_{(k,l),p}$ and $d_{q,(k,l)}$ are the  $t_p$-$m_{k,l}$ and $m_{k,l}$-$r_q$ link distances given in \eqref{Approxklp} and \eqref{Approxqkl}, respectively, $k \in \mathcal{I}_{Q_{\mathsf{x}}},~ l \in \mathcal{I}_{Q_{\mathsf{y}}}$. With the RE phase shift in \eqref{REFOCUSING}, the EM wave emitted from $t_p$, after being reflected by different REs, have the same phase at $r_q$. 

Without loss of generality, we henceforth always assume that the Tx-Rx antenna pair concerned in reflective focusing is $(t_0,r_0)$, achieved by letting
\begin{equation}\label{REfocusing0}
    \beta_{k,l} = \bar{\beta}_{k,l} \triangleq \frac{d_{(k,l),0} + d_{0,(k,l)}}{\lambda},
\end{equation}
$k \in \mathcal{I}_{Q_{\mathsf{x}}},~ l \in \mathcal{I}_{Q_{\mathsf{y}}}$.
Based on channel model \eqref{H}, with reflective focusing in \eqref{REfocusing0}, the channel between the Tx antenna $t_p$ and the Rx antenna $r_q$ is given by 
\begin{subequations}\label{hqp}
    \begin{align}
        h_{q,p} &\triangleq \eta_0 \sum_{k=-\frac{Q_{\mathsf{x}}-1}{2}}^{\frac{Q_{\mathsf{x}}-1}{2}} \sum_{l=-\frac{Q_{\mathsf{y}}-1}{2}}^{\frac{Q_{\mathsf{y}}-1}{2}}  h_{(k,l),p} \bar{\beta}_{k,l} h_{q,(k,l)} \\
        &= \eta_0 P_{q,p} Q_{q,p}
    \end{align}
\end{subequations}
where
\begin{align}\label{Pqpfinal}
    P_{q,p} &\triangleq \exp \left\{ \! \frac{-j2\pi}{\lambda} \left( \! \frac{\left(d_{\mathsf{t}}\sin \psi_{\mathsf{t}} \cos\gamma_{\mathsf{t}} \right)^2}{2D_{\mathsf{t}}} p^2 \! + \! \frac{\left(d_{\mathsf{r}}\sin \psi_{\mathsf{r}} \cos \gamma_{\mathsf{r}} \right)^2}{2D_{\mathsf{r}}} q^2 \! + \! pd_{\mathsf{t}}\cos\psi_{\mathsf{t}} \! + \! qd_{\mathsf{r}}\cos\psi_{\mathsf{r}} \! \right) \! \right\}
\end{align}
is irrelevant to the indices $k$ and $l$, and
\begin{subequations}\label{Qqpfinal}
    \begin{align}
        Q_{q,p} &\triangleq \sum_{k=-\frac{Q_{\mathsf{x}}-1}{2}}^{\frac{Q_{\mathsf{x}}-1}{2}} \sum_{l=-\frac{Q_{\mathsf{y}}-1}{2}}^{\frac{Q_{\mathsf{y}}-1}{2}} \exp \left\{\frac{j2\pi}{\lambda} \left[ \left( \frac{d_{\mathsf{t}} S_{\mathsf{x}} \sin \psi_{\mathsf{t}} (\cos\gamma_{\mathsf{t}} \sin \omega_{\mathsf{t}}+ \sin\gamma_{\mathsf{t}} \cos \varphi_{\mathsf{t}} \cos\omega_{\mathsf{t}} )}{D_{\mathsf{t}}}p  \right. \right. \right. \notag \\
        & \quad \left. \left. \left. + \frac{d_{\mathsf{r}} S_{\mathsf{x}} \sin \psi_{\mathsf{r}} (\cos\gamma_{\mathsf{r}} \sin \omega_{\mathsf{r}}+ \sin\gamma_{\mathsf{r}} \cos \varphi_{\mathsf{r}} \cos\omega_{\mathsf{r}} )}{D_{\mathsf{r}}}q \right)k \right. \right. \notag \\
        & \quad + \left. \left. \left( \frac{d_{\mathsf{t}}S_{\mathsf{y}}\sin \psi_{\mathsf{t}} (-\cos \gamma_{\mathsf{t}} \cos \omega_{\mathsf{t}} + \sin\gamma_{\mathsf{t}}\cos \varphi_{\mathsf{t}} \sin \omega_{\mathsf{t}})}{D_{\mathsf{t}}}p \right. \right. \right. \notag \\
        & \quad \left. \left. \left. \frac{d_{\mathsf{r}}S_{\mathsf{y}}\sin \psi_{\mathsf{r}} (-\cos \gamma_{\mathsf{r}} \cos \omega_{\mathsf{r}} + \sin\gamma_{\mathsf{r}}\cos \varphi_{\mathsf{r}} \sin \omega_{\mathsf{r}})}{D_{\mathsf{r}}}q \right)l \right] \right\} \label{Qqp3} \\    
        &= \frac{\sin\left(\pi \left( C_{\mathsf{t},\mathsf{x}}p + C_{\mathsf{r},\mathsf{x}}q \right)\right)\sin\left(\pi \left( C_{\mathsf{t},\mathsf{y}}p + C_{\mathsf{r},\mathsf{y}}q \right)\right)}{\sin\left(\frac{\pi}{Q_{\mathsf{x}}} \left( C_{\mathsf{t},\mathsf{x}}p + C_{\mathsf{r},\mathsf{x}}q \right)\right)\sin\left(\frac{\pi}{Q_{\mathsf{y}}} \left( C_{\mathsf{t},\mathsf{y}}p + C_{\mathsf{r},\mathsf{y}}q \right)\right)},\label{Qqpfrac}
    \end{align}
\end{subequations}
$p \in \mathcal{I}_{N_{\mathsf{t}}},~ q \in \mathcal{I}_{N_{\mathsf{r}}}$, with
\begin{subequations}\label{C}
    \begin{align}
        C_{\mathsf{t},\mathsf{x}} &\triangleq \frac{d_{\mathsf{t}}S_{\mathsf{x}}Q_{\mathsf{x}}A_{\mathsf{t},\mathsf{x}} \sin\psi_{\mathsf{t}}\cos(\gamma_{\mathsf{t}}-\bar{\gamma}_{\mathsf{t},\mathsf{x}})}{\lambda D_{\mathsf{t}}},~ C_{\mathsf{t},\mathsf{y}}\triangleq \frac{d_{\mathsf{t}}S_{\mathsf{y}}Q_{\mathsf{y}}A_{\mathsf{t},\mathsf{y}}\sin\psi_{\mathsf{t}}\cos(\gamma_{\mathsf{t}}-\bar{\gamma}_{\mathsf{t},\mathsf{y}})}{\lambda D_{\mathsf{t}}}, \label{Ctx}\\
        C_{\mathsf{r},\mathsf{x}} &\triangleq \frac{d_{\mathsf{r}}S_{\mathsf{x}}Q_{\mathsf{x}}A_{\mathsf{r},\mathsf{x}}\sin\psi_{\mathsf{r}}\cos(\gamma_{\mathsf{r}}-\bar{\gamma}_{\mathsf{r},\mathsf{x}})}{\lambda D_{\mathsf{r}}},~ C_{\mathsf{r},\mathsf{y}} \triangleq \frac{d_{\mathsf{r}}S_{\mathsf{y}}Q_{\mathsf{y}}A_{\mathsf{r},\mathsf{y}}\sin\psi_{\mathsf{r}}\cos(\gamma_{\mathsf{r}}-\bar{\gamma}_{\mathsf{r},\mathsf{y}})}{\lambda D_{\mathsf{r}}}, \label{Crx}
    \end{align}
\end{subequations}
and $A_{\mathsf{t},\mathsf{x}}$, $A_{\mathsf{t},\mathsf{y}}$, $A_{\mathsf{r},\mathsf{x}}$, $A_{\mathsf{r},\mathsf{y}}$, $\bar{\gamma}_{\mathsf{t},\mathsf{x}}$, $\bar{\gamma}_{\mathsf{t},\mathsf{y}}$, $\bar{\gamma}_{\mathsf{r},\mathsf{x}}$, and $\bar{\gamma}_{\mathsf{r},\mathsf{y}}$ defined by
\begin{subequations}
    \begin{align}
        A_{\mathsf{t},\mathsf{x}} &\triangleq \sqrt{\sin^2 \omega_{\mathsf{t}} + \cos^2 \varphi_{\mathsf{t}} \cos^2 \omega_{\mathsf{t}}}, ~\cos \bar{\gamma}_{\mathsf{t},\mathsf{x}}= \frac{\sin \omega_{\mathsf{t}}}{A_{\mathsf{t},\mathsf{x}}}, ~ \sin \bar{\gamma}_{\mathsf{t},\mathsf{x}}= \frac{\cos \varphi_{\mathsf{t}}\cos \omega_{\mathsf{t}}}{A_{\mathsf{t},\mathsf{x}}}, \label{Atx}\\
        A_{\mathsf{t},\mathsf{y}} &\triangleq \sqrt{\cos^2 \omega_{\mathsf{t}} + \cos^2 \varphi_{\mathsf{t}} \sin^2 \omega_{\mathsf{t}}}, ~\cos \bar{\gamma}_{\mathsf{t},\mathsf{y}}= -\frac{\cos \omega_{\mathsf{t}}}{A_{\mathsf{t},\mathsf{y}}}, ~ \sin \bar{\gamma}_{\mathsf{t},\mathsf{y}}= \frac{\cos \varphi_{\mathsf{t}}\sin \omega_{\mathsf{t}}}{A_{\mathsf{t},\mathsf{y}}},\label{Aty}\\
        A_{\mathsf{r},\mathsf{x}} &\triangleq \sqrt{\sin^2 \omega_{\mathsf{r}} + \cos^2 \varphi_{\mathsf{r}} \cos^2 \omega_{\mathsf{r}}}, ~\cos \bar{\gamma}_{\mathsf{r},\mathsf{x}}= \frac{\sin \omega_{\mathsf{r}}}{A_{\mathsf{r},\mathsf{x}}}, ~ \sin \bar{\gamma}_{\mathsf{r},\mathsf{x}}= \frac{\cos \varphi_{\mathsf{r}}\cos \omega_{\mathsf{r}}}{A_{\mathsf{r},\mathsf{x}}},\label{Arx}\\
        A_{\mathsf{r},\mathsf{y}} &\triangleq \sqrt{\cos^2 \omega_{\mathsf{r}} + \cos^2 \varphi_{\mathsf{r}} \sin^2 \omega_{\mathsf{r}}}, ~\cos \bar{\gamma}_{\mathsf{r},\mathsf{y}}= -\frac{\cos \omega_{\mathsf{r}}}{A_{\mathsf{r},\mathsf{y}}}, ~ \sin \bar{\gamma}_{\mathsf{r},\mathsf{y}}= \frac{\cos \varphi_{\mathsf{r}}\sin \omega_{\mathsf{r}}}{A_{\mathsf{r},\mathsf{y}}}.\label{Ary}
    \end{align}
\end{subequations}
The simplification from \eqref{Qqp3} to \eqref{Qqpfrac} employs the geometric sum formula and the trigonometric addition formulas. Denote by $\boldsymbol{P}\in \mathbb{C}^{N_{\mathsf{r}}\times N_{\mathsf{t}}}$ and $\boldsymbol{Q}\in \mathbb{C}^{N_{\mathsf{r}}\times N_{\mathsf{t}}}$ the matrices with $P_{q,p}$ and $Q_{q,p}$ in the $\left(q+\frac{N_{\mathsf{t}}+1}{2}\right)$-th row and the $\left(p+\frac{N_{\mathsf{r}}+1}{2}\right)$-th column, respectively, $p \in \mathcal{I}_{N_{\mathsf{t}}},~ q \in \mathcal{I}_{N_{\mathsf{r}}}$. Then the channel with reflective focusing can be expressed as
\begin{equation}\label{HPQ}
    \boldsymbol{H} = \eta_0 \boldsymbol{P} \odot \boldsymbol{Q},
\end{equation}
where $\odot$ is the Hadamard product operator. Next, we will show that the cascaded LoS MIMO channel with reflective focusing has the potential to support \textit{full multiplexing} communication between the Tx and the Rx.




\section{Full Multiplexing Region of the Cascaded LoS MIMO Channel}\label{SecFMR}
In this section, we study the full multiplexing region (FMR) of the cascaded LoS MIMO channel. We first give the Rayleigh distances of the single-hop LoS channels at the transmit and receive sides, and then derive an inner bound of the FMR of the overall cascaded LoS MIMO channel by using reflective focusing. 

\subsection{Rayleigh Distances of the Tx-IRS and IRS-Rx Channels}
In this subsection, we study the Rayleigh distances of the two single-hop LoS MIMO channels, i.e., the Tx-IRS and the IRS-Rx channels.\footnotemark
\footnotetext{There are many existing works related to the Rayleigh distance of the MIMO channel. For example, Ref. \cite{wang_tens_2014} discusses the Rayleigh distance between ULAs; Ref. \cite{bohagen_optimal_2007} provides the optimal design criteria for antenna spacings to enable full spatial multiplexing between URAs with given distance. In this paper, we assume that the Tx, the Rx, and the IRS are with fixed antenna (RE) spacings, and focus on how the channel properties vary over the distances between them. Thus, the definition of Rayleigh distance is slightly different from the one in \cite{bohagen_optimal_2007} and \cite{wang_tens_2014}.}
Roughly speaking, for a LoS MIMO channel with $M$ antennas at the Tx, $N$ antennas at the Rx, and fixed antenna spacings, Rayleigh distance is the largest Tx-Rx distance that allows the LoS MIMO channel to support $\min\{M,N\}$ simultaneous spatial streams with equal channel quality \cite{wang_tens_2014}.
With the Tx-Rx distance shorter than the Rayleigh distance, the LoS MIMO system is able to harvest the best multiplexing gain brought by exploiting the spatial degree of freedom of the channel. In our model, the IRS can be viewed as a $Q_{\mathsf{x}} \times Q_{\mathsf{y}}$ uniform rectangular array (URA). Thus, the Tx-IRS (or IRS-Rx) channel can be viewed as a MIMO channel with a ULA (or URA) and a URA (or ULA) deployed at the transmitter and the receiver. Studying the Rayleigh distances of the Tx-IRS and IRS-Rx channels will help understand the properties of the overall cascaded LoS MIMO channel.

In our model, the Tx-IRS and IRS-Rx channels are given by matrices $\boldsymbol{H}_{\mathsf{t}} \in \mathbb{C}^{Q_{\mathsf{x}}Q_{\mathsf{y}}\times N_{\mathsf{t}}}$ and $\boldsymbol{H}_{\mathsf{r}} \in \mathbb{C}^{N_{\mathsf{r}}\times Q_{\mathsf{x}}Q_{\mathsf{y}}}$, respectively. Elements of $\boldsymbol{H}_{\mathsf{t}}$ and $\boldsymbol{H}_{\mathsf{r}}$ are normalized complex exponentials with the phases determined by the link distances. We assume that the number of REs on the IRS is no less than the number of antennas in both Tx and Rx, i.e., $Q_{\mathsf{x}}Q_{\mathsf{y}} \geq N_{\mathsf{t}},N_{\mathsf{r}}$. For the Tx-IRS channel to support spatial multiplexing of $N_{\mathsf{t}}$ data streams with equal channel gain, the columns of $\boldsymbol{H}_{\mathsf{t}}$ are required to satisfy
\begin{align}\label{TOrthCrtt}
    \langle \boldsymbol{h}_{\mathsf{t};p_1}, \boldsymbol{h}_{\mathsf{t};p_2} \rangle = 
    \left\{\begin{array}{ll}
        Q_{\mathsf{x}}Q_{\mathsf{y}}, & p_1 = p_2 \\
        0, & p_1 \neq p_2   
    \end{array}\right.,
        ~ p_1 \in \mathcal{I}_{N_{\mathsf{t}}},~ p_2 \in \mathcal{I}_{N_{\mathsf{t}}},
\end{align}
where $\langle \cdot,\cdot \rangle$ denotes the inner-product operator, and $\boldsymbol{h}_{\mathsf{t};p}$ denotes the $\left(p+\frac{N_{\mathsf{t}}+1}{2}\right)$-th column of $\boldsymbol{H}_{\mathsf{t}}$, $p \in \mathcal{I}_{N_{\mathsf{t}}}$.
In \eqref{TOrthCrtt}, the columns of $\boldsymbol{H}_{\mathsf{t}}$ are orthogonal, and the $N_{\mathsf{t}}$ eigenchannels of $\boldsymbol{H}_{\mathsf{t}}$ share a common channel gain of $Q_{\mathsf{x}}Q_{\mathsf{y}}$. The corresponding Tx-IRS channel is said to be able to support \textit{full multiplexing} of $N_{\mathsf{t}}$ spatial streams. Similarly, for the IRS-Rx channel to support full multiplexing of $N_{\mathsf{r}}$ spatial streams, the rows of $\boldsymbol{H}_{\mathsf{r}}$ satisfy the orthogonality and equal-gain requirement as
\begin{align}\label{TOrthCrtr}
    \langle \boldsymbol{h}_{{\mathsf{r}},q_1}, \boldsymbol{h}_{{\mathsf{r}},q_2} \rangle = 
    \left\{\begin{array}{ll}
        Q_{\mathsf{x}}Q_{\mathsf{y}}, & q_1 = q_2 \\
        0, & q_1 \neq q_2   
    \end{array}\right.,
        ~ q_1 \in \mathcal{I}_{N_{\mathsf{r}}},~ q_2 \in \mathcal{I}_{N_{\mathsf{r}}},
\end{align}
where $\boldsymbol{h}_{{\mathsf{r}},q}$ denotes the $\left(q+\frac{N_{\mathsf{r}}+1}{2}\right)$-th row of $\boldsymbol{H}_{\mathsf{r}}$, $q \in \mathcal{I}_{N_{\mathsf{r}}}$.
Consider that the positions of the IRS and the centers of Tx and Rx are fixed. When the Tx-IRS (or IRS-Rx) distance $D_{\mathsf{t}}$ (or $D_{\mathsf{r}}$) is shorter than the Rayleigh distance of the corresponding LoS MIMO channel, the orthogonality and equal-gain requirement on $\boldsymbol{H}_{\mathsf{t}}$ (or $\boldsymbol{H}_{\mathsf{r}}$) can be met by appropriately adjusting the Tx (Rx) orientation. We now present the formal definitions of the Rayleigh distances of the Tx-IRS and IRS-Rx channels.
\begin{defn}
    For given Tx-IRS direction $(\varphi_{\mathsf{t}},\omega_{\mathsf{t}})$, the \textbf{Tx-IRS Rayleigh distance $D_{\mathsf{t}}^{\mathsf{R}}$} is the largest $D_{\mathsf{t}}$ that enables the columns of $\boldsymbol{H}_{\mathsf{t}}$ to satisfy \eqref{TOrthCrtt} only by tuning the Tx orientation angles, i.e., $\psi_{\mathsf{t}}$ and $\gamma_{\mathsf{t}}$. For given IRS-Rx direction $(\varphi_{\mathsf{r}},\omega_{\mathsf{r}})$, the \textbf{IRS-Rx Rayleigh distance $\mathcal{D}_{\mathsf{r}}(\varphi_{\mathsf{r}},\omega_{\mathsf{r}})$} are defined in a similar way as $D_{\mathsf{r}}^{\mathsf{R}}$ by replacing the subscript $\mathsf{t}$ by $\mathsf{r}$ and \eqref{TOrthCrtt} by \eqref{TOrthCrtr}.
\end{defn}

We have the following results on $D_{\mathsf{t}}^{\mathsf{R}}$ and $D_{\mathsf{r}}^{\mathsf{R}}$.
\begin{prop}\label{propSingleSideDray}
    For the Tx-IRS channel, if $Q_{\mathsf{x}} \geq N_{\mathsf{t}}$, $D_{\mathsf{t}}^{\mathsf{R}} \geq D_{\mathsf{t},\mathsf{x}}^{\mathsf{R}} \triangleq \frac{d_{\mathsf{t}} S_{\mathsf{x}}Q_{\mathsf{x}}A_{\mathsf{t},\mathsf{x}}}{\lambda}$ with $A_{\mathsf{t},\mathsf{x}}$ in \eqref{Atx}; if $Q_{\mathsf{y}} \geq N_{\mathsf{t}}$, $D_{\mathsf{t}}^{\mathsf{R}} \geq D_{\mathsf{t},\mathsf{y}} \triangleq \frac{d_{\mathsf{t}} S_{\mathsf{y}}Q_{\mathsf{y}}A_{\mathsf{t},\mathsf{y}}}{\lambda }$ with $A_{\mathsf{t},\mathsf{y}}$ in \eqref{Aty}; if $Q_{\mathsf{x}} \geq N_{\mathsf{t}}$ and $Q_{\mathsf{y}} \geq N_{\mathsf{t}}$, $D_{\mathsf{t}}^{\mathsf{R}}$ is given by
    \begin{equation}\label{DT}
        D_{\mathsf{t}}^{\mathsf{R}} = \max\left\{ D_{\mathsf{t},\mathsf{x}}^{\mathsf{R}}, D_{\mathsf{t},\mathsf{y}}^{\mathsf{R}} \right\}.
    \end{equation}
    When $D_{\mathsf{t}} \leq D_{\mathsf{t}}^{\mathsf{R}}=D_{\mathsf{t},\mathsf{x}}^{\mathsf{R}}$, the orthogonality and equal-gain requirement \eqref{TOrthCrtt} can be satisfied by letting
    \begin{align}\label{txangleCrt}
        |\cos(\gamma_{\mathsf{t}}-\bar{\gamma}_{\mathsf{t},\mathsf{x}})| = 1 ~\textup{and} ~ |\sin\psi_{\mathsf{t}}| = \frac{D_{\mathsf{t}}}{D_{\mathsf{t},\mathsf{x}}^{\mathsf{R}}}, ~ \textup{or}~ |\cos(\gamma_{\mathsf{t}}-\bar{\gamma}_{\mathsf{t},\mathsf{x}})| = \frac{D_{\mathsf{t}}}{D_{\mathsf{t},\mathsf{x}}^{\mathsf{R}}}  ~\textup{and} ~ |\sin\psi_{\mathsf{t}}| = 1.
    \end{align}
    When $D_{\mathsf{t}} \leq D_{\mathsf{t}}^{\mathsf{R}}=D_{\mathsf{t},\mathsf{y}}^{\mathsf{R}}$, the requirement \eqref{TOrthCrtt} can be satisfied by letting
    \begin{align}\label{tyangleCrt}
        |\cos(\gamma_{\mathsf{t}}-\bar{\gamma}_{\mathsf{t},\mathsf{y}})| = 1 ~\textup{and} ~ |\sin\psi_{\mathsf{t}}| = \frac{D_{\mathsf{t}}}{D_{\mathsf{t},\mathsf{y}}^{\mathsf{R}}}, ~ \textup{or}~ |\cos(\gamma_{\mathsf{t}}-\bar{\gamma}_{\mathsf{t},\mathsf{y}})| = \frac{D_{\mathsf{t}}}{D_{\mathsf{t},\mathsf{y}}^{\mathsf{R}}}  ~\textup{and} ~ |\sin\psi_{\mathsf{t}}| = 1.
    \end{align}
    The above results literally hold for the IRS-Rx channel by replacing the subscript $\mathsf{t}$ by $\mathsf{r}$, \eqref{Atx} by \eqref{Arx}, \eqref{Aty} by \eqref{Ary}, and \eqref{TOrthCrtt} by \eqref{TOrthCrtr}.
\end{prop}

The proof of Proposition \ref{propSingleSideDray} can be found in Appendix \ref{proofSingleSideDray}. 
Proposition \ref{propSingleSideDray} gives the Rayleigh distances of the Tx-IRS/IRS-Rx channels and the Tx/Rx orientations to enable full multiplexing between Tx/Rx and the IRS when $D_{\mathsf{t}} \leq D_{\mathsf{t}}^{\mathsf{R}}$ or $D_{\mathsf{r}} \leq D_{\mathsf{r}}^{\mathsf{R}}$. 

From \eqref{H}, the overall cascaded LoS MIMO channel $\boldsymbol{H}$ is jointly determined by the Tx-IRS channel $\boldsymbol{H}_{\mathsf{t}}$, the IRS-Rx channel $\boldsymbol{H}_{\mathsf{r}}$, the passive beamforming matrix $\boldsymbol{\Theta}$, and the path loss $\eta_0$. This means that even if both the Tx-IRS channel and the IRS-Rx channel can support full multiplexing, the full multiplexing capability of the overall cascaded LoS MIMO channel cannot be necessarily guaranteed. As such, it is necessary to further consider the full multiplexing capability of the overall cascaded LoS MIMO channel, as detailed in the next subsection. 

\subsection{Full Multiplexing Region of the Cascaded LoS MIMO Channel}

We first give a formal definition to the notion of \textit{full multiplexing} in the cascaded LoS MIMO channel. Without loss of generality, we henceforth always assume $N_{\mathsf{t}}\geq N_{\mathsf{r}}$. A cascaded LoS MIMO channel (with $N_{\mathsf{t}}\geq N_{\mathsf{r}}$) is qualified to support full multiplexing if the rows of $\boldsymbol{H}$ satisfy
    \begin{equation}\label{NtgeqNr}
        \langle \boldsymbol{h}_{q_1}, \boldsymbol{h}_{q_2} \rangle = \eta_0^2 Q_{\mathsf{x}}^2 Q_{\mathsf{y}}^2,~ \textrm{for}~ q_1 = q_2, ~\textrm{and} ~  \langle \boldsymbol{h}_{q_1}, \boldsymbol{h}_{q_2} \rangle =0,~ \textrm{for}~ q_1 \neq q_2,
    \end{equation}
where $q_1 \in \mathcal{I}_{N_{\mathsf{r}}},~ q_2 \in \mathcal{I}_{N_{\mathsf{r}}}$, and $\boldsymbol{h}_{q}$ denotes the $\left(q+\frac{N_{\mathsf{r}}+1}{2}\right)$-th row of $\boldsymbol{H}$. In \eqref{NtgeqNr}, $\eta_0^2 Q_{\mathsf{x}}^2 Q_{\mathsf{y}}^2$ is the overall channel gain by taking into account the power gain of an eigenchannel of $\boldsymbol{H}_{\mathsf{t}}$ in \eqref{TOrthCrtt}, that of $\boldsymbol{H}_{\mathsf{r}}$ in \eqref{TOrthCrtr}, and the total path loss $\eta_0^2$. 

With the above notion of full multiplexing, we now present the formal definition of the FMR of the cascaded LoS MIMO channel.
\begin{defn}\label{DefnGDRay} 
    For given Tx-IRS direction $(\varphi_{\mathsf{t}},\omega_{\mathsf{t}})$ and IRS-Rx direction $(\varphi_{\mathsf{r}},\omega_{\mathsf{r}})$, the \textbf{full multiplexing region (FMR)} of the cascaded LoS MIMO channel, denoted by $\mathcal{D}_{\mathsf{c}}(\mathcal{L})$ with $\mathcal{L} \triangleq\{\varphi_{\mathsf{t}},\omega_{\mathsf{t}},\varphi_{\mathsf{r}},\omega_{\mathsf{r}}\}$, is the union of all $(D_{\mathsf{t}},D_{\mathsf{r}})$ pairs that enable $\boldsymbol{H}$ to satisfy \eqref{NtgeqNr} by tuning the Tx/Rx orientations $\mathcal{A}\triangleq \left\{\psi_{\mathsf{t}},\gamma_{\mathsf{t}},\psi_{\mathsf{r}},\gamma_{\mathsf{r}}\right\}$ and the PB of the IRS.
\end{defn}

We have the following results on the FMR $\mathcal{D}_{\mathsf{c}}(\mathcal{L})$.

\begin{prop}\label{propGray}
    For a cascaded LoS MIMO channel with $N_{\mathsf{t}}+N_{\mathsf{r}}-2<2Q_{\mathsf{x}}$ and $N_{\mathsf{t}}+N_{\mathsf{r}}-2<2Q_{\mathsf{y}}$, $\mathcal{D}_{\mathsf{c}}(\mathcal{L})$ at least covers the union of two regions $\mathcal{D}_{\mathsf{c},\mathsf{x}}(\mathcal{L})$ and $\mathcal{D}_{\mathsf{c},\mathsf{y}}(\mathcal{L})$. Specifically, $\mathcal{D}_{\mathsf{c},\mathsf{x}}(\mathcal{L})$ is the union of two regions defined by
    \begin{align}
        &\mathcal{D}_{\mathsf{c},\mathsf{x}}^1(\mathcal{L})\triangleq \left\{ (D_{\mathsf{t}},D_{\mathsf{r}})| 0<D_{\mathsf{t}} \leq D_{\mathsf{t},\mathsf{x}}^*,~ 0<D_{\mathsf{r}} \leq D_{\mathsf{r},\mathsf{x}}^{\mathsf{R}} \right\} ~ \textup{and}~ \label{DcxStart} \\
        &\mathcal{D}_{\mathsf{c},\mathsf{x}}^2(\mathcal{L})\triangleq \left\{ (D_{\mathsf{t}},D_{\mathsf{r}})| D_{\mathsf{t},\mathsf{x}}^* < D_{\mathsf{t}} \leq D_{\mathsf{t},\mathsf{x}}^{\mathsf{R}},~ 0<D_{\mathsf{r}} \leq \mathcal{B}_{\mathsf{c},\mathsf{x}}(D_{\mathsf{t}}) \right\}, \label{Dcx2}
    \end{align}
    where 
    \begin{align}
        D_{\mathsf{t},\mathsf{x}}^* &\triangleq D_{\mathsf{t},\mathsf{x}}^{\mathsf{R}}|\cos(\gamma_{\mathsf{t},\mathsf{x}}^*-\bar{\gamma}_{\mathsf{t},\mathsf{x}})|~ \textup{and} \label{GammatxStara} \\
        \gamma_{\mathsf{t},\mathsf{x}}^* &\triangleq \arctan \left(\frac{A_{\mathsf{t},\mathsf{y}}A_{\mathsf{r},\mathsf{x}} \cos \bar{\gamma}_{\mathsf{t},\mathsf{y}}-A_{\mathsf{t},\mathsf{x}} A_{\mathsf{r},\mathsf{y}} \cos(\bar{\gamma}_{\mathsf{r},\mathsf{x}}-\bar{\gamma}_{\mathsf{r},\mathsf{y}}) \cos \bar{\gamma}_{\mathsf{t},\mathsf{x}}}{A_{\mathsf{t},\mathsf{x}}A_{\mathsf{r},\mathsf{y}}\cos(\bar{\gamma}_{\mathsf{r},\mathsf{x}}-\bar{\gamma}_{\mathsf{r},\mathsf{y}}) \sin \bar{\gamma}_{\mathsf{t},\mathsf{x}}-A_{\mathsf{t},\mathsf{y}} A_{\mathsf{r},\mathsf{x}}  \sin \bar{\gamma}_{\mathsf{t},\mathsf{y}}} \right). \label{GammatxStarb} 
    \end{align}
    $\mathcal{B}_{\mathsf{c},\mathsf{x}}(D_{\mathsf{t}})$ is defined by 
    \begin{align}
        &\mathcal{B}_{\mathsf{c},\mathsf{x}}(D_{\mathsf{t}}) \triangleq D_{\mathsf{r},\mathsf{x}}^{\mathsf{R}} \left\{|\cos(\gamma_{\mathsf{x}}(D_{\mathsf{t}})-\bar{\gamma}_{\mathsf{r},\mathsf{x}})|\right\},~ \textrm{where} \label{Bcf2}
    \end{align}
    \begin{multline}
        \gamma_{\mathsf{x}}(D_{\mathsf{t}}) = \\  \! \arctan \left( \! \! \frac{A_{\mathsf{t},\mathsf{x}}A_{\mathsf{r},\mathsf{y}} \cos \bar{\gamma}_{\mathsf{r},\mathsf{y}} \! - \! A_{\mathsf{t},\mathsf{y}} A_{\mathsf{r},\mathsf{x}} \cos \bar{\gamma}_{\mathsf{r},\mathsf{x}} \left( \! \cos(\bar{\gamma}_{\mathsf{t},\mathsf{x}} \! - \! \bar{\gamma}_{\mathsf{t},\mathsf{y}}) \pm \sin(\bar{\gamma}_{\mathsf{t},\mathsf{x}} \! - \! \bar{\gamma}_{\mathsf{t},\mathsf{y}}) \sqrt{\left(\frac{D_{\mathsf{t},\mathsf{x}}^{\mathsf{R}}}{D_{\mathsf{t}}}\right)^2 \! - \! 1} \right)}{A_{\mathsf{t},\mathsf{y}} A_{\mathsf{r},\mathsf{x}} \sin \bar{\gamma}_{\mathsf{r},\mathsf{x}} \left( \! \cos(\bar{\gamma}_{\mathsf{t},\mathsf{x}} \! - \! \bar{\gamma}_{\mathsf{t},\mathsf{y}}) \pm \sin(\bar{\gamma}_{\mathsf{t},\mathsf{x}} \! - \! \bar{\gamma}_{\mathsf{t},\mathsf{y}}) \sqrt{\left(\frac{D_{\mathsf{t},\mathsf{x}}^{\mathsf{R}}}{D_{\mathsf{t}}}\right)^2 \! - \! 1} \right) \! - \! A_{\mathsf{t},\mathsf{x}}A_{\mathsf{r},\mathsf{y}} \sin \bar{\gamma}_{\mathsf{r},\mathsf{y}}} \! \! \right) \label{gammaxPstva}
    \end{multline}
    with ``$\pm$'' selected to be ``$+$'' or ``$-$'' that leads to a larger $\mathcal{B}_{\mathsf{c},\mathsf{x}}(D_{\mathsf{t}})$ in \eqref{Bcf2}.
    When $(D_{\mathsf{t}},D_{\mathsf{r}}) \in \mathcal{D}_{\mathsf{c},\mathsf{x}}^1(\mathcal{L})$, \eqref{NtgeqNr} can be met by letting
    \begin{align}
        &\tan\gamma_{\mathsf{t}} = \tan\gamma_{\mathsf{t},\mathsf{x}}^* = \frac{A_{\mathsf{t},\mathsf{y}}A_{\mathsf{r},\mathsf{x}} \cos \bar{\gamma}_{\mathsf{t},\mathsf{y}}-A_{\mathsf{t},\mathsf{x}} A_{\mathsf{r},\mathsf{y}} \cos(\bar{\gamma}_{\mathsf{r},\mathsf{x}}-\bar{\gamma}_{\mathsf{r},\mathsf{y}}) \cos \bar{\gamma}_{\mathsf{t},\mathsf{x}}}{A_{\mathsf{t},\mathsf{x}}A_{\mathsf{r},\mathsf{y}}\cos(\bar{\gamma}_{\mathsf{r},\mathsf{x}}-\bar{\gamma}_{\mathsf{r},\mathsf{y}}) \sin \bar{\gamma}_{\mathsf{t},\mathsf{x}}-A_{\mathsf{t},\mathsf{y}} A_{\mathsf{r},\mathsf{x}}  \sin \bar{\gamma}_{\mathsf{t},\mathsf{y}}}, \label{Dcx1Orienta} \\
        & \tan\gamma_{\mathsf{r}}=\tan\bar{\gamma}_{\mathsf{r},\mathsf{x}} = \frac{\cos \varphi_{\mathsf{r}}}{\tan \omega_{\mathsf{r}}} \label{Dcx1Orienta2}  \\
        &\sin\psi_{\mathsf{t}} = \frac{D_{\mathsf{t}}}{D_{\mathsf{t},\mathsf{x}}^*}, ~\textup{and}~ \sin\psi_{\mathsf{r}} = \frac{D_{\mathsf{r}}}{D_{\mathsf{r},\mathsf{x}}^{\mathsf{R}}}. \label{Dcx1Orientb}
    \end{align}
    In $\mathcal{D}_{\mathsf{c},\mathsf{x}}^2(\mathcal{L})$, \eqref{NtgeqNr} can be met by letting 
    \begin{align}
        &\tan(\gamma_{\mathsf{t}}-\bar{\gamma}_{\mathsf{t},\mathsf{x}}) = \sqrt{\left(\frac{D_{\mathsf{t},\mathsf{x}}^{\mathsf{R}}}{D_{\mathsf{t}}}\right)^2-1} \label{Dcx2Orienta} \\
        &\tan\gamma_{\mathsf{r}} = \tan(\gamma_{\mathsf{x}}(D_{\mathsf{t}})), \label{Dcx2Orientb} \\
        &\psi_{\mathsf{t}} = \frac{\pi}{2}, ~ \textup{and}~ \sin\psi_{\mathsf{r}} = \frac{D_{\mathsf{r}}}{D_{\mathsf{r},\mathsf{x}}^{\mathsf{R}}|\cos(\gamma_{\mathsf{x}}(D_{\mathsf{t}})-\bar{\gamma}_{\mathsf{r},\mathsf{x}})|}. \label{Dcx2Orientd}
    \end{align}
    $\mathcal{D}_{\mathsf{c},\mathsf{y}}(\mathcal{L})$, $\mathcal{D}_{\mathsf{c},\mathsf{y}}^1(\mathcal{L})$, and $\mathcal{D}_{\mathsf{c},\mathsf{y}}^2(\mathcal{L})$ are defined by swapping the subscripts $\mathsf{x}$ and $\mathsf{y}$ in the definitions of $\mathcal{D}_{\mathsf{c},\mathsf{x}}(\mathcal{L})$, $\mathcal{D}_{\mathsf{c},\mathsf{x}}^1(\mathcal{L})$, and $\mathcal{D}_{\mathsf{c},\mathsf{x}}^2(\mathcal{L})$, respectively. In $\mathcal{D}_{\mathsf{c},\mathsf{y}}^1(\mathcal{L})$ and $\mathcal{D}_{\mathsf{c},\mathsf{y}}^2(\mathcal{L})$, \eqref{NtgeqNr} can be met by swapping the subscripts $\mathsf{x}$ and $\mathsf{y}$ in \eqref{Dcx1Orienta}-\eqref{Dcx1Orientb} and \eqref{Dcx2Orienta}-\eqref{Dcx2Orientd}, respectively.
\end{prop}

\begin{figure}[t]
    \centering
    \includegraphics[width = .5\linewidth]{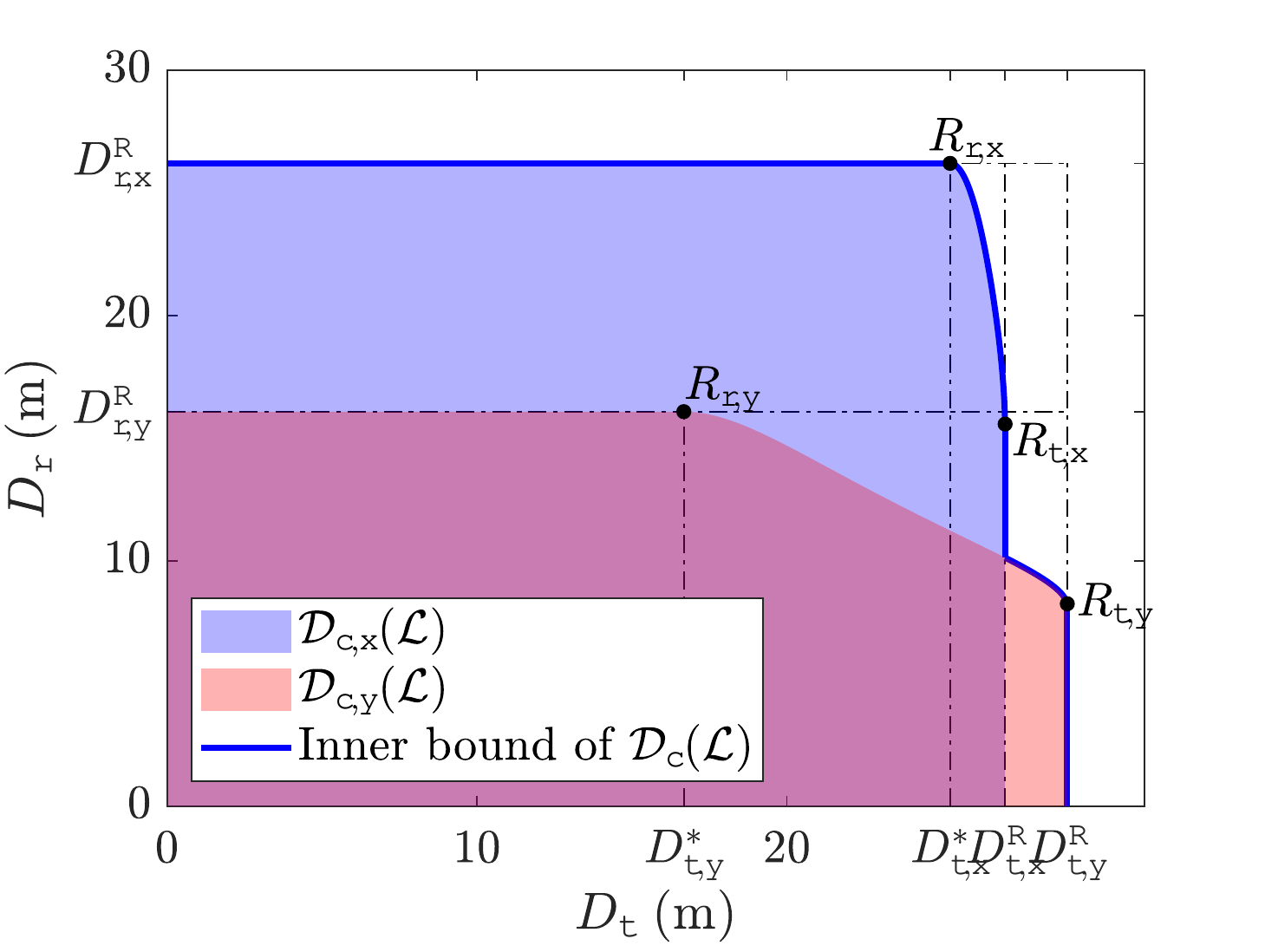}
    \caption{Predicted inner bound of the FMR of a general cascaded LoS MIMO channel. $N_{\mathsf{t}}=N_{\mathsf{r}}=5$, $Q_{\mathsf{x}} = Q_{\mathsf{y}} =15$, $d_{\mathsf{t}}=S_{\mathsf{x}}=S_{\mathsf{y}}=0.1$m, $\lambda=0.005$m, $\omega_{\mathsf{t}}=\frac{7\pi}{6}$, $\varphi_{\mathsf{t}}=\frac{\pi}{6}$, $\omega_{\mathsf{r}}=\frac{\pi}{3}$, and $\varphi_{\mathsf{r}}=\frac{3\pi}{7}$.}
    \label{FMRGeneral}
\end{figure}

The proof of Proposition \ref{propGray} can be found in Appendix \ref{proofGray}. Proposition \ref{propGray} generally gives an inner bound of the FMR based on the PB of reflective focusing. Fig. \ref{FMRGeneral} illustrates such a inner bound of the FMR of a general cascaded LoS MIMO channel with $N_{\mathsf{t}}=N_{\mathsf{r}}=5$, $Q_{\mathsf{x}} = Q_{\mathsf{y}} =15$, $d_{\mathsf{t}}=S_{\mathsf{x}}=S_{\mathsf{y}}=0.1$m, $\lambda=0.005$m, $\omega_{\mathsf{t}}=\frac{7\pi}{6}$, $\varphi_{\mathsf{t}}=\frac{\pi}{6}$, $\omega_{\mathsf{r}}=\frac{\pi}{3}$, and $\varphi_{\mathsf{r}}=\frac{3\pi}{7}$. $\mathcal{D}_{\mathsf{c},\mathsf{x}}(\mathcal{L})$ and $\mathcal{D}_{\mathsf{c},\mathsf{y}}(\mathcal{L})$ are colored in light blue and pink, respectively. The inner bound of the FMR predicted by Proposition \ref{propGray}, i.e., the boundary of $\mathcal{D}_{\mathsf{c},\mathsf{x}}(\mathcal{L}) \cup \mathcal{D}_{\mathsf{c},\mathsf{y}}(\mathcal{L})$ is marked by blue solid lines. $D_{\mathsf{t}}=D_{\mathsf{t},\mathsf{x}}^{\mathsf{R}}$, $D_{\mathsf{t}}=D_{\mathsf{t},\mathsf{y}}^{\mathsf{R}}$, $D_{\mathsf{t}}=D_{\mathsf{t},\mathsf{x}}^*$, $D_{\mathsf{t}}=D_{\mathsf{t},\mathsf{y}}^*$, $D_{\mathsf{r}}=D_{\mathsf{r},\mathsf{x}}^{\mathsf{R}}$, and $D_{\mathsf{r}}=D_{\mathsf{r},\mathsf{y}}^{\mathsf{R}}$ are marked by black dot-dash lines. We show that in such a general case, the region inside the predicted inner bound of the FMR is smaller than the rectangular area defined by the two single-hop Rayleigh distances, i.e., $\left\{(D_{\mathsf{t}},D_{\mathsf{r}})|0<D_{\mathsf{t}}\leq D_{\mathsf{t}}^{\mathsf{R}}, 0<D_{\mathsf{r}} \leq D_{{\mathsf{r}}}^{\mathsf{R}} \right\}$.

We have the following remark on Proposition \ref{propGray}.
\begin{remark}
    Denote by $R_{\mathsf{r},\mathsf{x}}$ the intersection of the curve $D_{\mathsf{r}}=\mathcal{B}_{\mathsf{c},\mathsf{x}}(D_{\mathsf{t}})$ and the vertical line $D_{\mathsf{t}} = D_{\mathsf{t},\mathsf{x}}^*$, and by $R_{\mathsf{t},\mathsf{x}}$ the intersection of the curve $D_{\mathsf{r}}=\mathcal{B}_{\mathsf{c},\mathsf{x}}(D_{\mathsf{t}})$ and the vertical line $D_{\mathsf{t}} = D_{\mathsf{t},\mathsf{x}}^{\mathsf{R}}$. Then, by plugging $D_{\mathsf{t}} = D_{\mathsf{t},\mathsf{x}}^*$ and $D_{\mathsf{t}} = D_{\mathsf{t},\mathsf{x}}^{\mathsf{R}}$ into \eqref{Bcf2}, we have $R_{\mathsf{r},\mathsf{x}}=(D_{\mathsf{t},\mathsf{x}}^*,D_{\mathsf{r},\mathsf{x}}^{\mathsf{R}})$, and $R_{\mathsf{t},\mathsf{x}}=(D_{\mathsf{t},\mathsf{x}}^{\mathsf{R}},D_{\mathsf{r},\mathsf{x}}^*)$ with $D_{\mathsf{r},\mathsf{x}}^*$ given by 
    \begin{align}
        D_{\mathsf{r},\mathsf{x}}^* &\triangleq D_{\mathsf{r},\mathsf{x}}^{\mathsf{R}}|\cos(\gamma_{\mathsf{r},\mathsf{x}}^*-\bar{\gamma}_{\mathsf{r},\mathsf{x}})|~ \textup{with}\\
        \gamma_{\mathsf{r},\mathsf{x}}^* &\triangleq \arctan \left(\frac{A_{\mathsf{r},\mathsf{y}}A_{\mathsf{t},\mathsf{x}} \cos \bar{\gamma}_{\mathsf{r},\mathsf{y}}-A_{\mathsf{r},\mathsf{x}} A_{\mathsf{t},\mathsf{y}} \cos(\bar{\gamma}_{\mathsf{t},\mathsf{x}}-\bar{\gamma}_{\mathsf{t},\mathsf{y}}) \cos \bar{\gamma}_{\mathsf{r},\mathsf{x}}}{A_{\mathsf{r},\mathsf{x}}A_{\mathsf{t},\mathsf{y}}\cos(\bar{\gamma}_{\mathsf{t},\mathsf{x}}-\bar{\gamma}_{\mathsf{t},\mathsf{y}}) \sin \bar{\gamma}_{\mathsf{r},\mathsf{x}}-A_{\mathsf{r},\mathsf{y}} A_{\mathsf{t},\mathsf{x}}  \sin \bar{\gamma}_{\mathsf{r},\mathsf{y}}} \right).
    \end{align}
    Similarly, $R_{\mathsf{r},\mathsf{y}}$, $R_{\mathsf{t},\mathsf{y}}$, and $D_{\mathsf{r},\mathsf{y}}^*$ can be defined by swapping the subscripts $\mathsf{x}$ and $\mathsf{y}$ in the definitions of $R_{\mathsf{r},\mathsf{x}}$, $R_{\mathsf{t},\mathsf{x}}$, and $D_{\mathsf{r},\mathsf{x}}^*$, respectively. $R_{\mathsf{r},\mathsf{x}}$, $R_{\mathsf{t},\mathsf{x}}$, $R_{\mathsf{r},\mathsf{y}}$, and $R_{\mathsf{t},\mathsf{y}}$ are marked in Fig. \ref{FMRGeneral}. With the IRS employing the reflective focusing PB strategy, $R_{\mathsf{r},\mathsf{x}}$ characterizes the largest $D_{\mathsf{t}}$ for ensuring full multiplexing when $D_{\mathsf{r}}=D_{\mathsf{r},\mathsf{x}}^{\mathsf{R}}$; $R_{\mathsf{t},\mathsf{x}}$ characterizes the largest $D_{\mathsf{r}}$ for ensuring full multiplexing when $D_{\mathsf{t}}=D_{\mathsf{t},\mathsf{x}}^{\mathsf{R}}$; $R_{\mathsf{r},\mathsf{y}}$ characterizes the largest $D_{\mathsf{t}}$ for ensuring full multiplexing when $D_{\mathsf{r}}=D_{\mathsf{r},\mathsf{y}}^{\mathsf{R}}$; $R_{\mathsf{t},\mathsf{y}}$ characterizes the largest $D_{\mathsf{r}}$ for ensuring full multiplexing when $D_{\mathsf{t}}=D_{\mathsf{t},\mathsf{y}}^{\mathsf{R}}$.
\end{remark}

We have the following corollary on Proposition \ref{propGray}.

\begin{figure}[t]
    \centering
    \begin{subfigure}{.46\linewidth}
        \centering
        \includegraphics[height=4.8cm]{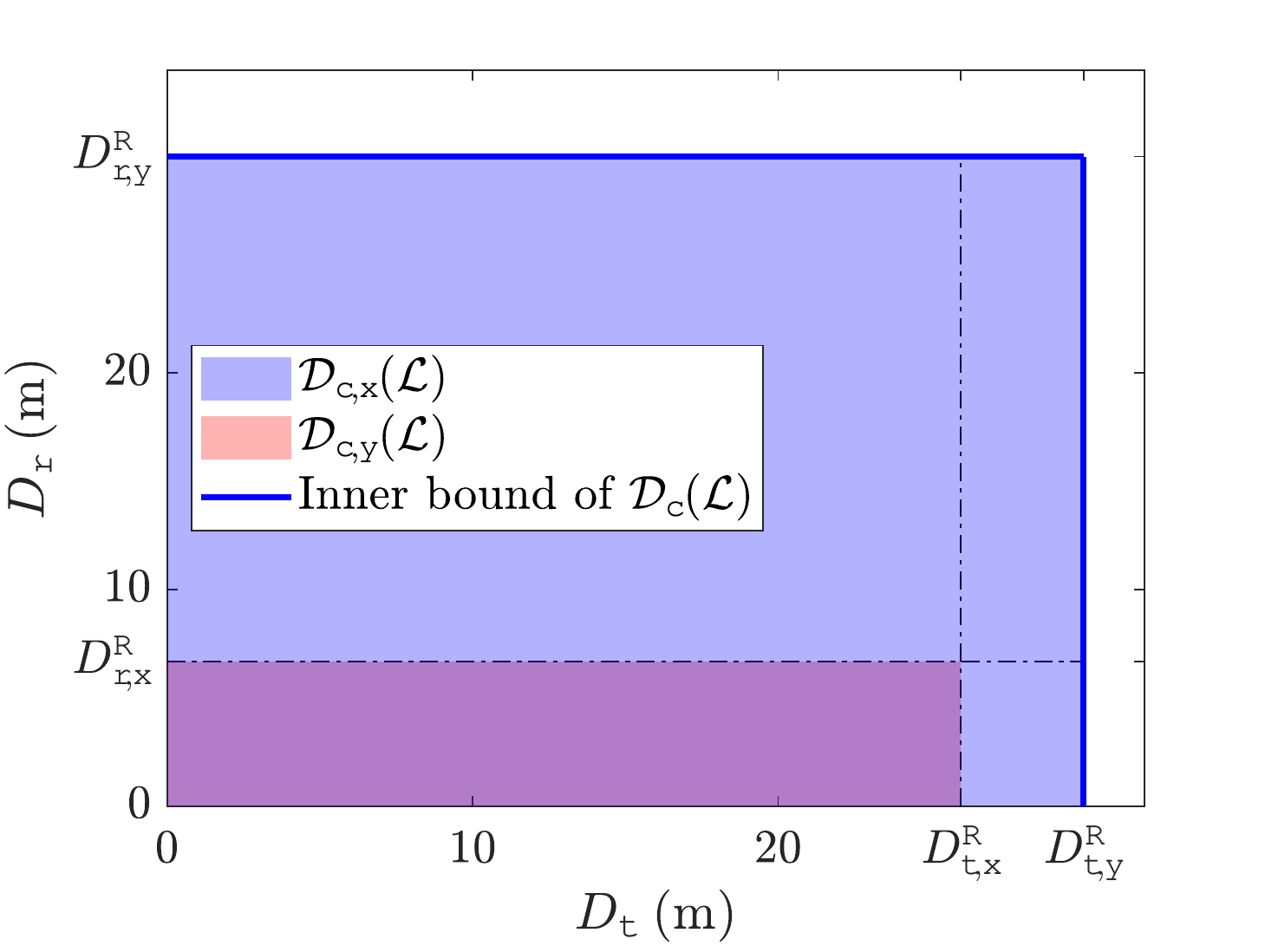}
        \caption{}
        \label{Figxza}
    \end{subfigure}
    \begin{subfigure}{.46\linewidth}
        \centering
        \includegraphics[height=4.8cm]{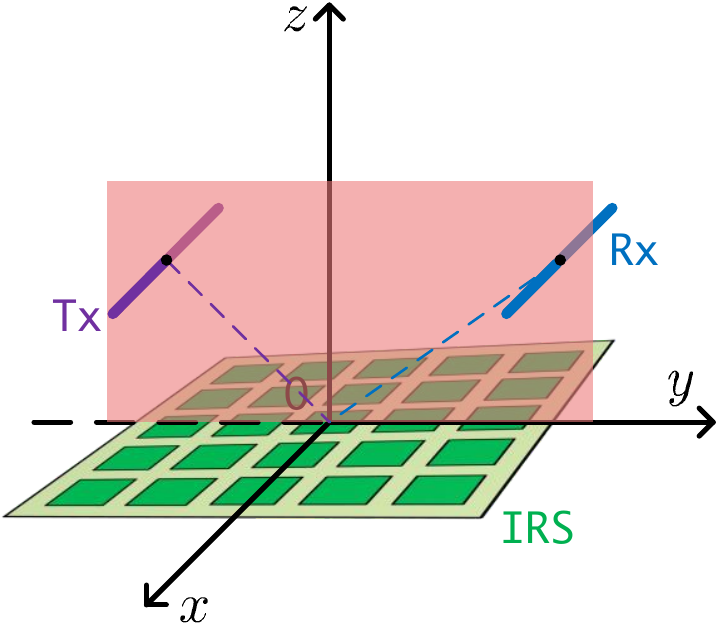}
        \caption{}
        \label{Figxzb}
    \end{subfigure}
    \caption{(a) Predicted inner bound of the FMR when both the centers of the Tx and the Rx are located on the $y$-$z$ plane. (b) Tx/Rx orientations to realize full multiplexing with $(D_{\mathsf{t}},D_{\mathsf{r}})=(D_{\mathsf{t},\mathsf{y}}^{\mathsf{R}},D_{\mathsf{r},\mathsf{y}}^{\mathsf{R}})$.}
    \label{Figxz}
\end{figure}

\begin{coro}\label{CORO1}
    When both the centers of the Tx and the Rx are located on the $y$-$z$ or $x$-$z$ planes, i.e., $\omega_{\mathsf{t}}\in \left\{0,\frac{\pi}{2},\pi,\frac{3\pi}{2}\right\}$ and $\omega_{\mathsf{r}} \in \left\{0,\frac{\pi}{2},\pi,\frac{3\pi}{2}\right\}$, we have
    \begin{equation}\label{DtxStarEqDtxR}
        D_{\mathsf{t},\mathsf{x}}^* = D_{\mathsf{t},\mathsf{x}}^{\mathsf{R}} ~ \textup{and} ~ D_{\mathsf{t},\mathsf{y}}^* = D_{\mathsf{t},\mathsf{y}}^{\mathsf{R}},
    \end{equation}
\end{coro}
\begin{proof}
    First consider $\omega_{\mathsf{t}}=0$ and $\omega_{\mathsf{r}}=\frac{\pi}{2}$. By plugging $\omega_{\mathsf{t}}=0$ and $\omega_{\mathsf{r}}=\frac{\pi}{2}$ into \eqref{GammatxStarb}, we have $\tan \gamma_{\mathsf{t},\mathsf{x}}^* = \infty$ and thus $\gamma_{\mathsf{t},\mathsf{x}}^* = \frac{\pi}{2}$ or $\frac{3\pi}{2}$. By plugging $\omega_{\mathsf{t}}=0$ into \eqref{Atx}, we have $\bar{\gamma}_{\mathsf{t},\mathsf{x}}=\frac{\pi}{2}$ or $\frac{3\pi}{2}$. Therefore, $|\cos(\gamma_{\mathsf{t},\mathsf{x}}^*-\bar{\gamma}_{\mathsf{t},\mathsf{x}})|=1$, and thus $D_{\mathsf{t},\mathsf{x}}^*=D_{\mathsf{t},\mathsf{x}}^{\mathsf{R}}$. Similarly, with $\omega_{\mathsf{t}}=0$ and $\omega_{\mathsf{r}}=\frac{\pi}{2}$, we have $\gamma_{\mathsf{t},\mathsf{y}}^* = 0$ or $\pi$, and $\bar{\gamma}_{\mathsf{t},\mathsf{y}}=0$ or $\pi$. Therefore, $|\cos(\gamma_{\mathsf{t},\mathsf{y}}^*-\bar{\gamma}_{\mathsf{t},\mathsf{y}})|=1$, and thus $D_{\mathsf{t},\mathsf{y}}^*=D_{\mathsf{t},\mathsf{y}}^{\mathsf{R}}$. For other $\omega_{\mathsf{t}}\in \left\{0,\frac{\pi}{2},\pi,\frac{3\pi}{2}\right\}$ and $\omega_{\mathsf{r}} \in \left\{0,\frac{\pi}{2},\pi,\frac{3\pi}{2}\right\}$, \eqref{DtxStarEqDtxR} can be proved in a similar way. 
\end{proof}
Corollary \ref{CORO1} means that when both the centers of the Tx and the Rx are located on the $y$-$z$ or $x$-$z$ planes, $R_{\mathsf{r},\mathsf{x}}$ and $R_{\mathsf{t},\mathsf{x}}$ merge into a single point $(D_{\mathsf{t}},D_{\mathsf{r}})=(D_{\mathsf{t},\mathsf{x}}^{\mathsf{R}},D_{\mathsf{r},\mathsf{x}}^{\mathsf{R}})$, and $R_{\mathsf{r},\mathsf{y}}$ and $R_{\mathsf{t},\mathsf{y}}$ merge into a single point $(D_{\mathsf{t}},D_{\mathsf{r}})=(D_{\mathsf{t},\mathsf{y}}^{\mathsf{R}},D_{\mathsf{r},\mathsf{y}}^{\mathsf{R}})$. Then, $\mathcal{D}_{\mathsf{c},\mathsf{x}}(\mathcal{L})$ and $\mathcal{D}_{\mathsf{c},\mathsf{y}}(\mathcal{L})$ are two rectangular areas on the $D_{\mathsf{t}}$-$D_{\mathsf{r}}$ plane. 
\blue{
Two special cases are illustrated in Fig. \ref{Figxz} and Fig. \ref{Figxyz} with $N_{\mathsf{t}}$, $N_{\mathsf{r}}$, $Q_{\mathsf{x}}$, $Q_{\mathsf{y}}$, $d_{\mathsf{t}}$, $S_{\mathsf{x}}$, $S_{\mathsf{y}}$, and $\lambda$ same as in Fig. \ref{FMRGeneral}. Fig. \ref{Figxza} illustrates the predicted inner bound of the FMR with $\omega_{\mathsf{t}}=\frac{3\pi}{2}$, $\omega_{\mathsf{r}}=\frac{\pi}{2}$, i.e., the Tx and the Rx are centered on the $y$-$z$ plane. From Corollary \ref{CORO1}, $\mathcal{D}_{\mathsf{c},\mathsf{x}}(\mathcal{L})$ and $\mathcal{D}_{\mathsf{c},\mathsf{y}}(\mathcal{L})$ are two rectangular areas on the $D_{\mathsf{t}}$-$D_{\mathsf{r}}$ plane, and $\mathcal{D}_{\mathsf{c},\mathsf{y}}(\mathcal{L})$ is contained in $\mathcal{D}_{\mathsf{c},\mathsf{x}}(\mathcal{L})$. Furthermore, $\mathcal{D}_{\mathsf{c},\mathsf{x}}(\mathcal{L})$ coincides with the rectangular area defined by the Rayleigh distances of the Tx-IRS and the IRS-Rx channels, i.e, $\left\{(D_{\mathsf{t}},D_{\mathsf{r}})|0<D_{\mathsf{t}}\leq D_{\mathsf{t}}^{\mathsf{R}}=D_{\mathsf{t},\mathsf{y}}^{\mathsf{R}}, 0<D_{\mathsf{r}} \leq D_{{\mathsf{r}}}^{\mathsf{R}}=D_{\mathsf{r},\mathsf{y}}^{\mathsf{R}} \right\}$. At the top-right corner of $\mathcal{D}_{\mathsf{c},\mathsf{x}}(\mathcal{L})$, by plugging $\omega_{\mathsf{t}}=\frac{3\pi}{2}$, $\omega_{\mathsf{r}}=\frac{\pi}{2}$, $D_{\mathsf{t}}=D_{\mathsf{t},\mathsf{x}}^{\mathsf{R}}$, and $D_{\mathsf{r}}=D_{\mathsf{r},\mathsf{x}}^{\mathsf{R}}$ into \eqref{Dcx1Orienta}-\eqref{Dcx1Orientb}, we have $\gamma_{\mathsf{t}}=0$ or $\pi$, $\gamma_{\mathsf{r}}=0$ or $\pi$, $\psi_{\mathsf{t}}=\psi_{\mathsf{r}}=\frac{\pi}{2}$. This means that the Tx and the Rx are parallel to the $x$-axis, as illustrated in Fig. \ref{Figxzb}. 
Fig. \ref{Figxyza} illustrates the predicted inner bound of the FMR with $\omega_{\mathsf{t}}=0$, $\omega_{\mathsf{r}}=\frac{\pi}{2}$, i.e., the Tx and the Rx are centered on the $x$-$z$ and $y$-$z$ planes, respectively. We note that $\mathcal{D}_{\mathsf{c},\mathsf{x}}(\mathcal{L})$ and $\mathcal{D}_{\mathsf{c},\mathsf{y}}(\mathcal{L})$ are overlapped rectangular areas, and the union of $\mathcal{D}_{\mathsf{c},\mathsf{x}}(\mathcal{L})$ and $\mathcal{D}_{\mathsf{c},\mathsf{y}}(\mathcal{L})$ is smaller than the rectangular area defined by the Rayleigh distances of the Tx-IRS and the IRS-Rx channels, i.e, $\left\{(D_{\mathsf{t}},D_{\mathsf{r}})|0<D_{\mathsf{t}}\leq D_{\mathsf{t}}^{\mathsf{R}}=D_{\mathsf{t},\mathsf{y}}^{\mathsf{R}}, 0<D_{\mathsf{r}} \leq D_{{\mathsf{r}}}^{\mathsf{R}}=D_{\mathsf{r},\mathsf{x}}^{\mathsf{R}} \right\}$. At the top-right corner of $\mathcal{D}_{\mathsf{c},\mathsf{x}}(\mathcal{L})$, by plugging $\omega_{\mathsf{t}}=0$, $\omega_{\mathsf{r}}=\frac{\pi}{2}$, $D_{\mathsf{t}}=D_{\mathsf{t},\mathsf{x}}^{\mathsf{R}}$, and $D_{\mathsf{r}}=D_{\mathsf{r},\mathsf{x}}^{\mathsf{R}}$ into \eqref{Dcx1Orienta}-\eqref{Dcx1Orientb}, we have $\gamma_{\mathsf{t}}=0$ or $\pi$, $\gamma_{\mathsf{r}}=\frac{\pi}{2}$ or $\frac{3\pi}{2}$, $\psi_{\mathsf{t}}=\psi_{\mathsf{r}}=\frac{\pi}{2}$. This means that the Tx lie on the $x$-$z$ plane and is perpendicular to the line that passes the origin and the center of the Tx, and the Rx is parallel to the $x$-axis, as illustrated in Fig. \ref{Figxyzb}.
}

\begin{figure}[t]
    \centering
    \begin{subfigure}{.46\linewidth}
        \centering
        \includegraphics[height=4.8cm]{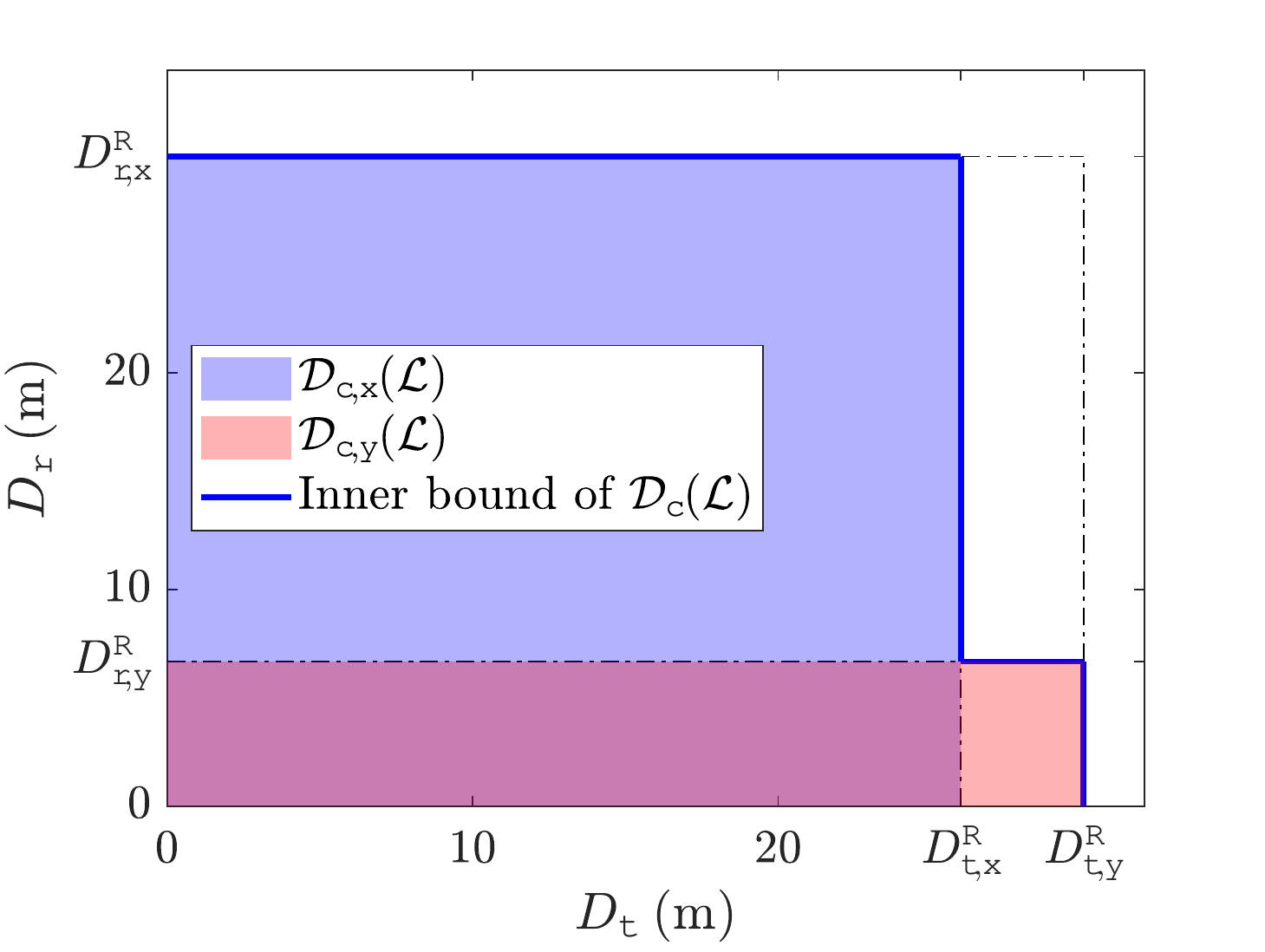}
        \caption{}
        \label{Figxyza}
    \end{subfigure}
    \begin{subfigure}{.46\linewidth}
        \centering
        \includegraphics[height=4.8cm]{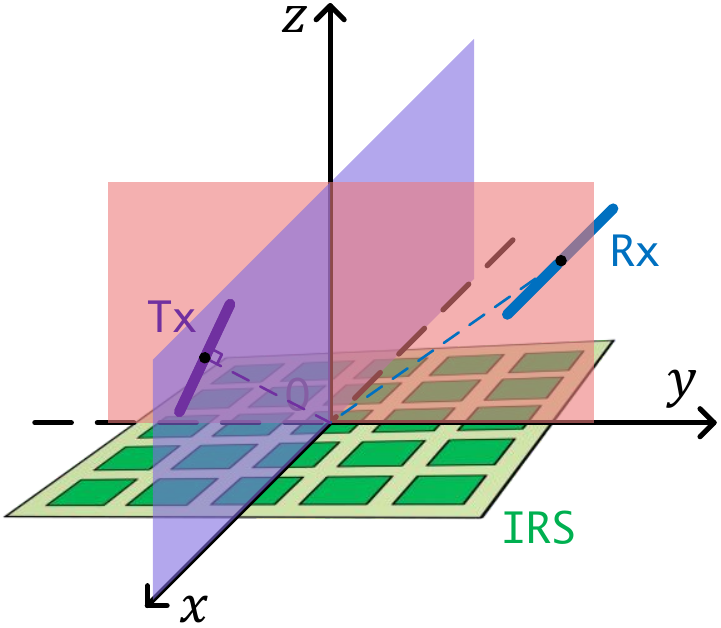}
        \caption{}
        \label{Figxyzb}
    \end{subfigure}
    \caption{(a) Predicted inner bound of the FMR when the Tx is centered on the $x$-$z$ plane and the Rx is centered on the $y$-$z$ plane. (b) Tx/Rx orientations to realize full multiplexing with $(D_{\mathsf{t}},D_{\mathsf{r}})=(D_{\mathsf{t},\mathsf{x}}^{\mathsf{R}},D_{\mathsf{r},\mathsf{x}}^{\mathsf{R}})$.}
    \label{Figxyz}
\end{figure}

\section{Conclusion}\label{SecConclusion}
In this paper, we developed an IRS-aided MIMO channel model with cascaded LoS communication links, and analyzed the capacity of the proposed cascaded LoS MIMO system. 
When modeling the reflection of the incident wave at the IRS, the curvature of the wavefront on different REs is taken into account. Based on the proposed channel model, we studied the spatial multiplexing capability of the cascaded LoS MIMO channel. To measure the spatial multiplexing capability of the cascaded LoS MIMO channel with the Tx-IRS and IRS-Rx distances, we introduced the notion of FMR for the cascaded LoS MIMO channel. Based on a special PB strategy named reflective focusing, we derived an inner bound of the FMR, and provided the orientation settings of the antenna arrays to enable full multiplexing (between Tx and Rx) in the region defined by such an inner bound. Part II of this work considers the mutual information maximization of the cascaded LoS MIMO systems over the PB and the array orientations.

\appendices

\section{Proof of Proposition \ref{propSingleSideDray}}\label{proofSingleSideDray}
It suffices to prove the Rayleigh distance on the Tx-IRS side. The Rayleigh distance on the IRS-Rx side can be proved in a similar manner. By recalling $h_{(k,l),p} = e^{-j\zeta_{(k,l),p}}$ and $\zeta_{(k,l),p} = \frac{2\pi d_{(k,l),p}}{\lambda}$, $\langle \boldsymbol{h}_{\mathsf{t};p_1}, \boldsymbol{h}_{\mathsf{t};p_2} \rangle$ in \eqref{TOrthCrtt} can be rewritten as
    \begin{align}
        \langle \boldsymbol{h}_{\mathsf{t};p_1}, \boldsymbol{h}_{\mathsf{t};p_2} \rangle &= \sum_{k = -\frac{Q_{\mathsf{x}}-1}{2}}^{\frac{Q_{\mathsf{x}}-1}{2}} \sum_{l = -\frac{Q_{\mathsf{y}}-1}{2}}^{\frac{Q_{\mathsf{y}}-1}{2}} e^{\frac{j2\pi}{\lambda}\left( d_{(k,l),p_2} - d_{(k,l),p_1} \right)}, \label{IPhtp1htp2}
    \end{align}
$p_1 \in \mathcal{I}_{N_{\mathsf{t}}},~ p_2 \in \mathcal{I}_{N_{\mathsf{t}}}$.
By plugging \eqref{Approxklp} and \eqref{Approxqkl} in \eqref{IPhtp1htp2}, and defining
\begin{align}
    C_{p_1,p_2} &\triangleq \exp\left\{ \frac{j2\pi}{\lambda}\left( (p_2-p_1)d_t\cos\psi_{\mathsf{t}}+ \frac{(p_2^2-p_1^2)d_t^2\sin^2 \psi_{\mathsf{t}}}{2D_t} \right)\right\}
\end{align}
we have
    \begin{align}\label{ip}
        \langle \boldsymbol{h}_{p_1},\boldsymbol{h}_{p_2} \rangle &= C_{p_1,p_2} A_{p_1,p_2}^{\mathsf{x}} A_{p_1,p_2}^{\mathsf{y}}, ~ p_1 \in \mathcal{I}_{N_{\mathsf{t}}},~ p_2 \in \mathcal{I}_{N_{\mathsf{t}}},
    \end{align}
where $A_{p_1,p_2}^{\mathsf{x}} \triangleq \frac{\sin(\pi (p_1-p_2) C_{\mathsf{t},\mathsf{x}})}{\sin (\frac{\pi}{Q_{\mathsf{x}}} (p_1-p_2) C_{\mathsf{t},\mathsf{x}})}$ and $A_{p_1,p_2}^{\mathsf{y}} \triangleq \frac{\sin(\pi (p_1-p_2) C_{\mathsf{t},\mathsf{y}})}{ \sin (\frac{\pi}{Q_{\mathsf{y}}} (p_1-p_2) C_{\mathsf{t},\mathsf{y}})}$, with $C_{\mathsf{t},\mathsf{x}}$ and $C_{\mathsf{r},\mathsf{y}}$ defined in \eqref{Ctx}. 

From \eqref{ip}, for $\boldsymbol{H}_{\mathsf{t}}$ to meet the requirements in \eqref{TOrthCrtt}, a necessary and sufficient condition is 
\begin{subequations}\label{AxAycond}
    \begin{align}
        A_{p_1,p_2}^{\mathsf{x}} &= Q_{\mathsf{x}} ~ \textup{and}~ A_{p_1,p_2}^{\mathsf{y}}=Q_{\mathsf{y}},~  \forall p_1 = p_2, \label{AxAyconda}\\
        A_{p_1,p_2}^{\mathsf{x}} &= 0 ~ \textrm{or} ~ A_{p_1,p_2}^{\mathsf{y}}=0,~  \forall p_1 \neq p_2 \label{AxAycondb}
    \end{align}
\end{subequations}
with $p_1 \in \mathcal{I}_{N_{\mathsf{t}}},~ p_2 \in \mathcal{I}_{N_{\mathsf{t}}}$. When $p_1=p_2$, \eqref{AxAyconda} is naturally satisfied. To meet \eqref{AxAycondb}, a necessary and sufficient condition is
\begin{align}
    |C_{\mathsf{t},\mathsf{x}}| &\in \mathbb{Z}^+ ~ \textup{and} ~ \frac{C_{\mathsf{t},\mathsf{x}}(p_1-p_2)}{Q_{\mathsf{x}}} \notin \mathbb{Z},~ \forall p_1 \neq p_2,~ p_1,p_2 \in \mathcal{I}_{N_{\mathsf{t}}},~ \textrm{or} \label{CtxLim} \\
    |C_{\mathsf{t},\mathsf{y}}| &\in \mathbb{Z}^+ ~ \textup{and} ~ \frac{C_{\mathsf{t},\mathsf{y}}(p_1-p_2)}{Q_{\mathsf{y}}} \notin \mathbb{Z},~ \forall p_1 \neq p_2,~ p_1,p_2 \in \mathcal{I}_{N_{\mathsf{t}}} \label{CtyLim}
\end{align}
where $\mathbb{Z}^+$ and $\mathbb{Z}$ are the sets of positive integers and all integers, respectively. We first consider how to meet \eqref{CtxLim}. By recalling that $C_{\mathsf{t},\mathsf{x}}=\frac{d_{\mathsf{t}}S_{\mathsf{x}}Q_{\mathsf{x}}A_{\mathsf{t},\mathsf{x}}\sin\psi_{\mathsf{t}}\cos(\gamma_{\mathsf{t}}-\bar{\gamma}_{\mathsf{t},\mathsf{x}})}{\lambda D_{\mathsf{t}}}$ is inversely proportional to $D_{\mathsf{t}}$ in \eqref{Ctx}, we note that to ensure $|C_{\mathsf{t},\mathsf{x}}| \in \mathbb{Z}^+$, $D_{\mathsf{t}}$ is upper-bounded by $D_{\mathsf{t},\mathsf{x}}^{\mathsf{R}}\triangleq \frac{d_{\mathsf{t}}S_{\mathsf{x}}Q_{\mathsf{x}}A_{\mathsf{t},\mathsf{x}}}{\lambda}$. When $D_{\mathsf{t}} \leq D_{\mathsf{t},\mathsf{x}}^{\mathsf{R}}$, we can always ensure $|C_{\mathsf{t},\mathsf{x}}|=1$ by letting $\psi_{\mathsf{t}}$ and $\gamma_{\mathsf{t}}$ satisfy \eqref{txangleCrt}. Otherwise, when $D_{\mathsf{t}}> D_{\mathsf{t},\mathsf{x}}^{\mathsf{R}}$, $|C_{\mathsf{t},\mathsf{x}}|$ is always less than $1$. Futhermore, since $(p_1-p_2) \in \left\{-N_{\mathsf{t}}+1,\ldots,N_{\mathsf{t}}-1\right\}$,
    $|C_{\mathsf{t},\mathsf{x}}|=1 ~\textup{and}~ Q_{\mathsf{x}} \geq N_{\mathsf{t}}$
is a sufficient condition for $\frac{C_{\mathsf{t},\mathsf{x}}(p_1-p_2)}{Q_{\mathsf{x}}} \notin \mathbb{Z}$. Thus when $Q_{\mathsf{x}} \geq N_{\mathsf{t}}$ and $D_{\mathsf{t}} \leq D_{\mathsf{t},\mathsf{x}}^{\mathsf{R}}$, \eqref{txangleCrt} is a sufficient condition for \eqref{CtxLim}, as well as for the orthogonality and equal-gain requirement in \eqref{TOrthCrtt}. Similarly, to meet \eqref{CtyLim}, $D_{\mathsf{t}}$ is upper-bounded by $D_{\mathsf{t},\mathsf{y}}^{\mathsf{R}} \triangleq \frac{d_{\mathsf{t}}S_{\mathsf{y}}Q_{\mathsf{y}}A_{\mathsf{t},\mathsf{y}}}{\lambda}$. When $Q_{\mathsf{y}} \geq N_{\mathsf{t}}$ and $D_{\mathsf{t}} \leq D_{\mathsf{t},\mathsf{y}}^{\mathsf{R}}$, \eqref{tyangleCrt} is a sufficient condition for \eqref{CtyLim}, and also for the orthogonality and equal-gain requirement in \eqref{TOrthCrtt}. 

Based on the above, when $Q_{\mathsf{x}} \geq N_{\mathsf{t}}$ and $Q_{\mathsf{y}} \geq N_{\mathsf{t}}$, the Rayleigh distance of the Tx-IRS channel $D_{\mathsf{t}}^{\mathsf{R}}$ is given by $\max\{D_{\mathsf{t},\mathsf{x}}^{\mathsf{R}},D_{\mathsf{t},\mathsf{y}}^{\mathsf{R}}\}$. When $D_{\mathsf{t}} \leq D_{\mathsf{t}}^{\mathsf{R}}$, the orthogonality and equal-gain requirement in \eqref{TOrthCrtt} can be met by letting $\psi_{\mathsf{t}}$ and $\gamma_{\mathsf{t}}$ satisfy \eqref{txangleCrt} or \eqref{tyangleCrt}.

\section{Proof of Proposition \ref{propGray}}\label{proofGray}
We prove Proposition \ref{propGray} based on reflective focusing. Recall in \eqref{HPQ} that, with reflective focusing, the overall channel matrix $\boldsymbol{H}$ can be expressed by the Hadamard product of two matrices $\boldsymbol{P}$ and $\boldsymbol{Q}$ apart from the path loss factor $\eta_0$, with the elements of $\boldsymbol{P}$ and $\boldsymbol{Q}$, denoted by $P_{q,p}$ and $Q_{q,p}$, given in \eqref{Pqpfinal} and \eqref{Qqpfinal}, respectively, $p \in \mathcal{I}_{N_{\mathsf{t}}},~ q \in \mathcal{I}_{N_{\mathsf{r}}}$. Let $Q_{q,p}^{\mathsf{x}}\triangleq\frac{\sin\left(\pi \left( C_{\mathsf{t},\mathsf{x}}p + C_{\mathsf{r},\mathsf{x}}q \right)\right)}{\sin\left(\frac{\pi}{Q_{\mathsf{x}}} \left( C_{\mathsf{t},\mathsf{x}}p + C_{\mathsf{r},\mathsf{x}}q \right)\right)}$ and $Q_{q,p}^{\mathsf{y}}\triangleq\frac{\sin\left(\pi \left( C_{\mathsf{t},\mathsf{y}}p + C_{\mathsf{r},\mathsf{y}}q \right)\right)}{\sin\left(\frac{\pi}{Q_{\mathsf{y}}} \left( C_{\mathsf{t},\mathsf{y}}p + C_{\mathsf{r},\mathsf{y}}q \right)\right)}$ with $C_{\mathsf{t},\mathsf{x}}$, $C_{\mathsf{t},\mathsf{y}}$, $C_{\mathsf{r},\mathsf{x}}$, and $C_{\mathsf{r},\mathsf{y}}$ defined in \eqref{C}. Then, from \eqref{HPQ}, we have $Q_{q,p}=Q_{q,p}^{\mathsf{x}}Q_{q,p}^{\mathsf{y}}$, $p \in \mathcal{I}_{N_{\mathsf{t}}}$, $q \in \mathcal{I}_{N_{\mathsf{r}}}$. When 
\begin{align}
    \textup{sign}\left(C_{\mathsf{t},\mathsf{x}}C_{\mathsf{r},\mathsf{x}}\right) &= \textup{sign}\left(C_{\mathsf{t},\mathsf{y}}C_{\mathsf{r},\mathsf{y}}\right), \label{CSign} \\
    |C_{\mathsf{t},\mathsf{x}}|&=|C_{\mathsf{r},\mathsf{x}}| = 1, \label{CEquala} \\
    |C_{\mathsf{t},\mathsf{y}}|&=|C_{\mathsf{r},\mathsf{y}}|=C_{\mathsf{y}} \label{CEqualb}
\end{align}
where $C_{\mathsf{y}}$ is an arbitrary constant, we obtain
\begin{align}
    Q_{q,p}^{\mathsf{x}} &= \frac{\sin\left(\pi \left( p + q \right)\right)}{\sin\left(\frac{\pi}{Q_{\mathsf{x}}} \left( p + q \right)\right)}~\textup{and}~Q_{q,p}^{\mathsf{y}} = \frac{\sin\left(\pi C_{\mathsf{y}} \left( p + q \right)\right)}{\sin\left(\frac{\pi C_{\mathsf{y}}}{Q_{\mathsf{y}}} \left(p + q \right)\right)},  ~p \in \mathcal{I}_{N_{\mathsf{t}}},~ q \in \mathcal{I}_{N_{\mathsf{r}}},~ \textrm{or} \label{pplusq1}\\
    Q_{q,p}^{\mathsf{x}} &= \frac{\sin\left(\pi \left( p - q \right)\right)}{\sin\left(\frac{\pi}{Q_{\mathsf{x}}} \left( p - q \right)\right)}~\textup{and}~Q_{q,p}^{\mathsf{y}} = \frac{\sin\left(\pi C_{\mathsf{y}} \left( p - q \right)\right)}{\sin\left(\frac{\pi C_{\mathsf{y}}}{Q_{\mathsf{y}}} \left(p - q \right)\right)},  ~p \in \mathcal{I}_{N_{\mathsf{t}}},~ q \in \mathcal{I}_{N_{\mathsf{r}}}. \label{pplusq2}
\end{align}
Particularly, $Q_{q,p}^{\mathsf{x}}=Q_{\mathsf{x}}$ and $Q_{q,p}^{\mathsf{y}}=Q_{\mathsf{y}}$ for $q=-p$ in \eqref{pplusq1}, or alternatively for $q=p$ in \eqref{pplusq2}, $p \in \mathcal{I}_{N_{\mathsf{t}}}$, $q \in \mathcal{I}_{N_{\mathsf{r}}}$.
Since $p+q$ and $p-q$ are integers, the numerators of $Q_{q,p}^{\mathsf{x}}$ in \eqref{pplusq1} and \eqref{pplusq2} are always zeros. Furthermore, when $N_{\mathsf{t}}+N_{\mathsf{r}}-2 < 2Q_{\mathsf{x}}$, we have $\frac{\left|p+q\right|}{Q_{\mathsf{x}}} < 1$ and $\frac{\left|p-q\right|}{Q_{\mathsf{x}}} < 1$. Then, the denominators of $Q_{q,p}^{\mathsf{x}}$ in \eqref{pplusq1} and \eqref{pplusq2} are non-zero except for the cases of $p-q=0$ and $p+q=0$, respectively. Based on the above, when \eqref{CSign}, \eqref{CEquala}, \eqref{CEqualb}, and $N_{\mathsf{t}}+N_{\mathsf{r}}-2 < 2Q_{\mathsf{x}}$ are met, the elements of $\boldsymbol{Q}$ are given by
\begin{equation}\label{Qqppp}
    Q_{q,p}=\left\{\begin{array}{ll}
        Q_{\mathsf{x}}Q_{\mathsf{y}}, & q = -p ~\textup{or}~ q=p \\
        0, & \rm{otherwise}
    \end{array}\right.
    ,~ p \in \mathcal{I}_{N_{\mathsf{t}}},~ q \in \mathcal{I}_{N_{\mathsf{r}}}.
\end{equation}
Correspondingly, the elements of $\boldsymbol{H}$ in \eqref{HPQ} are given by
\begin{equation}\label{tildeHqp}
    h_{q,p}=\left\{\begin{array}{ll}
        \eta_0 Q_{\mathsf{x}} Q_{\mathsf{y}} P_{q,p}, & q = -p ~\textup{or}~ q=p  \\
        0, & \rm{otherwise}
    \end{array}\right.
    ,~ p \in \mathcal{I}_{N_{\mathsf{t}}},~ q \in \mathcal{I}_{N_{\mathsf{r}}}.
\end{equation}
By noting $|P_{q,p}|=1$, we see that \eqref{CSign}, \eqref{CEquala}, \eqref{CEqualb}, and $N_{\mathsf{t}}+N_{\mathsf{r}}-2 < 2Q_{\mathsf{x}}$ give a sufficient condition of the orthogonality requirements in \eqref{NtgeqNr}. 
Thus, to prove Proposition \ref{propGray}, it suffices to show that for a cascaded LoS MIMO system with $N_{\mathsf{t}}+N_{\mathsf{r}}-2 < 2Q_{\mathsf{x}}$, in the region $\mathcal{D}_{\mathsf{c},\mathsf{x}}(\mathcal{L})$ (defined by \eqref{DcxStart} and \eqref{Dcx2}), \eqref{CSign}-\eqref{CEqualb} can be satisfied by adjusting $\psi_{\mathsf{t}}$, $\gamma_{\mathsf{t}}$, $\psi_{\mathsf{r}}$, and $\gamma_{\mathsf{r}}$. 

To this end, we see that \eqref{CSign} can be rewritten as
\begin{equation}\label{signtransform}
    \textup{sign}\left(\frac{C_{\mathsf{t},\mathsf{x}}}{C_{\mathsf{t},\mathsf{y}}}\right) = \textup{sign}\left(\frac{C_{\mathsf{r},\mathsf{x}}}{C_{\mathsf{r},\mathsf{y}}}\right).
\end{equation}
By plugging \eqref{CEquala} and \eqref{CEqualb} into \eqref{signtransform}, we obtain $\frac{C_{\mathsf{t},\mathsf{x}}}{C_{\mathsf{t},\mathsf{y}}} = \frac{C_{\mathsf{r},\mathsf{x}}}{C_{\mathsf{r},\mathsf{y}}}$. 
By plugging \eqref{Ctx} into $\frac{C_{\mathsf{t},\mathsf{x}}}{C_{\mathsf{t},\mathsf{y}}} = \frac{C_{\mathsf{r},\mathsf{x}}}{C_{\mathsf{r},\mathsf{y}}}$, we see that to meet \eqref{CSign}-\eqref{CEqualb}, a necessary condition for $\gamma_{\mathsf{t}}$ and $\gamma_{\mathsf{r}}$ is
\begin{align}\label{CosRatio}
    \frac{ \cos(\gamma_{\mathsf{r}}-\bar{\gamma}_{\mathsf{r},\mathsf{y}})}{\cos(\gamma_{\mathsf{r}}-\bar{\gamma}_{\mathsf{r},\mathsf{x}})} &= \frac{A_{\mathsf{t},\mathsf{y}}A_{\mathsf{r},\mathsf{x}}\cos(\gamma_{\mathsf{t}}-\bar{\gamma}_{\mathsf{t},\mathsf{y}})}{A_{\mathsf{t},\mathsf{x}}A_{\mathsf{r},\mathsf{y}} \cos(\gamma_{\mathsf{t}}-\bar{\gamma}_{\mathsf{t},\mathsf{x}})}. 
\end{align}
Next, based on \eqref{Ctx}, \eqref{Crx}, $D_{\mathsf{t},\mathsf{x}}^{\mathsf{R}} \triangleq \frac{d_{\mathsf{t}} S_{\mathsf{x}}Q_{\mathsf{x}}A_{\mathsf{t},\mathsf{x}}}{\lambda}$, and $D_{\mathsf{r},\mathsf{x}}^{\mathsf{R}} \triangleq \frac{d_{\mathsf{r}} S_{\mathsf{x}}Q_{\mathsf{x}}A_{\mathsf{r},\mathsf{x}}}{\lambda}$, we rewrite \eqref{CEquala} as
\begin{align}
    |C_{\mathsf{t},\mathsf{x}}|&=\frac{D_{\mathsf{t},\mathsf{x}}^{\mathsf{R}} \sin\psi_{\mathsf{t}}|\cos(\gamma_{\mathsf{t}}-\bar{\gamma}_{\mathsf{t},\mathsf{x}})|}{ D_{\mathsf{t}}}=1 ~\textup{and}~ \label{CtxrxEqual1a} \\
    |C_{\mathsf{r},\mathsf{x}}|&=\frac{D_{\mathsf{r},\mathsf{x}}^{\mathsf{R}}\sin\psi_{\mathsf{r}}|\cos(\gamma_{\mathsf{r}}-\bar{\gamma}_{\mathsf{r},\mathsf{x}})|}{ D_{\mathsf{r}}}=1. \label{CtxrxEqual1b}
\end{align}
The equation set \eqref{CSign}-\eqref{CEqualb} is equivalent to \eqref{CosRatio}-\eqref{CtxrxEqual1b}. To meet \eqref{CtxrxEqual1a} and \eqref{CtxrxEqual1b}, $D_{\mathsf{t}}$ and $D_{\mathsf{r}}$ are upper bounded by $D_{\mathsf{t},\mathsf{x}}^{\mathsf{R}}$ and $D_{\mathsf{t},\mathsf{y}}^{\mathsf{R}}$, respectively, i.e., $0< D_{\mathsf{t}} \leq D_{\mathsf{t},\mathsf{x}}^{\mathsf{R}}$ and $0< D_{\mathsf{r}} \leq D_{\mathsf{r},\mathsf{x}}^{\mathsf{R}}$.

\begin{figure}[t]
    \centering
    \includegraphics[width = .45\linewidth]{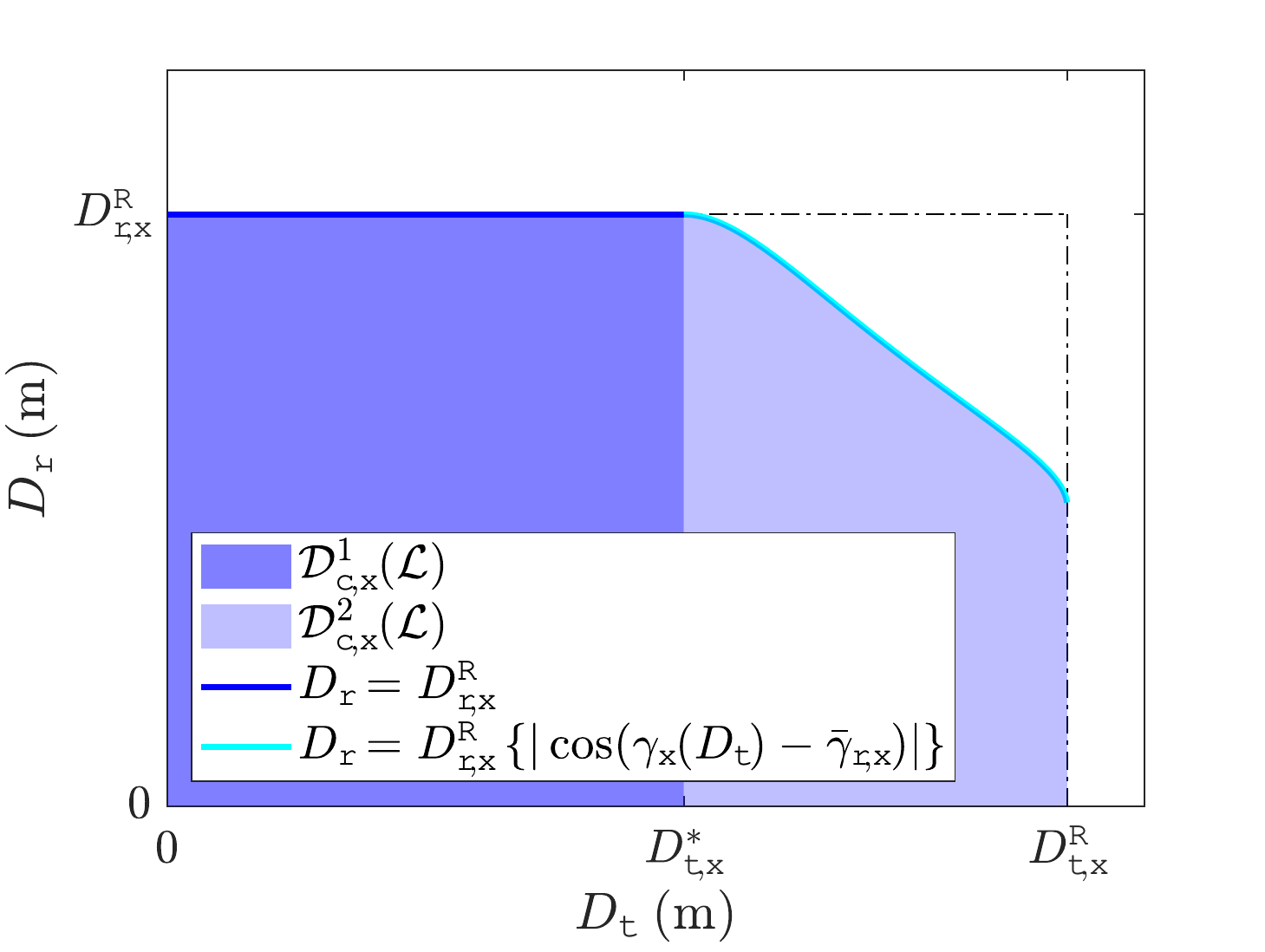}
    \caption{Illustration of a general $\mathcal{D}_{\mathsf{c},\mathsf{x}}(\mathcal{L})$.}
    \label{FMRPROOF}
\end{figure}

To be clear, Fig. \ref{FMRPROOF} illustrates a general $\mathcal{D}_{\mathsf{c},\mathsf{x}}(\mathcal{L})=\mathcal{D}_{\mathsf{c},\mathsf{x}}^1(\mathcal{L})\cup \mathcal{D}_{\mathsf{c},\mathsf{x}}^2(\varphi_{\mathsf{t}},\omega_{\mathsf{t}},\allowbreak \varphi_{\mathsf{r}}, \allowbreak \omega_{\mathsf{r}})$, with $\mathcal{D}_{\mathsf{c},\mathsf{x}}^1(\mathcal{L})$ defined by \eqref{DcxStart} and $\mathcal{D}_{\mathsf{c},\mathsf{x}}^2(\mathcal{L})$ defined by \eqref{Dcx2}. We first prove that when $(D_{\mathsf{t}},D_{\mathsf{r}}) \in \mathcal{D}_{\mathsf{c},\mathsf{x}}^1(\mathcal{L})$, \eqref{CosRatio}-\eqref{CtxrxEqual1b} can be ensured. On the horizontal line $D_{\mathsf{r}}=D_{\mathsf{r},\mathsf{x}}^{\mathsf{R}}$, to meet \eqref{CtxrxEqual1b}, we have $\sin \psi_{\mathsf{r}} = 1$ and $|\cos(\gamma_{\mathsf{r}}-\bar{\gamma}_{\mathsf{r},\mathsf{x}})|=1$, i.e., $\psi_{\mathsf{t}}=\frac{\pi}{2}$ and $\tan \gamma_{\mathsf{r}}= \tan \bar{\gamma}_{\mathsf{r},\mathsf{x}}$. By plugging $|\cos(\gamma_{\mathsf{r}}-\bar{\gamma}_{\mathsf{r},\mathsf{x}})|=1$ into the constraint for $\gamma_{\mathsf{t}}$ and $\gamma_{\mathsf{r}}$ in \eqref{CosRatio}, we have $\tan \gamma_{\mathsf{t}}= \tan \gamma_{\mathsf{t},\mathsf{x}}^*$ with $\gamma_{\mathsf{t},\mathsf{x}}^*$ defined in \eqref{GammatxStarb}.
With $\tan \gamma_{\mathsf{t}}= \tan \gamma_{\mathsf{t},\mathsf{x}}^*$, to meet \eqref{CtxrxEqual1a}, $\psi_{\mathsf{t}}$ is restrained to $\sin \psi_{\mathsf{t}} = \frac{D_{\mathsf{t}}}{D_{\mathsf{t},\mathsf{x}}^*}$ with $D_{\mathsf{t},\mathsf{x}}^* \triangleq D_{\mathsf{t},\mathsf{x}}^{\mathsf{R}}|\cos(\gamma_{\mathsf{t},\mathsf{x}}^*-\bar{\gamma}_{\mathsf{t},\mathsf{x}})|$, and $D_{\mathsf{t}}$ is upper bounded by $D_{\mathsf{t},\mathsf{x}}^*$ for $0\leq \sin \psi_{\mathsf{t}} \leq 1$. Thus for $(D_{\mathsf{t}},D_{\mathsf{r}}) \in \left\{ 0<D_{\mathsf{t}} \leq D_{\mathsf{t},\mathsf{x}}^*,~ D_{\mathsf{r}} = D_{\mathsf{r},\mathsf{x}}^{\mathsf{R}} \right\}$, \eqref{CosRatio}-\eqref{CtxrxEqual1b} can be met by letting $\sin \psi_{\mathsf{t}}= \frac{D_{\mathsf{t}}}{D_{\mathsf{t},\mathsf{x}}^*}$, $\psi_{\mathsf{r}}=\frac{\pi}{2}$, and $\gamma_{\mathsf{t}}$ and $\gamma_{\mathsf{r}}$ meet \eqref{Dcx1Orienta} and \eqref{Dcx1Orienta2}, respectively. Note that \eqref{CosRatio} is irrelevant to $\psi_{\mathsf{t}}$ and $\psi_{\mathsf{r}}$. Therefore, for $(D_{\mathsf{t}},D_{\mathsf{r}}) \in \left\{ 0<D_{\mathsf{t}} \leq D_{\mathsf{t},\mathsf{x}}^*,~ 0 < D_{\mathsf{r}} < D_{\mathsf{r},\mathsf{x}}^{\mathsf{R}} \right\}$, \eqref{CosRatio} can be met still by letting $\gamma_{\mathsf{t}}$ and $\gamma_{\mathsf{r}}$ satisfy \eqref{Dcx1Orienta} and \eqref{Dcx1Orienta2}, respectively. Then, \eqref{CtxrxEqual1a} can be met still by letting $\sin\psi_{\mathsf{t}} = \frac{D_{\mathsf{t}}}{D_{\mathsf{t},\mathsf{x}}^*}$, and \eqref{CtxrxEqual1b} can be met by letting $\sin\psi_{\mathsf{r}} = \frac{D_{\mathsf{r}}}{D_{\mathsf{r},\mathsf{x}}^{\mathsf{R}}}$. Based on the above, In $\mathcal{D}_{\mathsf{c},\mathsf{x}}^1(\mathcal{L})$, \eqref{CosRatio}-\eqref{CtxrxEqual1b} can be met by letting $\gamma_{\mathsf{t}}$, $\gamma_{\mathsf{r}}$, $\psi_{\mathsf{t}}$, and $\psi_{\mathsf{r}}$ meet \eqref{Dcx1Orienta}-\eqref{Dcx1Orientb}.

Next we prove that when $(D_{\mathsf{t}},D_{\mathsf{r}}) \in \mathcal{D}_{\mathsf{c},\mathsf{x}}^2(\mathcal{L})$, \eqref{CosRatio}-\eqref{CtxrxEqual1b} can be ensured by letting $\gamma_{\mathsf{t}}$, $\gamma_{\mathsf{r}}$, $\psi_{\mathsf{t}}$, and $\psi_{\mathsf{r}}$ satisfy \eqref{Dcx2Orienta}-\eqref{Dcx2Orientd}. 
Note that \eqref{Dcx2Orientd} is valid only for $0<D_{\mathsf{r}} \leq D_{\mathsf{r},\mathsf{x}}^{\mathsf{R}}|\cos(\gamma_{\mathsf{x}}(D_{\mathsf{t}})-\bar{\gamma}_{\mathsf{r},\mathsf{x}})|$. 
By letting $\psi_{\mathsf{t}}=\frac{\pi}{2}$ and $\gamma_{\mathsf{t}}$ meet \eqref{Dcx2Orienta}, we have $|\cos(\gamma_{\mathsf{t}}-\bar{\gamma}_{\mathsf{t},\mathsf{x}})|=\frac{D_{\mathsf{t}}}{D_{\mathsf{t},\mathsf{x}}^{\mathsf{R}}}$ and $\sin \psi_{\mathsf{t}}=1$, and thus \eqref{CtxrxEqual1a} is met. By letting $\sin \psi_{\mathsf{r}} = \frac{D_{\mathsf{r}}}{D_{\mathsf{r},\mathsf{x}}^{\mathsf{R}}|\cos(\gamma_{\mathsf{x}}(D_{\mathsf{t}})-\bar{\gamma}_{\mathsf{r},\mathsf{x}})|}$ and $\gamma_{\mathsf{r}}$ meet \eqref{Dcx2Orientb}, we have $\sin \psi_{\mathsf{r}}|\cos(\gamma_{\mathsf{r}}-\bar{\gamma}_{\mathsf{r},\mathsf{x}})| = \frac{D_{\mathsf{r}}}{D_{\mathsf{r},\mathsf{x}}^{\mathsf{R}}}$, and thus \eqref{CtxrxEqual1b} is met. By defining 
\begin{equation}\label{Gx}
    G_{\mathsf{x}}(\gamma_{\mathsf{t}}) \triangleq  \frac{A_{\mathsf{t},\mathsf{y}}A_{\mathsf{r},\mathsf{x}}\left( \cos(\bar{\gamma}_{\mathsf{t},\mathsf{x}}-\bar{\gamma}_{\mathsf{t},\mathsf{y}}) - \sin(\bar{\gamma}_{\mathsf{t},\mathsf{x}}-\bar{\gamma}_{\mathsf{t},\mathsf{y}}) \tan(\gamma_{\mathsf{t}}-\bar{\gamma}_{\mathsf{t},\mathsf{x}}) \right)}{A_{\mathsf{t},\mathsf{x}}A_{\mathsf{r},\mathsf{y}}},
\end{equation}
\eqref{CosRatio} can be rewritten as 
\begin{equation}\label{tanG}
    \tan \gamma_{\mathsf{r}} = \frac{\cos \bar{\gamma}_{\mathsf{r},\mathsf{y}}-G_{\mathsf{x}}(\gamma_{\mathsf{t}})\cos \bar{\gamma}_{\mathsf{r},\mathsf{x}}}{G_{\mathsf{x}}(\gamma_{\mathsf{t}})\sin \bar{\gamma}_{\mathsf{r},\mathsf{x}}-\sin \bar{\gamma}_{\mathsf{r},\mathsf{y}}}.
\end{equation}
Then, by plugging \eqref{gammaxPstva}, \eqref{Dcx2Orienta}, \eqref{Dcx2Orientb}, and \eqref{Gx} into \eqref{tanG}, we see that the left-hand side equals the right-hand side. Thus, when $(D_{\mathsf{t}},D_{\mathsf{r}}) \in \left\{D_{\mathsf{t},\mathsf{x}}^* < D_{\mathsf{t}} \leq D_{\mathsf{t},\mathsf{x}}^{\mathsf{R}},0<D_{\mathsf{r}}\leq D_{\mathsf{r},\mathsf{x}}^{\mathsf{R}}|\cos(\gamma_{\mathsf{x}}(D_{\mathsf{t}})-\bar{\gamma}_{\mathsf{r},\mathsf{x}})|\right\}$, \eqref{CosRatio}-\eqref{CtxrxEqual1b} can always be met by letting $\gamma_{\mathsf{t}}$, $\gamma_{\mathsf{r}}$, $\psi_{\mathsf{t}}$, and $\psi_{\mathsf{r}}$ satisfy \eqref{Dcx2Orienta}-\eqref{Dcx2Orientd}. 

Based on the above, when $(D_{\mathsf{t}},D_{\mathsf{r}}) \in \mathcal{D}_{\mathsf{c},\mathsf{x}}(\mathcal{L})$, the orthogonality requirements in \eqref{NtgeqNr} can always be met by letting $\gamma_{\mathsf{t}}$, $\gamma_{\mathsf{r}}$, $\psi_{\mathsf{t}}$, and $\psi_{\mathsf{r}}$ satisfy \eqref{Dcx1Orienta}-\eqref{Dcx1Orientb} or \eqref{Dcx2Orienta}-\eqref{Dcx2Orientd}. Thus, the FMR $\mathcal{D}_{\mathsf{c}}(\mathcal{L})$ covers $\mathcal{D}_{\mathsf{c},\mathsf{x}}(\mathcal{L})$. It can be proved in a similar way that $\mathcal{D}_{\mathsf{c}}(\mathcal{L})$ also covers $\mathcal{D}_{\mathsf{c},\mathsf{y}}(\mathcal{L})$, which is defined by swapping the subscripts $\mathsf{x}$ and $\mathsf{y}$ in the definition of $\mathcal{D}_{\mathsf{c},\mathsf{x}}(\mathcal{L})$. Therefore, when $N_{\mathsf{t}}+N_{\mathsf{r}}-2<2Q_{\mathsf{x}}$ and $N_{\mathsf{t}}+N_{\mathsf{r}}-2<2Q_{\mathsf{y}}$, the FMR $\mathcal{D}_{\mathsf{c}}(\mathcal{L})$ at least covers the union of $\mathcal{D}_{\mathsf{c},\mathsf{x}}(\mathcal{L})$ and $\mathcal{D}_{\mathsf{c},\mathsf{y}}(\mathcal{L})$, which concludes the proof.

\ifCLASSOPTIONcaptionsoff
  \newpage
\fi

\bibliographystyle{IEEEtran}
\bibliography{IEEEabrv,IEEEexample.bib}





\end{document}